%% file: zaibatsu_main.tex
\documentclass[12pt,a4paper,notitlepage,fleqn]{article}

\usepackage{setspace}
\usepackage[margin=1in]{geometry}
\usepackage[symbol]{footmisc}
\usepackage[T1]{fontenc}
\usepackage[utf8]{inputenc}
\usepackage{authblk}



\title{Wartime Controls, Political Connections, and the Pricing of Zaibatsu Rents in Japan, 1930--1943}%

\author{Keiichi Morimoto$^{a}$ \quad Akihiko Noda$^{b}$\thanks{\scriptsize Corresponding Author. E-mail: anoda@meiji.ac.jp, Tel/Fax. +81-3-3296-2265.} \quad Takenobu Yuki$^{c,d}$

{\scriptsize ${}^{a}$ \it School of Political Science and Economics, Meiji University, 1-1 Kanda-Surugadai, Chiyoda-ku, Tokyo 101-8301, Japan} 

{\scriptsize ${}^{b}$ \it School of Commerce, Meiji University, 1-1 Kanda-Surugadai, Chiyoda-ku, Tokyo 101-8301, Japan} 

{\scriptsize ${}^{c}$ \it School of Commerce, Waseda University, 1-6-1 Nishi-Waseda, Shinjuku-ku, Tokyo 169-0051, Japan}

{\scriptsize ${}^{d}$ \it Graduate School of Economics \& Management, Tohoku University, 27-1 Kawauchi, Aoba-ku, Sendai, Miyagi 980-8576, Japan}}

\date{This Version: \today}

\renewcommand\thefootnote{\arabic{footnote}}

\usepackage{graphicx}

\usepackage[sort]{natbib}%
\usepackage{amsmath,amssymb,amsthm}%
\usepackage{ascmac}%
\usepackage{multirow}%
\usepackage{pdflscape}%
\usepackage{subfigmat}

\usepackage{pifont}%
\usepackage{arydshln}%
\usepackage[format=hang]{caption}
\usepackage[all]{xy}
\usepackage{url}

\PassOptionsToPackage{hyphens}{url}
\usepackage{color}
\usepackage[
            bookmarkstype=toc,
	    colorlinks=true, 
	    pdfborder={0 0 1}, 
	    bookmarks=true, 
	    bookmarksnumbered=true, 
	    citecolor=blue]{hyperref}
\bibpunct{(}{)}{;}{a}{}{,}
\usepackage{here}
\usepackage{booktabs}

\usepackage{mathrsfs}
\usepackage{mathtools}

\usepackage{tikz}
\usetikzlibrary{positioning}

\bibpunct{(}{)}{;}{a}{}{,}

\def\hsymbu#1{\smash{\lower1.7ex\hbox{\huge$#1$}}}

\def\ve #1{{\mbox{\boldmath $#1$}}}

\newtheoremstyle{nopunct}
  {3pt}{3pt}
  {\itshape}
  {}
  {\bfseries}
  {.}
  {0.5em}
  {\thmname{#1}\thmnumber{ #2}\ifx&#3&\else\ \textnormal{(#3)}\fi}

\theoremstyle{nopunct}

\newtheorem{proposition}{Proposition}

\newcommand{\citetapos}[1]{\citeauthor{#1}'s \citeyearpar{#1}}

\newcommand{\ex}{{\mathbb{E}}}
\newcommand{\var}{{\mathbb{V}}{\rm{ar}}}
\newcommand{\cov}{{\mathbb{C}}{\rm{ov}}}

\def\ve #1{{\mbox{\boldmath $#1$}}}

\begin{document}

\begin{titlepage}

\renewcommand{\thepage}{}
\renewcommand{\thefootnote}{\fnsymbol{footnote}}

\maketitle

\vspace{-10mm}

\noindent
\hrulefill

\noindent
{\bfseries Abstract:} This paper examines how wartime economic controls shaped stock-price formation in Japan from 1930 to 1943. We develop a four-portfolio asset-pricing model in which zaibatsu affiliation affects expected payoffs and the translation of valuations into economic scale through lower financing wedges. We then construct daily capitalization-weighted indices and four benchmark portfolios based on a $2\times 2$ sort by zaibatsu affiliation and military orientation. Using a CAPM-AR($p$)-SV event-study framework that allows for serial correlation and stochastic volatility, we show that the model rationalizes capitalization concentration, segmented abnormal returns, delayed cumulative adjustment, regime-risk insulation of zaibatsu portfolios, and zaibatsu-concentrated responses to embedded-rent or group-continuation shocks. The evidence is consistent not with a collapse of semi-strong efficiency, but with institutionally contingent efficiency: stock prices continued to respond to news while capitalizing uneven access to credit, materials, and procurement.\\

\noindent
{\bfseries Keywords:} Zaibatsu; Political Connection; Wartime Control; Event Study; Stochastic Volatility.\\

\noindent
{\bfseries JEL Classification Numbers:} C22; G12; G14; N25.

\noindent
\hrulefill

\end{titlepage}

\bibliographystyle{asa}


\input{zaibatsu_intro}

\input{zaibatsu_hist}

\input{zaibatsu_model}

\input{zaibatsu_frame}

\input{zaibatsu_data}

\input{zaibatsu_empirical}

\input{zaibatsu_conclusion}

\clearpage

\input{zaibatsu_ack}

\input{zaibatsu_main.bbl}

\clearpage

\input{zaibatsu_table}

\end{document}

%% file: zaibatsu_intro.tex
\section{Introduction}\label{zaibatsu_sec1}

In most modern economies, equity markets are understood to be central institutions in long-run economic development because they mobilize savings, facilitate risk sharing, and impose governance discipline on firms and managers. Classic contributions such as \citet{goldsmith1969fsd} and \citet{mckinnon1973mce}, building on \citetapos{schumpeter1911ted} insight that finance and innovation are jointly implicated in economic transformation, emphasize that financial development lowers the cost of capital, relaxes financing constraints, and thereby supports sustained growth. In this view, stock prices do not merely record ex post assessments of firm value. They also summarize dispersed information about future profitability and macroeconomic conditions, thereby helping to direct capital toward more productive uses. More broadly, securities markets matter not only because they finance investment, but also because they generate signals that shape the allocation of resources across firms and sectors.

Yet this finance-and-growth narrative is not institutionally universal. Under wartime mobilization or, more generally, under state-capitalist arrangements, rationing, capital controls, directed credit, procurement priorities, and administrative allocation become binding features of the economic environment, and the informational and allocative roles of equity prices are no longer straightforward. As \citet{musacchio2015nvs} emphasizes, the coexistence of private ownership and extensive state intervention gives rise to settings in which market signals and administrative commands interact in complex ways. More broadly, political-economy research has stressed that the development and functioning of financial markets depend not only on economic fundamentals but also on legal institutions, political coalitions, and distributional conflict (\citet{bordo2012hef,rajan2003tgr}). Once price formation is conditioned by policy rules and political hierarchy, it becomes unclear whether observed market valuations continue to reflect fundamentals in a manner comparable to that in liberal peacetime settings or instead embody a distinct logic of valuation shaped by state priorities and unequal access to policy-induced advantages. This raises a fundamental question: to what extent can equity prices continue to aggregate economically meaningful information when governments intervene pervasively in the allocation of finance and resources?

Japan's wartime economic system offers an especially revealing setting in which to study this question. It constitutes a rare historical ``laboratory'' in which an active stock exchange continued to operate even as state controls over the economy became progressively more stringent. Although many belligerent economies moved toward highly centralized command systems during the 1930s and 1940s (\citet{hara1995nsk}), Japan did not fully abolish market-based allocation mechanisms. Profit incentives remained operative, corporate organization continued to matter, and a formal stock exchange persisted throughout much of the 1930--43 period (\citet{okazaki1994jwe}). At the same time, the state deepened its control over credit, foreign exchange, and strategic materials, culminating in legislation such as the Temporary Funds Adjustment Law, under which industries were classified into priority categories and finance was increasingly directed in accordance with wartime objectives. The resulting regime was therefore neither a liberal market economy nor a pure command economy, but rather a hybrid system in which administrative guidance and market pricing coexisted. For investors, this meant that the valuation of firms depended not only on conventional market fundamentals but also on how state policy shaped access to credit, inputs, and procurement. Japan thus provides a particularly suitable setting for examining whether, and how, equity prices continued to aggregate information about risks and rents under extensive wartime intervention.

Within this hybrid regime, the large family-based business groups known as zaibatsu occupied a particularly consequential position. Their networks of banks, trading companies, and industrial affiliates supported internal capital markets and coordination across sectors. Evidence for Japan indicates that group finance attenuated investment--liquidity sensitivity (\citet{hoshi1991csl}), while theory clarifies the conditions under which internal reallocation can substitute for distorted or incomplete external finance (\citet{almeida2006sbg}). Yet, in wartime Japan, zaibatsu were not simply diversified conglomerates. Because business-group organization was coupled with political access, internal finance could interact with procurement privileges, regulatory protection, and preferential access to scarce resources, thereby generating group-specific rents and incentives for rent seeking (\citet{fisman2001evp}). The major zaibatsu were therefore natural intermediaries of mobilization, able to shift resources across affiliates and to exploit informational advantages arising from their proximity to wartime policymaking.

From a corporate-finance perspective, zaibatsu affiliation therefore had a dual meaning. On the one hand, it denoted an organizational characteristic, namely access to internal capital markets, centralized monitoring, and group-governance structures that shaped firms' financing conditions, operating policies, and risk exposures. On the other hand, under the Temporary Funds Adjustment Law and related wartime controls, it also served as a proxy for privileged political access and priority status in the allocation of credit, materials, and munitions contracts. Historical work emphasizes that the state both relied on and protected the major zaibatsu as key agents of mobilization, extending preferential access to finance, inputs, and procurement while leaving non-zaibatsu firms relatively disadvantaged (\citet{morikawa1970osm,yamamura1964zpz}). Evidence from prewar Japan further suggests that political ties were capitalized in stock prices: firms that established new political connections through elections subsequently experienced significantly higher stock returns (\citet{okazaki2017mei}). Under wartime controls, then, group affiliation may have been priced through more than one channel, not only through compensation for differential exposure to risk, but also through the capitalization of expected rents arising from protected access to strategically allocated resources.

This observation points to a gap in the literature. Studies of wartime or heavily controlled financial markets have generally focused on whether prices responded to military news, political shocks, or policy announcements. By contrast, the literature on business groups has focused on how ownership, control, internal finance, and governance shape firms' cash flows, investment, and valuation. Yet these literatures have, for the most part, developed in parallel. The former rarely differentiates firms according to their access to internal capital markets or political privilege, while the latter seldom examines environments in which state allocation and mobilization are themselves first-order determinants of corporate prospects. Consequently, we still know relatively little about whether and how markets operating under heavy intervention priced the distinctive risks and rents associated with powerful business groups.

This paper addresses that question in the context of wartime Japan. We ask whether Japan's wartime stock market satisfied a form of semi-strong efficiency once heterogeneity in group affiliation and political access is taken into account, whether zaibatsu affiliation carried a wartime return premium or discount beyond standard risk factors, consistent with either compensation for group-specific risk or the capitalization of political rents, and how major wartime policy interventions altered these pricing relations. Our central claim is that wartime markets need not be characterized by a simple collapse of efficiency. Rather, prices may continue to respond systematically to news while \emph{efficiently pricing prevailing institutional distortions} created by policy constraints, mobilization priorities, and political hierarchy.

To organize this argument, we combine a simple theoretical model with a high-frequency empirical design. On the theoretical side, we develop a four-portfolio asset-pricing model in which zaibatsu affiliation affects both expected payoffs and the translation of valuations into economic scale through lower financing wedges. The model is designed to rationalize both three core patterns and two additional sources of event-study heterogeneity. The core patterns are the divergence between relative price performance and market-capitalization dominance, segmented sign patterns in event-time abnormal returns across the four portfolios, and significant cumulative abnormal returns even when day-0 abnormal returns are weak. The additional patterns are cases in which broad regime-risk shocks move non-zaibatsu portfolios while zaibatsu portfolios remain muted, and cases in which cumulative responses are concentrated in zaibatsu portfolios because events revise state-embedded rents or group-level continuation values.

Our empirical framework is based on a capital asset pricing model with $p$-th-order autoregressive residuals and stochastic volatility, or CAPM-AR($p$)-SV. This specification is designed to relax assumptions embedded in conventional constant-parameter event studies that are especially fragile in wartime daily data, particularly homoskedastic innovations and serially independent residuals. In our setting, stochastic volatility captures volatility clustering and thick-tailed uncertainty, while AR($p$) residuals absorb persistence arising from market frictions and non-synchronous trading. Combined with portfolio heterogeneity and event windows tied to institutional interventions, this design permits an assessment of whether security returns during 1930--43 incorporated both the additional risks and the expected rents associated with zaibatsu membership in a manner consistent with institutionally contingent efficiency.

The paper makes four related contributions. First, it contributes to the literature on wartime market efficiency by showing that inferences of ``inefficiency'' derived from constant-parameter linear models may be highly sensitive to assumptions that are unlikely to hold in historical wartime daily data, including parameter stability, homoskedasticity, and serial independence of disturbances. Second, it contributes to the business-group literature by distinguishing two valuation channels that are often conflated: internal capital markets and group governance, on the one hand, and politically mediated access to priority credit, materials, and procurement, on the other. Third, it contributes theoretically by developing a four-portfolio asset-pricing model tailored to wartime Japan. The model links zaibatsu affiliation to expected payoffs, financing wedges, and the translation of valuations into economic scale, while also decomposing wartime policy news into regime-risk and embedded-rent components. This structure rationalizes not only capitalization concentration, segmented event-time responses, and delayed cumulative adjustment, but also muted zaibatsu reactions to broad regime shocks and zaibatsu-concentrated cumulative responses to events that revise state-embedded rents or group-level continuation values. Fourth, it contributes methodologically by combining this theoretical structure with an event-study design explicitly tailored to historical daily data in which volatility is time-varying and residual dependence is a first-order feature.

The remainder of the paper is organized as follows. Section \ref{zaibatsu_sec2} situates our analysis historically by describing Japan's interwar market microstructure, the wartime control regime, and the role of zaibatsu networks in finance, procurement, and information transmission. Section~\ref{zaibatsu_sec3} presents the model and derives its implications for capitalization dynamics, segmented event-time responses, delayed cumulative adjustment, regime-risk insulation, and zaibatsu-concentrated responses to embedded-rent shocks. Section \ref{zaibatsu_sec4} introduces the CAPM-AR($p$)-SV event-study framework and presents the abnormal-return measures used to assess day-0 and cumulative event responses. Section \ref{zaibatsu_sec5} constructs daily capitalization-weighted market indices and forms four benchmark portfolios based on a $2 \times 2$ sort by zaibatsu affiliation and military orientation. Section \ref{zaibatsu_sec6} reports the empirical results. Section \ref{zaibatsu_sec7} concludes and considers the broader implications of controlled markets for the reproduction of institutional asymmetries.

\subsection*{Related Literature}

A substantial literature in finance and economic history argues that the development of securities markets is closely associated with long-run economic growth. Early cross-country studies, including \citet{atje1993smd}, \citet{king1993fgs}, and \citet{levine1996smd,levine1998smb}, show that measures of market depth and liquidity, such as capitalization, turnover, and the breadth of tradable claims, are positively associated with subsequent per capita income growth even after controlling for banking-sector development. The mechanism emphasized in this literature is not merely that financial markets expand the volume of available funds, but that they improve the efficiency with which capital is allocated across firms and sectors. In this spirit, \citet{beck2000fsg} and \citet{wurgler2000fma} argue that developed financial systems enhance productivity by directing resources toward higher-return uses. Historical research broadly supports this view in relatively liberal settings. \citet{neal1991rfc} and \citet{goetzmann2001nhd} document how securities markets in Europe and North America financed canals, railways, and heavy industry while widening opportunities for diversification and risk sharing, and \citet{rousseau2003fse,rousseau2005efm}, together with \citet{schularick2010fii}, emphasize the importance of securities-market deepening in the emergence of modern growth regimes. Because most of this evidence is drawn from settings in which governments only partially displaced market allocation, however, it provides incomplete guidance for understanding how financial markets function when states intervene extensively in the distribution of finance, materials, and investment opportunities.

This limitation has motivated a growing literature on war, intervention, and financial markets. One strand examines whether asset prices continued to aggregate information under extreme political and military stress. Event studies of U.S.\ Civil War-era greenback and gold markets and of British securities show that prices reacted promptly to news about military fortunes and political developments even under severe uncertainty (\citet{brown2000tpu,willard1996tpc}). Related work on Nazi Germany, Vichy France, and the United States during World War I\hspace{-1.2pt}I indicates that politically connected firms earned excess returns and that bond and equity prices incorporated changing beliefs about regime durability, military prospects, and the distribution of wartime advantage (\citet{choudhry2010ww2,ferguson2008bhv,oosterlinck2003bml}). These studies suggest that wartime markets need not cease to be informative. Rather, what is being priced in such settings is often inseparable from political power, state strategy, and the expected consequences of military events.

A second strand asks more directly whether asset prices remained efficient in the conventional time-series sense under war or heavy intervention. For the United States, weak-form tests generally find that return dynamics are broadly consistent with random-walk behavior, albeit with temporary deterioration around major crises (\citet{ito2016esm,kim2011srp}). For Japan, by contrast, the evidence is more divided. Using a 1930s sample and standard linear time-series methods, \citet{bassino2015iet} argue that prewar Tokyo stock prices departed from weak-form efficiency. By contrast, using a longer 1924--1943 sample together with time-varying econometric methods, \citet{hirayama2025mtv} conclude that although efficiency deteriorated episodically under the increasingly militarized economy of the late 1930s, return dynamics on the Tokyo Stock Exchange were, on average, broadly consistent with informational efficiency. Semi-strong evidence points in a similar direction. \citet{kataoka2004b} and \citet{suzuki2022gke} show that prewar Japanese equity prices responded systematically to firm-specific and natural-disaster news, while evidence from Japanese rice-futures markets reported by \citet{ito2016meg,ito2018fpe} suggests that the degree of efficiency depended closely on the extent and design of government intervention. Taken together, this body of work suggests that market efficiency in controlled economies is best understood as institutionally contingent: prices may react to news even when the mapping from fundamentals to valuations is mediated by intervention, regulation, and political hierarchy.

Running alongside these literatures is a distinct body of research on business groups, internal capital markets, and political connections. A central concern in this literature is how group organization alters firms' financing conditions, investment behavior, governance, and valuation. For Japan, \citet{hoshi1991csl} show that group finance attenuated investment--liquidity sensitivity, indicating that affiliation could soften external financing constraints. More generally, \citet{almeida2006sbg} clarify the theoretical conditions under which internal capital markets may substitute for distorted or incomplete external finance. Yet the corporate-finance implications of business groups are not exhausted by the mitigation of frictions. When group organization is combined with concentrated control and political access, it may facilitate not only efficient internal reallocation but also the extraction and capitalization of political rents, as well as minority-shareholder expropriation through intra-group transactions (\citet{baek2006bgt,claessens2002die,fisman2001evp}). For the prewar Japanese case, \citet{okazaki2001rhc} show that zaibatsu-affiliated firms achieved higher profitability than comparable non-zaibatsu firms and that holding-company structures strengthened centralized monitoring and restructuring capacity. Historical scholarship further emphasizes that the wartime Japanese state both relied upon and protected the major zaibatsu as agents of mobilization, extending preferential access to finance, materials, and procurement while leaving non-zaibatsu firms at a relative disadvantage (\citet{morikawa1970osm,yamamura1964zpz}). Evidence from prewar Japan also indicates that political connections were capitalized in stock prices: firms that acquired new political ties through elections subsequently experienced significantly higher returns (\citet{okazaki2017mei}).

Recent work in political science also reinforces the broader view that wartime intervention generated highly uneven consequences across sectors and elites. Analyzing Japanese legislators between 1936 and 1942, \citet{fukumoto2026cms} shows that legislators connected to sanction-exposed sectors became more likely to align with military-backed authoritarian policies after the sanction shock, whereas those tied to procurement-linked sectors did not exhibit a comparable shift. Although the outcome of interest is legislative alignment rather than asset pricing, this evidence is complementary to the present analysis because it indicates that wartime controls and external shocks redistributed economic vulnerability, political leverage, and access to state favor in systematically heterogeneous ways rather than uniformly across the economy.

On the theoretical side, our framework is also related to asset-pricing models in which prices reflect not only fundamentals but also higher-order beliefs about how other traders react to public signals. In particular, it is in the spirit of beauty-contest models such as \citet{allen2006bci}, but applied to a historically specific environment in which wartime demand, administrative controls, internal capital markets, and staggered processing of public news jointly shape asset prices. Our objective is not to develop a general theory of information aggregation, but rather to build a tractable structure that matches the empirical four-portfolio design and the event-study patterns documented below.

Despite their evident conceptual proximity, these strands of research have rarely been brought together in a unified empirical framework. Studies of wartime or heavily controlled markets have generally asked whether asset prices responded to military news, policy announcements, or macroeconomic disturbances, but they have rarely distinguished firms according to their access to internal capital markets or politically mediated privilege. Conversely, the business-group literature has emphasized how ownership, control, and internal finance shape corporate outcomes, but has seldom examined settings in which state allocation and wartime mobilization are themselves first-order determinants of firm performance. The same is true of adjacent political-economy work, which highlights the uneven political consequences of sanctions and procurement but does not examine how such asymmetries were capitalized in security prices. The result is a gap in our understanding of how markets under heavy intervention priced the risks and rents associated with powerful business groups.

This paper is positioned at precisely that intersection. Relative to the literature on wartime market efficiency, it introduces heterogeneity in group affiliation and political access directly into the empirical design. Relative to the literature on business groups, it examines an environment in which state allocation, administrative control, and mobilization priorities are central determinants of firm-level prospects. Relative to adjacent political-economy research, it shifts attention from legislative alignment to market valuation and abnormal returns. Relative to the theoretical asset-pricing literature, it develops a wartime four-portfolio framework tailored to the institutional structure of prewar and wartime Japan. Wartime Japan therefore provides a setting in which zaibatsu affiliation can be interpreted simultaneously as an organizational characteristic and as a proxy for privileged access to scarce, state-controlled resources. By analyzing how these features were capitalized in equity returns during 1930--43, the paper seeks to clarify how internal finance, political access, wartime intervention, and the diffusion of public information jointly shaped the operation of Japanese equity markets. The model below formalizes this intersection by distinguishing baseline wartime shocks from regime-risk and embedded-rent shocks, allowing the same market to generate non-zaibatsu-dominated responses to broad political, diplomatic, or monetary shocks and zaibatsu-concentrated responses to revisions in state-embedded rents.

%% file: zaibatsu_hist.tex
\section{Historical Background}\label{zaibatsu_sec2}

To rigorously evaluate the informational efficiency of the Japanese stock market during World War I\hspace{-1.2pt}I, we must first dismantle the historical misconception that the ``wartime economy'' was a command economy in which market mechanisms were entirely suspended. If capital allocation had been determined solely by bureaucratic fiat, asset pricing models such as the CAPM-AR($p$)-SV would be theoretically misspecified, and stock prices would function merely as administrative shadow prices rather than as signals of future cash flows. However, the historical reality was far more complex and amenable to econometric analysis. Between 1930 and 1943, Japan's political economy operated as a hybrid system, that is, a controlled market economy. In this regime, extensive state intervention via capital adjustment laws coexisted with profit-maximizing corporate behavior and active trading in the secondary market. This hybridity introduced specific informational frictions into the market mechanism.

In a frictionless market, stock prices aggregate dispersed information regarding firm fundamentals. In the context of wartime Japan, however, ``fundamentals'' were determined not only by consumer demand or technological advantage, but increasingly by state policy, specifically raw material quotas, credit prioritization, and military procurement contracts. Consequently, the nature of the information that the market had to aggregate shifted from the commercial to the political. A firm's revenue-generating capacity became, to a significant extent, a function of its regulatory classification and its political connectivity, rather than of its market competitiveness alone. This section reconstructs the institutional architecture of the period to provide the theoretical micro-foundations for the empirical analysis that follows. We argue that the observed heterogeneity in abnormal returns, across zaibatsu and non-zaibatsu firms and across military and non-military firms, was not a random anomaly, but reflected the market's differential pricing of institutional privilege, policy exposure, and organizational capacity. By detailing the sophistication of the prewar market, the mechanics of the 1937 capital controls, and the internal capital markets of the zaibatsu, we show that the wartime stock market did not simply malfunction; rather, it continued to process the rents and constraints generated by the military-industrial complex.

\subsection{Historical Development: Market Microstructure and Sophistication}

The validity of testing the Efficient Market Hypothesis (EMH) in 1930s Japan depends on whether the market possessed sufficient depth and liquidity to facilitate arbitrage. Contrary to the view of prewar Japan as a financial backwater, quantitative evidence indicates that by the mid-1930s the Tokyo Stock Exchange (TSE) and Osaka Stock Exchange (OSE) were among the most sophisticated in the world, often surpassing the U.S.\ market in depth and turnover. By 1936, the Japanese stock market had achieved a scale that challenges the stereotype of a bank-centered financial system. \citet{okazaki2005gdp} show that prewar Japan compared favorably even with the United States on standard indicators of stock-market development, and that before the outbreak of the Sino-Japanese War total stock trading value across Japanese exchanges amounted to roughly 2.0 to 2.5 times GNP. Such liquidity is a prerequisite for semi-strong efficiency, implying that the market possessed the mechanical capacity to incorporate new information into prices rapidly. The structural pillar of this liquidity was the short-term clearing transaction system (tanki seisan torihiki), institutionalized after the 1922 Exchange Act reform (\citet{kobayashi2012hjs}). Functioning in some respects like a modern futures market, this system did not require immediate physical delivery of certificates. In the Tokyo Stock Exchange's short-term clearing market, settlement could be deferred within one month and, through offsetting cross-transactions, positions could in practice be rolled over for longer periods. \citet{hirayama2022smf} shows that, from June 1924 to August 1943, the monthly average non-deliverable settlement ratio in short-term clearing transactions was 82.0\%, indicating that most trades were settled by differences rather than by physical delivery. This microstructure had two important implications for information processing.

First, by allowing speculation on price movements without the capital drag of physical delivery, the system amplified the market's sensitivity to news. Speculators could expand positions quickly in response to military reports or policy shifts, facilitating faster price discovery than would have been possible in a pure spot market. Second, the equity market was deeply intertwined with the banking sector. Loans secured by shares were common, constituting roughly 40\% of ordinary bank lending by the late Taisho period. This linkage created a feedback loop in which stock market volatility immediately affected bank balance sheets and credit availability.

The sophistication of the market was also evident in listing behavior. Regression analyses suggest that firms rationally selected listing on the TSE based on cost-benefit calculations related to firm size and maturity. The ``home bias'' observed in the Meiji era, whereby eastern firms were more likely to list in Tokyo and western firms in Osaka, had weakened substantially by the 1930s, especially after the 1918 reform of the Tokyo Stock Exchange, indicating the growing national integration of the market (\citet{okazaki2005gdp}). Therefore, the period we analyze (1930--1943) represents not the chaos of an immature market, but the stress testing of a highly liquid and sophisticated financial system under exogenous state intervention. The structural breaks observed after 1937 were not failures of market capacity, but the forced rewriting of the market's objective function by the state. The market remained an active information-processing mechanism, but the signals it processed increasingly reflected administratively imposed priorities alongside conventional market fundamentals.

\subsection{Wartime Financial Controls: The Bureaucratization of Price Signals}

The transition to a wartime economy was defined by the enactment of the Temporary Funds Adjustment Act (Rinji Shikin Chosei Ho) in September 1937 (\citet{mof1957hfm}). This legislation serves as the institutional pivot of our analysis. From the perspective of information economics, the Act substituted market price signals, such as interest rates and expected returns, with bureaucratic quantity signals, namely allocation priority. The Act mandated government approval for all major corporate financing activities, including stock issuance, bond flotation, and long-term borrowing. To administer this rationing, the Ministry of Finance and the Bank of Japan established the Standard for Adjustment of Industrial Funds, which classified industries into three tiers based on strategic utility, as shown in Table \ref{age_war_tab1}:
\begin{center}
(Table \ref{age_war_tab1} around here)
\end{center}
Category A ({\it{Ko}}-shu) comprised strategic war industries, including metal mining, heavy machinery, aircraft, and shipbuilding. Their regulatory status was approval in principle.'' These firms enjoyed prioritized access to capital. New share issuance was streamlined, and debt was often underwritten by the Industrial Bank of Japan. The economic implication was clear: for Category A firms, the cost of capital was artificially suppressed and funding volume was effectively guaranteed. Insolvency risk was, to a considerable extent, socialized by the state. Category B ({\it{Otsu}}-shu) comprised intermediate industries essential for civilian subsistence or indirect military support, such as textiles for uniforms and agriculture. Their regulatory status was case-by-case screening.'' Financing depended on project-specific approval, which introduced substantial bureaucratic friction and uncertainty. Category C ({\it{Hei}}-shu) comprised non-essential civilian industries, including luxury goods, entertainment, and services. Their regulatory status was ``refusal in principle.'' These firms were effectively excluded from official capital markets. Capital for capacity expansion was prohibited, and their future cash flows were structurally capped by severe financing constraints.

Standard EMH assumes that capital flows to firms with the highest risk-adjusted returns. The Temporary Funds Adjustment Act severed this link. A highly profitable textile firm in a low-priority sector could be unable to raise capital, whereas an inefficient munitions plant in a high-priority sector could receive abundant low-cost funding. This created a structural asymmetry in information processing. For Category A firms, rising stock prices reflected not only profitability but also state endorsement. Military campaigns translated more directly into sales growth when the state guaranteed the financing needed to meet demand. For low-priority civilian firms, by contrast, positive macroeconomic news could paradoxically become negative news, insofar as it implied prolonged controls and a further squeeze on the civilian sector. Recent evidence on legislators’ corporate ties in late interwar Japan points to the same kind of asymmetry: sectors exposed to sanctions and sectors linked to wartime procurement did not respond in a uniform way to external and state-induced shocks, underscoring that wartime intervention did not create a single, homogeneous hierarchy of winners and losers even among politically salient industries (\citet{fukumoto2026cms}).

Market mechanisms were further distorted by the April 1939 Ordinance on Corporate Profit Distribution and Finance, which capped dividends. This effectively weakened the link between earnings growth and cash payouts, turning equities closer to quasi-fixed-income claims (see \href{https://at-noda.com/appendix/zaibatsu_appendix.pdf}{Online Appendix} A.1). With upside payouts constrained, stock prices became less sensitive to ordinary positive earnings news and more sensitive to regulatory news. However, zaibatsu firms could preserve and reallocate group-level value through internal capital markets, holding-company structures, intercorporate shareholding, and retained earnings in ways that were less available to many independent firms. At the same time, the market consequences of these controls were not uniform even within the military sector, as the empirical results below indicate substantial variation across organizational types and portfolio definitions.

The model below summarizes these institutional channels through several state variables. The baseline policy-control state captures the intensity of administrative intervention. The regime-risk state captures broad political, monetary, diplomatic, and sanctions-related shocks that changed firms' exposure to external dependence and reallocation risk. The embedded-rent state captures news that revised the value of being closely inserted into the state-led strategic order.

\subsection{The Rise of the Zaibatsu: Information Hubs and Rent Extraction}

In this hybrid system, the zaibatsu, especially Mitsui, Mitsubishi, and Sumitomo, were more than mere conglomerates. They were institutional solutions to the frictions of the wartime economy. We conceptualize the zaibatsu as privileged information intermediaries that internalized two markets that were increasingly impaired externally: the capital market and the market for political information. Theoretical literature suggests that business groups create value by reallocating capital among member firms when external markets are imperfect (\citet{almeida2006sbg,hoshi1991csl}). Wartime capital rationing created an extreme form of such imperfection.

While independent Category B firms faced liquidity shortages because of bureaucratic rejection, zaibatsu affiliates could access capital through their group's main bank or holding company (honsha). Zaibatsu banks acted as liquidity buffers, absorbing deposits and underwriting government bonds. The holding company could transfer excess cash from cash-rich subsidiaries, such as mature mining operations, to high-priority war industries, such as aircraft manufacturing, thereby bypassing the regulatory bottlenecks created by the Temporary Funds Adjustment Act. This internal capital market allowed zaibatsu firms to maintain investment levels closer to the optimum than their credit-constrained independent rivals. This organizational advantage plausibly contributed to the relative resilience of some zaibatsu portfolios, although the empirical results also show that wartime abnormal performance cannot be reduced to a simple and uniform zaibatsu premium. Beyond capital, the zaibatsu also possessed privileged access to information (see \href{https://at-noda.com/appendix/zaibatsu_appendix.pdf}{Online Appendix} A.2). The wartime economy was characterized by extreme information asymmetry, as the state held private information regarding strategy, procurement, and allocation. Zaibatsu executives were often embedded in the relevant decision-making apparatus.

Personnel networks illustrate this fusion of zaibatsu and state. Figures such as Shigeaki Ikeda of Mitsui, who served as Finance Minister and Governor of the Bank of Japan, exemplify this overlap. Zaibatsu leaders often headed the control associations (toseikai) that governed their own industries. In a censored and tightly controlled market, zaibatsu affiliation functioned as a credible signal of state backing. Investors could reasonably infer that a Mitsubishi factory would continue to receive raw materials even when competitors were idled. This proximity also facilitated rent seeking. \citet{okazaki2017mei} show that prewar firms with political connections earned significantly higher returns. In the total war economy, where the state's allocation power was greatly expanded, this political premium could be expected to widen, although not necessarily in a mechanically identical way across all affiliated firms.

The relationship between the zaibatsu and the military was symbiotic. The military required the large industrial scale that only the major zaibatsu could provide, for example in aircraft production or shipbuilding. In return, the state created a policy environment that reduced downside risk for selected strategic producers. As \citet{okazaki1994jwe} shows, wartime authorities increasingly accepted the need to preserve profit incentives, revised controlled prices, introduced subsidies or compensation in some priority sectors, and, under the Military Company Law, backed designated firms with privileged finance. Major military suppliers therefore faced a more protected demand and finance environment than ordinary civilian firms. Firms such as Mitsubishi Heavy Industries recorded dramatic profit growth during the early war years. The market therefore valued such firms partly on the basis of the present value of this more protected demand. As a result, zaibatsu shares could appear relatively defensive in selected wartime episodes. They offered a bundle of attributes unavailable to many independent firms: priority access to capital via internal capital markets and high-priority classification, priority access to raw materials via political connections, and protected demand via military procurement.

This structural position is precisely why our CAPM-AR($p$)-SV analysis distinguishes between zaibatsu and non-zaibatsu firms. To treat the market as a homogeneous pool would obscure the fundamental reality that zaibatsu firms operated under a distinct set of economic constraints and opportunities.

In summary, the Japanese stock market between 1930 and 1943 was not an undeveloped frontier, but a sophisticated mechanism operating under the tightening grip of state control. The hybrid nature of the economy meant that market signals were not extinguished, but redirected. The Temporary Funds Adjustment Act bifurcated the cost of capital according to strategic utility. Within this tiered system, the zaibatsu emerged as privileged institutional intermediaries, using internal capital markets, holding-company structures, and political connectivity to mitigate regulatory constraints and preserve group-level value. The event study that follows therefore does not merely trace price movements. It examines how the market processed a complex web of institutional rents, financing advantages, and regulatory risks in real time. The abnormal returns documented below should thus be interpreted not as evidence of a uniform collapse of market rationality, but as evidence that wartime investors priced state-sponsored asymmetries in heterogeneous ways across zaibatsu and non-zaibatsu firms and across military and non-military firms.

This historical structure also motivates the optional group-continuation state used in the \href{https://at-noda.com/appendix/zaibatsu_appendix.pdf}{Online Appendix} A.2. Zaibatsu rents were not confined to contemporaneous operating profits or to military subsidiaries. They could also be capitalized through holding-company dividends, intercorporate shareholding, retained earnings, and group-level financial claims.

%% file: zaibatsu_model.tex
\section{A Model of Wartime Controls and Zaibatsu Rents}
\label{zaibatsu_sec3}

This section develops a partial-equilibrium asset-pricing model designed to match the four empirical portfolios used in our analysis: zaibatsu military, zaibatsu non-military, non-zaibatsu military, and non-zaibatsu non-military firms. The purpose is not to provide a complete macroeconomic theory of wartime Japan. Rather, the model formalizes the mechanisms needed to interpret five empirical features of the data. The first three are the divergence between relative price performance and market-capitalization dominance, segmented sign patterns in event-time abnormal returns, and significant cumulative abnormal returns even when day-0 abnormal returns are weak. The remaining two are events in which non-zaibatsu portfolios react while zaibatsu portfolios remain close to zero, and events in which cumulative responses are concentrated in zaibatsu portfolios, especially the zaibatsu military portfolio.

The model therefore has two layers. The baseline layer contains wartime policy controls, war demand, civilian substitution, finance and issuance, and staged public diffusion. This layer explains the three core empirical patterns. The second layer decomposes wartime policy news into regime-risk shocks and embedded-rent shocks. It explains why some events mainly move non-zaibatsu firms, while others mainly revise the value of being embedded in the state-led strategic order.

\subsection{Trading Environment and Portfolio Claims}

Let $\mathcal G=\{\mathrm{ZM},\mathrm{ZN},\mathrm{NM},\mathrm{NN}\}$, where $\mathrm{ZM}$ denotes zaibatsu military firms, $\mathrm{ZN}$ denotes zaibatsu non-military firms, $\mathrm{NM}$ denotes non-zaibatsu military firms, and $\mathrm{NN}$ denotes non-zaibatsu non-military firms. For each $g\in\mathcal G$, let $Z_g\in\{0,1\}$ indicate zaibatsu affiliation, and let $M_g\in\{0,1\}$ indicate military orientation.

Time is discrete, $t=1,2,\ldots,T,T+1$. There is one risk-free asset with gross return $R>1$ and four risky assets. Risky asset $g$ is a representative traded claim to one unit of installed capital in group $g$. Let $p_t=(p_{g,t})_{g\in\mathcal G}\in\mathbb{R}^4$ denote the vector of ex-dividend market values per unit of installed capital, and let $d_t=(d_{g,t})_{g\in\mathcal G}\in\mathbb{R}^4$ denote the vector of per-unit dividends.

A unit-mass continuum of traders is born in every period, trades when young, and consumes when old. A young trader $i$ at date $t$ chooses risky holdings $q_{i,t}\in\mathbb{R}^4$ and a risk-free position $b_{i,t}\in\mathbb{R}$. If initial wealth is $W_{i,t}$, the budget equations are
\[
W_{i,t}=q_{i,t}^{\top}p_t+b_{i,t}, \qquad
C_{i,t+1}=q_{i,t}^{\top}(p_{t+1}+d_{t+1})+Rb_{i,t}.
\]
Substituting the first equation into the second gives
\[
C_{i,t+1}=RW_{i,t}+q_{i,t}^{\top}\bigl(p_{t+1}+d_{t+1}-Rp_t\bigr).
\]

Preferences are CARA,
\[
U(C)=-\exp\!\left(-\frac{C}{\gamma}\right), \qquad \gamma>0,
\]
where $\gamma$ is the common risk-tolerance parameter. Portfolio positions are unrestricted.

The risky assets are in noisy net supply, $s_t\sim N(0,\Sigma_{s,t})$, where $\Sigma_{s,t}$ is symmetric positive definite. The risk-free asset is in perfectly elastic supply, so only the risky-asset markets have to clear. The variable $s_t$ is a noise-trader imbalance in the representative-claim market. It is distinct from installed capital: noisy supply prevents prices from becoming fully revealing, whereas installed capital determines firm scale and therefore market capitalization. In this normalization, $p_{g,t}$ is a Tobin's-$q$-type valuation per unit of installed capital, so installed capital affects capitalization without entering the short-run representative-claim supply term.

Within each period, the timing is as follows. First, traders observe the information relevant for pricing and form beliefs about future payoffs. Second, they choose portfolio positions and risky-asset markets clear. Third, managers observe current equity prices and choose investment and issuance, which determine next period's installed capital. Finally, next-period dividends, prices, and state realizations occur.

With CARA utility and conditionally Gaussian payoffs, risky-asset demand is linear in expected excess payoffs. Let $\bar E_t[\cdot]$ denote the date-$t$ cohort-average expectation, and let $\Sigma_t$ denote the common conditional covariance matrix of next-period total payoffs. Then optimal risky demand is
\[
q_{i,t}=\gamma\Sigma_t^{-1}\bigl(E_{i,t}[p_{t+1}+d_{t+1}]-Rp_t\bigr),
\]
and market clearing yields the pricing recursion
\begin{equation}
p_t=\frac{1}{R}\bar E_t[p_{t+1}+d_{t+1}]
-\frac{1}{R\gamma}\Sigma_t s_t.
\label{eq:theory_price_recursion}
\end{equation}

Equation \eqref{eq:theory_price_recursion} is the central asset-pricing relation. Current prices depend on cohort-average expectations of next-period prices and dividends. Repeated substitution therefore loads current prices on the expected path of future payoffs. As in standard CARA-normal linear equilibria, sufficiently large noise-trader imbalances can push the affine price formula into economically irrelevant regions with negative prices. The \href{https://at-noda.com/appendix/zaibatsu_appendix.pdf}{Online Appendix} provides sufficient restrictions on coefficient magnitudes and on a compact event-window domain for states and noisy supply that guarantee strictly positive prices and strictly positive installed capital on that domain.

\subsection{Wartime State Variables and Group-level Payoffs}

The baseline model has three aggregate state variables. The policy-control state $a_t$ measures the intensity of wartime administrative intervention, including tighter capital controls, directed credit, dividend regulation, and other policy measures that affect groups differently. The war-demand state $w_t$ measures expected procurement, military expenditure, and the value of state-backed demand for military production. The civilian-substitution state $c_t$ measures the relative gains enjoyed by civilian firms that benefit from scarcity, substitution, or reallocation away from the most war-exposed sectors.

Each state follows an AR(1) process,
\[
a_{t+1}=\rho_a a_t+\zeta^a_{t+1}, \qquad
w_{t+1}=\rho_w w_t+\zeta^w_{t+1}, \qquad
c_{t+1}=\rho_c c_t+\zeta^c_{t+1},
\]
where $0\le \rho_a,\rho_w,\rho_c<1$ and the innovations are mean-zero Gaussian shocks.

Because the model is focused on event windows and medium-run shifts in capitalization rather than on long-run balanced growth, permanent group heterogeneity is absorbed into group constants rather than modeled as a separate latent fundamental. Thus $\bar d_g$ and $\bar p_g$ should be read as reduced-form group-specific average cash flows and average continuation values.

Per-unit dividends satisfy
\[
d_{g,t+1}=\bar d_g+\lambda_g a_t+\psi_g w_t+\nu_g c_t+u_{g,t+1},
\]
where $u_{g,t+1}$ is a mean-zero idiosyncratic dividend shock. Terminal liquidation values satisfy
\[
p_{g,T+1}=\bar p_g+\chi_g a_T+\varphi_g w_T+\omega_g c_T.
\]
The coefficients $\lambda_g$ and $\chi_g$ measure exposure to policy controls, $\psi_g$ and $\varphi_g$ measure exposure to war demand, and $\nu_g$ and $\omega_g$ measure exposure to civilian substitution.

To keep the mapping to the empirical $2\times 2$ design transparent, the war-demand loadings are parameterized as
\[
\psi_g=\psi^M M_g+\psi^{ZM}Z_gM_g, \qquad
\varphi_g=\varphi^M M_g+\varphi^{ZM}Z_gM_g,
\]
with $\psi^M>0$, $\varphi^M\ge 0$, and $\psi^{ZM},\varphi^{ZM}\ge 0$. This parameterization says that military firms benefit from war demand and that zaibatsu military firms may enjoy an additional rent.

The civilian-substitution loadings are parameterized as
\[
\nu_g=\nu^N(1-M_g)+\nu^{ZN}Z_g(1-M_g), \qquad
\omega_g=\omega^N(1-M_g)+\omega^{ZN}Z_g(1-M_g).
\]
This parameterization makes the substitution state relevant only for civilian portfolios and allows civilian zaibatsu and civilian independent firms to respond differently to the same event.

The policy-control loadings $(\lambda_g,\chi_g)$ are left flexible, because tighter controls could raise expected rents for some groups but tighten effective financing or material constraints for others. In particular, the model allows independent military firms to benefit from war demand and still be hurt by tighter administrative controls.

\subsection{Information, Beliefs, and Price Formation}

At each date, all traders observe noisy public signals about the three aggregate states,
\[
z_t^a=a_t+\xi_t^a, \qquad
z_t^w=w_t+\xi_t^w, \qquad
z_t^c=c_t+\xi_t^c,
\]
where the signal noises are mean-zero Gaussian and mutually independent. Let $\hat a_t$, $\hat w_t$, and $\hat c_t$ denote the cohort-average posterior means that summarize the information relevant for pricing at date $t$. These posterior means may combine the public signals with information revealed through the equilibrium price.

To keep the model close to the empirical four-portfolio design, we work with an affine equilibrium under restricted information aggregation. The working assumption is that $(\hat a_t,\hat w_t,\hat c_t,s_t)$ are sufficient pricing statistics and that any additional cross-portfolio inference that prices might reveal is absorbed into the affine coefficients. This is the same sense in which the model is partial equilibrium: the theory is designed to generate disciplined pricing implications for the four observed portfolios rather than to solve the full information aggregation problem across all Japanese firms.

The model can also accommodate an optional leakage channel. If one wants to study pre-event drift, traders connected to zaibatsu military firms may observe an additional lead signal before a major war announcement,
\[
x^{\mathrm{lead}}_{i,\tau-1}=w_{\tau}+\nu_{i,\tau-1}, \qquad
\eta^{\mathrm{lead}}_{\mathrm{ZM}}>0,
\]
where $\eta^{\mathrm{lead}}_{\mathrm{ZM}}$ denotes the precision of the lead signal and no other group receives such a signal. This optional extension is not needed for the main propositions, but it is useful when the empirical application turns to pre-event drift and possible information leakage.

Finally, the linear-demand aggregation step requires a common conditional covariance matrix of next-period total payoffs at date $t$. In the affine equilibrium considered below, this matrix is not chosen independently of prices. It is the covariance induced by the same affine price process, the Gaussian state innovations, dividend shocks, and next-period noisy supply. This is the covariance-consistency restriction used in the \href{https://at-noda.com/appendix/zaibatsu_appendix.pdf}{Online Appendix}.

\subsection{Financing Wedges, Issuance, and Capitalization}

The key ingredient behind the divergence between price performance and market capitalization is the distinction between the per-unit stock price and the total quantity of installed capital. Let $k_{g,t}>0$ denote installed capital, or equivalently firm scale measured in units of productive capacity, in group $g$ at date $t$. The representative traded claim is normalized to one unit of installed capital, so $p_{g,t}$ is the market value of one unit of productive capacity and market capitalization equals that valuation times the installed capital stock. The market capitalization of group $g$ is
\[
\mathrm{mc}_{g,t}=p_{g,t}k_{g,t}, \qquad
\sigma_{g,t}=
\frac{\mathrm{mc}_{g,t}}{\sum_{h\in\mathcal G}\mathrm{mc}_{h,t}}.
\]
Thus a group's market-cap share depends on both valuation and scale. Current market capitalization uses the installed capital stock $k_{g,t}$ in place before current-period issuance, while new issuance changes $k_{g,t+1}$ and therefore the next period's capitalization.

The representative manager of group $g$ chooses an investment rate $i_{g,t}$ that determines next period's installed capital according to
\[
k_{g,t+1}=k_{g,t}(1-\delta+i_{g,t}), \qquad 0<\delta<1.
\]
Managers take current stock prices as given when choosing investment. Over the short horizon emphasized in the event study, per-unit dividends are treated as independent of current scale, so capital accumulation changes total capitalization through firm scale rather than through immediate feedback into per-unit payoffs.

Investment is financed by equity issuance and internal funds. Wartime controls create a financing wedge $\mu_{g,t}$ between the market value of an additional unit of equity and the effective replacement cost of an additional unit of capacity. Zaibatsu affiliation lowers this wedge because internal capital markets soften external financing frictions. Military priority lowers it because directed credit and procurement priorities improve access to funds. Zaibatsu military firms enjoy both advantages.

Formally, the manager solves
\[
\max_{i\in\mathbb{R}}
\left\{
p_{g,t}i-(1+\mu_{g,t})i-\frac{i^2}{2\Gamma_g}
\right\},
\qquad \Gamma_g>0,
\]
which yields the investment rule
\[
i_{g,t}=\Gamma_g(p_{g,t}-1-\mu_{g,t}).
\]
The parameter $\Gamma_g$ governs how aggressively a high stock price translates into new capital formation. A larger $\Gamma_g$ means that a given valuation difference creates a larger capitalization response. The control variable $i_{g,t}$ is a net expansion rate, so negative values represent contraction or incomplete replacement of depreciated capital rather than literal negative gross investment.

In the baseline model, the financing wedge is
\[
\mu_{g,t}=\bar\mu-\mu^ZZ_g-\mu^M M_g a_t-\mu^{ZM}Z_gM_ga_t,
\]
where $\mu^Z,\mu^M,\mu^{ZM}>0$. The term $\mu^Z$ captures the internal capital-market advantage of all zaibatsu firms. The term $\mu^M M_g a_t$ captures the fact that wartime administrative controls channel finance toward military production. The interaction term $\mu^{ZM}Z_gM_ga_t$ captures the extra financing advantage of politically connected zaibatsu military firms.

The capital accumulation equation, the investment rule, and the financing wedge allow capitalization to move differently from per-unit prices. This mechanism rationalizes why zaibatsu military firms can gain market-capitalization share even when non-zaibatsu military firms appear strong in relative price indices.

\subsection{Event Timing and Staged Public Diffusion}

The paper studies short event windows around public wartime announcements. Let $\tau$ denote an event date and let
\[
\Delta \hat a_\tau:=\hat a_\tau-\hat a_{\tau^-}, \qquad
\Delta \hat w_\tau:=\hat w_\tau-\hat w_{\tau^-}, \qquad
\Delta \hat c_\tau:=\hat c_\tau-\hat c_{\tau^-},
\]
be the announcement-induced revisions in the pricing states, where $\tau^-$ denotes beliefs immediately before the event.

To explain delayed capitalization, consider a two-day event window consisting of dates $\tau$ and $\tau+1$. A fraction $\pi\in(0,1]$ of traders are early processors who understand the announcement on date $\tau$, while the remaining fraction $1-\pi$ are late processors who fully process it only on date $\tau+1$. No additional shocks arrive inside the two-day event window. The announcement is public from the outset. What differs across traders is the speed with which they translate that public information into payoff-relevant beliefs.

This staged diffusion mechanism is local to the event window. It does not replace the underlying noisy rational expectations asset market. Its purpose is narrower: it provides a transparent explanation for the empirical cases in which day-0 abnormal returns are weak but cumulative abnormal returns over a short window are statistically important.

\subsection{Affine Pricing and Event-study Decomposition}

We focus on affine equilibria in which per-unit prices are linear in the pricing states and in noisy supply. Under the sufficient conditions stated in the \href{https://at-noda.com/appendix/zaibatsu_appendix.pdf}{Online Appendix}, there exists a unique affine price function of the form
\[
p_t=A_t+B_t\hat a_t+C_t\hat w_t+D_t\hat c_t
-\frac{1}{R\gamma}\Sigma_t s_t,
\]
where $A_t$, $B_t$, $C_t$, and $D_t$ are date-specific coefficient vectors in $\mathbb{R}^4$.

The baseline event-study decomposition follows immediately from repeated substitution in the pricing recursion \eqref{eq:theory_price_recursion}. Let $H=T+1-\tau$ denote the remaining horizon at the event date. Then the local event-time decomposition is
\begin{equation}
AR_{g,\tau}
\approx
\beta^a_{g,\tau}\Delta \hat a_\tau
+\beta^w_{g,\tau}\Delta \hat w_\tau
+\beta^c_{g,\tau}\Delta \hat c_\tau.
\label{eq:theory_ar_map}
\end{equation}
The associated exposure coefficients are
\begin{equation}
\beta^x_{g,\tau}
=
\frac{1}{p_{g,\tau^-}}
\left[
A_H(\rho_x)\ell_g^x+B_H(\rho_x)m_g^x
\right],
\qquad x\in\{a,w,c\},
\label{eq:theory_beta}
\end{equation}
where
\[
\ell_g^a=\lambda_g,\quad \ell_g^w=\psi_g,\quad \ell_g^c=\nu_g,
\qquad
m_g^a=\chi_g,\quad m_g^w=\varphi_g,\quad m_g^c=\omega_g,
\]
and
\[
A_H(\rho)=\sum_{j=1}^H R^{-j}\rho^{j-1},
\qquad
B_H(\rho)=R^{-H}\rho^{H-1}.
\]
Thus the empirical abnormal return for each portfolio can be read as the weighted sum of its exposures to policy controls, war demand, and civilian substitution.

\subsection{Decomposing Wartime Policy News}

The baseline model is designed for the three main empirical targets: capitalization concentration, segmented four-way sign patterns, and delayed capitalization. The event study, however, also contains two additional patterns. In some events, broad political, diplomatic, monetary, or sanctions-related shocks move the non-zaibatsu portfolios while the zaibatsu portfolios remain close to zero. In other events, cumulative responses are concentrated in zaibatsu portfolios, especially the zaibatsu military portfolio. To capture these patterns, we decompose wartime policy news into two additional components.

The first component, $a^R_t$, is a regime-risk state. It summarizes broad political, diplomatic, sanctions, and monetary-regime shocks that alter external dependence, access to outside finance, and vulnerability to administrative reallocation. The second component, $a^E_t$, is an embedded-rent state. It summarizes news that changes the continuation value of firms whose profits depend on close insertion into the state-led strategic order. These states follow
\[
a^R_{t+1}=\rho_R a^R_t+\zeta^R_{t+1}, \qquad
a^E_{t+1}=\rho_E a^E_t+\zeta^E_{t+1},
\qquad 0\le \rho_R,\rho_E<1.
\]

The extended per-unit dividend process is
\[
d^{\mathrm{ext}}_{g,t+1}
=
\bar d_g+\lambda^R_g a^R_t+\lambda^E_g a^E_t
+\psi_g w_t+\nu_g c_t+u_{g,t+1},
\]
and the extended terminal value is
\[
p^{\mathrm{ext}}_{g,T+1}
=
\bar p_g+\chi^R_g a^R_T+\chi^E_g a^E_T
+\varphi_g w_T+\omega_g c_T.
\]
The coefficients $(\lambda^R_g,\chi^R_g)$ measure exposure to regime risk, while $(\lambda^E_g,\chi^E_g)$ measure exposure to embedded rents.

The financing wedge may also load on regime risk:
\[
\mu^{\mathrm{ext}}_{g,t}
=
\bar\mu-\mu^ZZ_g-\mu^W M_gw_t-\mu^{ZW}Z_gM_gw_t+\tau^R_g a^R_t,
\]
where $\mu^W,\mu^{ZW}>0$. The term $\tau^R_g$ is a reduced-form measure of external dependence. The natural historical case is
\[
|\tau^R_{\mathrm{ZM}}|<|\tau^R_{\mathrm{NM}}|,
\qquad
|\tau^R_{\mathrm{ZN}}|<|\tau^R_{\mathrm{NN}}|,
\]
so that regime shocks move financing conditions and continuation values more strongly for non-zaibatsu firms than for zaibatsu firms. Relative to the baseline financing wedge above, the extended wedge uses $w_t$ rather than $a_t$ in the military-priority terms in order to tie financing advantages more directly to the war-demand environment. The regime-risk term captures a distinct channel: exposure to external dependence and vulnerability to broad political, diplomatic, and monetary-regime shocks.

The corresponding extended event-study decomposition is
\begin{equation}
AR^{\mathrm{ext}}_{g,\tau}
\approx
\beta^R_{g,\tau}\Delta \hat a^R_\tau
+\beta^E_{g,\tau}\Delta \hat a^E_\tau
+\beta^w_{g,\tau}\Delta \hat w_\tau
+\beta^c_{g,\tau}\Delta \hat c_\tau.
\label{eq:theory_ext_ar_map}
\end{equation}
The additional exposure coefficients are
\[
\beta^R_{g,\tau}
=
\frac{
A_H(\rho_R)\lambda^R_g+B_H(\rho_R)\chi^R_g
}{p_{g,\tau^-}},
\qquad
\beta^E_{g,\tau}
=
\frac{
A_H(\rho_E)\lambda^E_g+B_H(\rho_E)\chi^E_g
}{p_{g,\tau^-}},
\]
while $\beta^w_{g,\tau}$ and $\beta^c_{g,\tau}$ are defined as in \eqref{eq:theory_beta}.
The affine-equilibrium argument carries over once the sufficient statistic is enlarged to include
$(\hat a^R_t,\hat a^E_t)$.

An extension of the model allows for an optional zaibatsu group continuation state, denoted by $G_t$. This state captures revisions in the value of internal group claims, holding-company dividends, cross-shareholdings, and group-level continuation values. It is most useful for cases in which the zaibatsu non-military portfolio responds in cumulative terms even when a direct military-demand interpretation is weak. In this case, the extended event-study decomposition is augmented by one additional term, $\beta^G_{g,\tau}\Delta \hat G_\tau$, with loadings concentrated in zaibatsu portfolios.

The model remains partial equilibrium. It endogenizes capitalization and short-run issuance, but it does not solve a full macroeconomic wartime general equilibrium. This restriction is deliberate. The purpose is to match the empirical four-portfolio design and to isolate the channels that matter for the financial-history evidence: wartime demand, administrative controls, civilian substitution, internal capital markets, regime-risk insulation, embedded rents, and slow diffusion of public news.

\subsection{Theoretical Implications for the Empirical Patterns}

The following propositions correspond directly to the empirical patterns documented below.

\vspace{3mm}

\begin{proposition}
\label{prop:cap_concentration}
Suppose that prices are strictly positive, that installed capital is strictly positive, and that the capital accumulation equation and investment rule above hold on the event window under consideration. Define the gross price growth factor
\[
\mathcal G^p_{g,t+1}:=\frac{p_{g,t+1}}{p_{g,t}}.
\]
If, for every $h\neq \mathrm{ZM}$,
\[
\frac{\mathcal G^p_{\mathrm{ZM},t+1}}{\mathcal G^p_{h,t+1}}
\cdot
\frac{
1-\delta+\Gamma_{\mathrm{ZM}}
\bigl(p_{\mathrm{ZM},t}-1-\mu_{\mathrm{ZM},t}\bigr)
}{
1-\delta+\Gamma_h
\bigl(p_{h,t}-1-\mu_{h,t}\bigr)
}
>1,
\]
then the zaibatsu military market-cap share rises from $t$ to $t+1$:
\[
\sigma_{\mathrm{ZM},t+1}>\sigma_{\mathrm{ZM},t}.
\]
In particular, it is possible that non-zaibatsu military firms outperform zaibatsu military firms in per-unit prices,
\[
\mathcal G^p_{\mathrm{NM},t+1}>\mathcal G^p_{\mathrm{ZM},t+1},
\]
while the market-cap share of zaibatsu military firms still rises, provided the finance and issuance term in the condition above is sufficiently favorable to $\mathrm{ZM}$.
\end{proposition}

\begin{proof}
See the \href{https://at-noda.com/appendix/zaibatsu_appendix.pdf}{Online Appendix} A.3.2.
\end{proof}

Proposition \ref{prop:cap_concentration} formalizes the logic behind the divergence between relative price performance and market-capitalization dominance. Per-unit prices summarize the value of one unit of installed capital. Market capitalization depends on that valuation and on the scale of installed capital. Once wartime controls lower the financing wedge most strongly for zaibatsu military firms, they can translate valuations into scale more effectively than independent rivals.

\vspace{3mm}

\begin{proposition}
\label{prop:sign_segmentation}
Suppose the local event-study decomposition \eqref{eq:theory_ar_map} applies. For notational simplicity, write $\beta_g^x$ for the event-specific coefficient $\beta^x_{g,\tau}$. Assume the state exposures satisfy
\begin{align*}
&\beta_{\mathrm{ZM}}^w>0,\qquad \beta_{\mathrm{ZM}}^a\ge 0,\qquad \beta_{\mathrm{ZM}}^c=0,\\
&\beta_{\mathrm{ZN}}^a<0,\qquad \beta_{\mathrm{ZN}}^w=0,\qquad \beta_{\mathrm{ZN}}^c\le 0,\\
&\beta_{\mathrm{NM}}^a<0<\beta_{\mathrm{NM}}^w,\qquad \beta_{\mathrm{NM}}^c=0,\\
&\beta_{\mathrm{NN}}^c>0,\qquad \beta_{\mathrm{NN}}^w=0,\qquad \beta_{\mathrm{NN}}^a\le 0.
\end{align*}
Then the following hold.

\begin{itemize}
\item[(i)] If an escalation event satisfies
\[
\Delta \hat a_\tau>0,\qquad
\Delta \hat w_\tau>0,\qquad
\Delta \hat c_\tau>0,
\]
and
\begin{align*}
&\beta_{\mathrm{NN}}^c\Delta \hat c_\tau>
|\beta_{\mathrm{NN}}^a|\Delta \hat a_\tau,\\
&|\beta_{\mathrm{NM}}^a|\Delta \hat a_\tau>
\beta_{\mathrm{NM}}^w\Delta \hat w_\tau,
\end{align*}
then
\[
AR_{\mathrm{ZM},\tau}>0,\qquad
AR_{\mathrm{ZN},\tau}<0,\qquad
AR_{\mathrm{NM},\tau}<0,\qquad
AR_{\mathrm{NN},\tau}>0.
\]
This is the 1939-style segmented sign pattern.

\item[(ii)] If a setback event satisfies
\[
\Delta \hat a_\tau<0,\qquad
\Delta \hat w_\tau<0,\qquad
\Delta \hat c_\tau>0,
\]
and
\begin{align*}
&|\beta_{\mathrm{ZN}}^a\Delta \hat a_\tau|>
|\beta_{\mathrm{ZN}}^c\Delta \hat c_\tau|,\\
&\beta_{\mathrm{NN}}^c\Delta \hat c_\tau>
|\beta_{\mathrm{NN}}^a\Delta \hat a_\tau|,
\end{align*}
then
\[
AR_{\mathrm{ZM},\tau}<0,\qquad
AR_{\mathrm{ZN},\tau}>0,\qquad
AR_{\mathrm{NN},\tau}>0.
\]
If, in addition,
\[
|\beta_{\mathrm{NM}}^a\Delta \hat a_\tau|
\le
|\beta_{\mathrm{NM}}^w\Delta \hat w_\tau|,
\]
then the non-zaibatsu military reaction is nonpositive.
\end{itemize}
\end{proposition}

\begin{proof}
See the \href{https://at-noda.com/appendix/zaibatsu_appendix.pdf}{Online Appendix} A.3.3.
\end{proof}

Proposition \ref{prop:sign_segmentation} is where the civilian-substitution state matters most. Without $c_t$, the two civilian portfolios tend to move together. Once $c_t$ is added, a positive war event can hurt the civilian zaibatsu portfolio because tighter controls dominate, while helping the independent civilian portfolio because substitution or scarcity rents dominate.

\vspace{3mm}

\begin{proposition}
\label{prop:staged_diffusion}
Suppose the local event-study decomposition \eqref{eq:theory_ar_map} applies and consider a two-day event experiment over dates $\tau$ and $\tau+1$ in which no new shocks arrive inside the window and the announcement induces revisions in next-period payoff components. Let 
\[
\Theta_g
:=
\beta^a_{g,\tau}\Delta \hat a_\tau
+\beta^w_{g,\tau}\Delta \hat w_\tau
+\beta^c_{g,\tau}\Delta \hat c_\tau
\]
denote the full-information impact for group $g$. Suppose a fraction $\pi\in[0,1]$ of traders are early processors and update on date $\tau$, while the remaining fraction $1-\pi$ updates on date $\tau+1$. Then, in the local approximation,
\[
AR_{g,\tau}=\pi\Theta_g,\qquad
AR_{g,\tau+1}=(1-\pi)\Theta_g,\qquad
CAR_{g,[\tau,\tau+1]}=\Theta_g.
\]
Thus the day-0 abnormal return can be small even when the short-window cumulative abnormal return is large. In particular, if empirical critical values satisfy
\[
|\pi\Theta_g|<c_{AR}
\qquad \text{but} \qquad
|\Theta_g|>c_{CAR},
\]
then the event study records an insignificant day-0 abnormal return together with a significant short-window cumulative abnormal return.
\end{proposition}

\begin{proof}
See the \href{https://at-noda.com/appendix/zaibatsu_appendix.pdf}{Online Appendix} A.3.4.
\end{proof}

The lead-signal extension and staged diffusion are complementary mechanisms. A zaibatsu military lead signal can generate pre-event drift. Proposition \ref{prop:staged_diffusion}, by contrast, concerns the gradual processing of public news after the announcement. Leakage explains why $\mathrm{ZM}$ may begin to move before the event, while staged diffusion explains why a public announcement may generate only a weak day-0 abnormal return but a significant short-window CAR.

\vspace{3mm}

\begin{proposition}
\label{prop:regime_insulation}
Suppose the extended event-study decomposition \eqref{eq:theory_ext_ar_map} applies and an event at date $\tau$ is a pure regime-risk shock:
\[
\Delta \hat a^R_\tau\neq 0,\qquad
\Delta \hat a^E_\tau=\Delta \hat w_\tau=\Delta \hat c_\tau=0.
\]
Define
\[
\Theta^R_g:=\beta^R_{g,\tau}\Delta \hat a^R_\tau.
\]
If there exist constants $0\le \bar\varepsilon_Z<\varepsilon_N$ such that
\[
|\Theta^R_{\mathrm{ZM}}|,\ |\Theta^R_{\mathrm{ZN}}|
\le
\bar\varepsilon_Z,
\qquad
|\Theta^R_{\mathrm{NM}}|,\ |\Theta^R_{\mathrm{NN}}|
\ge
\varepsilon_N,
\]
then both zaibatsu portfolios are approximately insulated while both non-zaibatsu portfolios react. If, in addition,
\[
\Theta^R_{\mathrm{NM}}\Theta^R_{\mathrm{NN}}<0,
\]
then the two non-zaibatsu portfolios move in opposite directions. If the announcement is processed on date $\tau$ and no additional payoff-relevant revision arrives before $\tau+1$, the corresponding short-window cumulative abnormal returns have the same sign pattern.
\end{proposition}

\begin{proof}
See the \href{https://at-noda.com/appendix/zaibatsu_appendix.pdf}{Online Appendix} A.3.6.
\end{proof}

Proposition \ref{prop:regime_insulation} formalizes the idea that zaibatsu affiliation buffered firms against broad regime shocks. In the model, that buffering appears as a small regime-risk coefficient $\beta^R_{g,\tau}$ for the zaibatsu portfolios. Economically, the small coefficient summarizes diversification across activities, internal capital markets, access to group finance, and preferential treatment in administrative allocation. By contrast, the non-zaibatsu coefficients are large because their continuation values are more exposed to outside finance, sanctions, and reallocation shocks.

This mechanism is most directly relevant to the early-1930s regime shocks and the 1941 sanctions shock. These events are naturally interpreted as shocks to external dependence and reallocation risk: zaibatsu portfolios are partly insulated by internal capital markets and privileged access, whereas non-zaibatsu portfolios absorb most of the measured abnormal-return response.

\vspace{3mm}

\begin{proposition}
\label{prop:embedded_rent}
Suppose the extended event-study decomposition \eqref{eq:theory_ext_ar_map} applies and consider the two-day staged-diffusion experiment of Proposition \ref{prop:staged_diffusion}. Let the event at date $\tau$ be a pure embedded-rent shock:
\[
\Delta \hat a^E_\tau\neq 0,\qquad
\Delta \hat a^R_\tau=\Delta \hat w_\tau=\Delta \hat c_\tau=0.
\]
Define
\[
\Theta^E_g:=\beta^E_{g,\tau}\Delta \hat a^E_\tau.
\]
Suppose there exist constants $0\le \bar\varepsilon_O<\varepsilon_{\mathrm{ZM}}$ such that
\[
|\Theta^E_{\mathrm{ZM}}|\ge \varepsilon_{\mathrm{ZM}},
\qquad
|\Theta^E_{\mathrm{ZN}}|,\ |\Theta^E_{\mathrm{NM}}|,\ |\Theta^E_{\mathrm{NN}}|
\le \bar\varepsilon_O.
\]
If a fraction $\pi\in(0,1)$ of traders processes the public announcement on date $\tau$, then
\[
AR_{g,\tau}=\pi\Theta^E_g,
\qquad
CAR_{g,[\tau,\tau+1]}=\Theta^E_g.
\]
Hence the day-0 abnormal return for $\mathrm{ZM}$ can be weak while the short-window cumulative abnormal return is large, whereas the other three portfolios remain nearly unresponsive. In particular, if empirical critical values satisfy
\[
|\pi\Theta^E_{\mathrm{ZM}}|<c_{AR},
\qquad
|\Theta^E_{\mathrm{ZM}}|>c_{CAR},
\qquad
|\Theta^E_h|<c_{CAR}
\quad \text{for } h\in\{\mathrm{ZN},\mathrm{NM},\mathrm{NN}\},
\]
then the event study records an insignificant day-0 abnormal return and a significant CAR or SCAR only for $\mathrm{ZM}$.
\end{proposition}

\begin{proof}
See the \href{https://at-noda.com/appendix/zaibatsu_appendix.pdf}{Online Appendix} A.3.7.
\end{proof}

Proposition~\ref{prop:embedded_rent} formalizes a mechanism distinct from Proposition~\ref{prop:regime_insulation}. The event is not a broad regime-risk shock. It is public news about the value of being embedded in the state-led strategic order. Such news can revise the continuation value of the zaibatsu military portfolio even when the other portfolios remain muted. When the diffusion of that public information is gradual, the effect can appear more clearly in CAR and SCAR than in day-0 AR.

This mechanism is most directly relevant to events in which cumulative responses are concentrated in the zaibatsu military portfolio, such as diplomatic, military, or strategic events that revised the value of state-embedded rents.

Taken together, Propositions~\ref{prop:cap_concentration}--\ref{prop:embedded_rent} map the model's mechanisms into the empirical evidence. Financing wedges and issuance explain the divergence between per-unit price performance and zaibatsu military capitalization dominance. Civilian substitution explains why $\mathrm{ZN}$ and $\mathrm{NN}$ need not move together. Staged public diffusion explains why CAR can be informative even when day-0 AR is weak. Regime-risk insulation explains why broad political, diplomatic, or sanctions-related shocks can move non-zaibatsu portfolios while leaving zaibatsu portfolios muted. Embedded-rent shocks explain why some events generate cumulative responses concentrated in $\mathrm{ZM}$, while the optional group-continuation extension covers the narrower case of $\mathrm{ZN}$-concentrated cumulative responses.

%% file: zaibatsu_frame.tex
\section{Empirical Framework}\label{zaibatsu_sec4}

This section outlines the empirical framework used to evaluate how wartime controls and related regime shocks affected stock-price formation in Japan between 1930 and 1943. The analysis is based on an event study design grounded in an augmented version of the Capital Asset Pricing Model (CAPM) that incorporates serial correlation and stochastic volatility in the error structure. We refer to this model as the CAPM-AR($p$)-SV. Standard event studies typically rely on linear factor models such as the market model or the static CAPM, which assume homoskedastic and serially uncorrelated errors (\citet{brown1985uds} and \citet{mackinlay1997ese}). These assumptions are unlikely to hold in historical financial markets, which were subject to limited liquidity, institutional frictions, and frequent government intervention. A more flexible specification is therefore required to accommodate time-varying volatility and autocorrelated returns.

The CAPM-AR($p$)-SV can be viewed as a reduced-form representation of the conditional CAPM (\citet{jagannathan1996ccc} and \citet{lettau2001rcapm}), where second moments of returns vary with evolving information sets. In our setting, the market beta ($\beta$) is assumed to be constant, while the conditional variance follows a latent stochastic process, and the residuals follow an AR($p$) dynamic. These features serve as reduced-form proxies for time-varying risk premia and investor uncertainty.

This structure is well suited to capturing salient properties of daily financial data, including volatility clustering and return autocorrelation (\citet{engle1982arch}, \citet{taylor1986mft} and \citet{ghysels1996sv}). The autoregressive component reflects delayed information processing or market microstructure frictions, while the stochastic volatility component captures changes in perceived risk due to macroeconomic shocks or policy announcements. These dynamics are particularly relevant in prewar and wartime Japan, when market functioning was repeatedly disrupted by regulatory and policy interventions.

By using the CAPM-AR($p$)-SV model as the benchmark for expected returns, we address the autocorrelation and heteroskedasticity that can bias standard abnormal-return estimates. The resulting abnormal returns are therefore less likely to reflect misspecified return dynamics and more likely to capture event-related revisions in expected payoffs, risk exposures, and institutional rents. This framework enables more credible inference on how market prices incorporated wartime controls, regime-risk shocks, and embedded-rent revisions.

\subsection{CAPM with AR($p$) Residuals and Stochastic Volatility}

To capture the key empirical properties of historical daily returns, we employ a return-generating process that extends the standard CAPM by incorporating both autoregressive dynamics and stochastic volatility in the error term. The daily return for portfolio $i$ at time $t$, denoted $R_{i,t}$, is modeled as:
\begin{equation}
 R_{i,t}=\alpha_i+\beta_i R_{m,t}+\varepsilon_{i,t},\quad i = 1,\dots,n;\; t = 1,\dots,T, \label{eq1}
\end{equation}
where $R_{m,t}$ is the market return at time $t$, $\alpha_i$ is the portfolio-specific intercept, $\beta_i$ is the market beta for portfolio $i$, and $\varepsilon_{i,t}$ is the model residual for portfolio $i$ at tiem $t$. 

The residuals $\varepsilon_{i,t}$ follow an AR($p$) process with stochastic volatility:
\begin{align}
\varepsilon_{i,t} &= \rho_{i,1} \varepsilon_{i,t-1}+\cdots+\rho_{i,p} \varepsilon_{i,t-p}+\eta_{i,t}, \label{eq2}\\
\eta_{i,t} &= \sigma_{i,t} \zeta_{i,t}, \quad \zeta_{i,t}\sim\mathcal{N}(0,1), \quad \sigma^2_{i,t} = \exp(h_{i,t}), \label{eq3}\\
h_{i,t} &= \mu_i + \phi_i(h_{i,t-1} - \mu_i) + \sigma_{i,\tau} \xi_{i,t}, \quad \xi_{i,t} \sim \mathcal{N}(0,1), \label{eq4}
\end{align}
Here, $\rho_{i,j} \ (j=1,\ldots,p)$ capture the AR($p$) serial correlation coefficients in residuals for portfolio $i$. The conditional variance of the error term, $\sigma_{i,t}^2$, evolves over time according to a latent log-volatility process $h_{i,t}$. The parameters $\mu_i$, $\phi_i$, and $\sigma_{i,\tau}$ govern, respectively, the unconditional mean of the log-volatility process, persistence, and innovation variance of the log-volatility for each portfolio. Notably, $\sigma_{i,\tau}$ serves as a scale parameter, determining the magnitude of random shocks to the volatility process; larger values of $\sigma_{i,\tau}$ yield greater volatility fluctuations.

To implement this model in practice, it is necessary to adopt an estimation strategy that jointly accounts for the AR($p$) serial correlation and the latent stochastic volatility in the error process. Standard estimation methods, such as OLS or traditional GLS, are inadequate because they do not address these two sources of dependence simultaneously. Therefore, we employ a sequential multi-step estimation procedure, which ensures that both autocorrelation and time-varying volatility are properly accommodated in the estimation of model parameters.

\vspace{1em}
\noindent
\textbf{Estimation Procedure}~~We estimate the model in six steps to account for both serial correlation and time-varying volatility:
\begin{enumerate}
    \item \textbf{Estimate the static CAPM via OLS.}~~Using equation (\ref{eq1}), we obtain preliminary estimates $\hat{\alpha}_i$ and $\hat{\beta}_i$ for each portfolio $i$, and compute the residuals:
    \begin{equation}
    \hat{\varepsilon}_{i,t} = R_{i,t} - \hat{\alpha}_i - \hat{\beta}_i R_{m,t}\nonumber
    \end{equation}
    \item \textbf{Estimate AR($p$) coefficients and compute the residuals.}~~Using the OLS residuals $\hat\varepsilon_{i,t}$, we estimate the autocorrelation coefficients $\rho_{i,j}$ and select the optimal lag order based on the Bayesian Information Criterion (BIC) of \citet{schwarz1978edm}. We then compute the residuals in Equation (\ref{eq2}) as follows:
    \begin{eqnarray}
     \tilde{\eta}_{i,t}=
      \begin{cases}
       \hat{\varepsilon}_{i,t} & (t\leq p)\\
       \hat{\varepsilon}_{i,t}-\hat{\rho}_{i,1}\hat{\varepsilon}_{i,t-1}-\cdots-\hat{\rho}_{i,p}\hat{\varepsilon}_{i,t-p} & (t\geq p+1)
      \end{cases}\nonumber
    \end{eqnarray}
    \item \textbf{Estimate the stochastic volatility component.}~~We fit the stochastic volatility model (equations (\ref{eq2})--(\ref{eq4})) to $\tilde\eta_{i,t}$ using Markov Chain Monte Carlo (MCMC) methods, employing forward-filtering backward-sampling (FFBS) for the latent volatility $h_{i,t}$ and Metropolis-Hastings for the model parameters. The posterior mean of $\exp(h_{i,t})$ provides the estimate of $\sigma_{i,t}^2$.
    \item \textbf{Transform to the quasi-difference form (observations $t\geq p+1$).}~~Applying \citetapos{prais1954tes} procedure, we obtain the following quasi-difference series from Equations (\ref{eq1}) and (\ref{eq2}):
          \begin{equation}
	    R_{i,t}-\sum_{j=1}^p\rho_{i,j}R_{i,t-j}=\left(1-\sum_{j=1}^p\rho_{i,j}\right)\alpha_i+\left(R_{m,t}-\sum_{j=1}^p\rho_{i,j}R_{m,t-j}\right)\beta_i+\eta_{i,t}.\label{eq5}
	  \end{equation}
	  We then replace the coefficients ($\rho_{i,j}$, $\eta_{i,t}$) in Equation (\ref{eq5}) with their estimates ($\hat\rho_{i,j}$, $\tilde\eta_{i,t}$) and express the equation in matrix form as follows:
	  \[
	    \widetilde{\ve{R}}_{i}^{QD_2}=\widetilde{\ve{C}}_{i}^{QD_2}\alpha_{i}+\widetilde{\ve{R}}_{m}^{QD_2}\beta_i+\widetilde{\ve{\eta}}_i^{QD_2},
	  \]
	  where
	  \begin{align}
	    \widetilde{\ve{R}}_{i}^{QD_2}=&
	    \begin{bmatrix}
	      R_{i,p+1} - \sum_{j=1}^p\widehat{\rho}_{i,j}R_{i,p+1-j}\\
	      R_{i,p+2} - \sum_{j=1}^p\widehat{\rho}_{i,j}R_{i,p+2-j}\\
	      \vdots\\
	      R_{i,T} - \sum_{j=1}^p\widehat{\rho}_{i,j}R_{i,T-j}\\
	    \end{bmatrix},\quad
	    \widetilde{\ve{C}}_{i}^{QD_2}=\left(1-\sum_{j=1}^p\widehat{\rho}_{i,j}\right)\ve{1}_{T-p},\nonumber\\
	    \widetilde{\ve{R}}_{m}^{QD_2}=&
	    \begin{bmatrix}
	      R_{m,p+1} - \sum_{j=1}^p\widehat{\rho}_{i,j}R_{m,p+1-j}\\
	      R_{m,p+2} - \sum_{j=1}^p\widehat{\rho}_{i,j}R_{m,p+2-j}\\
	      \vdots\\
	      R_{m,T} - \sum_{j=1}^p\widehat{\rho}_{i,j}R_{m,T-j}\\
	    \end{bmatrix},\quad
	    \widetilde{\ve{\eta}}_i^{QD_2}=
	    \begin{bmatrix}
	      \widetilde{\eta}_{i,p+1}\\
	      \widetilde{\eta}_{i,p+2}\\
	      \vdots\\
	      \widetilde{\eta}_{i,T}\\
	    \end{bmatrix}.\nonumber
	  \end{align}
	  Here, $\ve{1}_{T-p}$ denotes a $(T-p)\times 1$ vector of ones. The time-varying variance is then given by:
	  \begin{equation}
	   \var(\widetilde{\ve{\eta}}_i^{QD_2})=\widetilde{\ve{\Sigma}}_i^{QD_2}=\begin{bmatrix}
	    \tilde\sigma_{i,p+1}^2 & 0 & \cdots & 0\\
	    0 & \tilde\sigma_{i,p+2}^2 & \cdots & 0\\
	    \vdots & \vdots & \ddots & \vdots\\
	    0 & 0 & \cdots & \tilde\sigma_{i,T}^2\\
	  \end{bmatrix}.\nonumber
	  \end{equation}
    \item \textbf{Transform to the quasi-difference form (observations $t<p$).}~~Define the $p$-dimensional state vector:
    \[
      \ve{\varepsilon}_{i,t}=
      \begin{bmatrix}
	\varepsilon_{i,t}\\
	\varepsilon_{i,t-1}\\
	\vdots\\
	\varepsilon_{i,t-p+1}\\
      \end{bmatrix}
    \]
    For the elements of $\ve{\varepsilon}_{i,t}$ with indices less than 1 (i.e., unobserved components when $t<1$), missing values are replaced by zeros. Under this definition, the AR($p$) process in Equation (\ref{eq2}) can be rewritten as a VAR(1) representation:
    \[
      \ve{\varepsilon}_{i,t}=\ve{A}_i\ve{\varepsilon}_{i,t-1}+\ve{e}_{i,t}
    \]
    where $\ve{A}_i$ and $\ve{e}_{i,t}$ are defined as:
    \begin{align}
      \ve{A}_i=&
      \begin{bmatrix}
	\rho_{i,1} & \rho_{i,2} & \cdots & \rho_{i,p-1} & \rho_{i,p}\\
	1 & 0 & \cdots & 0 & 0\\
	0 & 1 & \cdots & 0 & 0\\
	\vdots & \vdots & \ddots & \vdots & \vdots \\
	0 & 0 & \cdots & 1 & 0\\
      \end{bmatrix},\quad
      \ve{e}_{i,t}=
      \begin{bmatrix}
	\eta_{i,t}\\
	0\\
	0\\
	\vdots\\
	0\\
      \end{bmatrix},\nonumber\\
      \ve{e}_{i,t}|h_{i,t}\sim& \mathcal{N}(\ve{0},\ve{Q}_{i,t}),\quad 
      \ve{Q}_{i,t}=
      \begin{bmatrix}
	\sigma_{i,t}^2 & 0 & \cdots & 0 & 0\\
	0 & 0 & \cdots & 0 & 0\\
	0 & 0 & \cdots & 0 & 0\\
	\vdots & \vdots & \ddots & \vdots & \vdots \\
	0 & 0 & \cdots & 0 & 0\\
      \end{bmatrix}\nonumber
    \end{align}
    If $\ve{A}_i$ is stable (i.e., its roots lie outside the unit circle), then under time-varying variance, the state covariance satisfies:
    \[
      \ve{\Gamma}_{i,t}=\ve{A}_i\ve{\Gamma}_{i,t-1}\ve{A}_i^\prime + \ve{Q}_{i,t}
    \]
    where $\ve{\Gamma}_{i,t}=\var(\ve{\varepsilon}_{i,t}|h_{i,1},\ldots,h_{i,t})$. Setting the initial value as $\ve{\Gamma}_{i,0}=\ve{0}$, we compute:
    \[
      \widehat{\ve{\Gamma}}_{i,t}=\widehat{\ve{A}}_i\widehat{\ve{\Gamma}}_{i,t-1}\widehat{\ve{A}}_i^\prime + \widehat{\ve{Q}}_{i,t} \ \ (t=1,\ldots,p). 
    \]
    This yields estimates of the conditional variance-covariance matrices for the initial $p$ periods, $\widehat{\ve{\Gamma}}_{i,1},\ldots,\widehat{\ve{\Gamma}}_{i,p}$. We then construct the quasi-differenced series for $t=1,\ldots,p$:
    \begin{align}
      \widetilde{R}_{i,t}^{QD_1} &= \ve{e}_{1}^\prime\,\left(\widehat{\ve{Q}}^{QD_1}_{i,t}\right)^{1/2}\,\widehat{\ve{\Gamma}}_{i,t}^{-1/2}\,\ve{R}_{i,t}, 
      & \widetilde{C}_{i,t}^{QD_1} &= \ve{e}_{1}^\prime\,\left(\widehat{\ve{Q}}^{QD_1}_{i,t}\right)^{1/2}\,\widehat{\ve{\Gamma}}_{i,t}^{-1/2}\,\ve{\iota}_{p}, \nonumber\\
      \widetilde{R}_{m,t}^{QD_1} &= \ve{e}_{1}^\prime\,\left(\widehat{\ve{Q}}^{QD_1}_{i,t}\right)^{1/2}\,\widehat{\ve{\Gamma}}_{i,t}^{-1/2}\,\ve{R}_{m,t}, 
      & \widetilde{\eta}_{i,t}^{QD_1} &= \ve{e}_{1}^\prime\,\left(\widehat{\ve{Q}}^{QD_1}_{i,t}\right)^{1/2}\,\widehat{\ve{\Gamma}}_{i,t}^{-1/2}\,\widehat{\ve{\varepsilon}}_{i,t} \nonumber
    \end{align}
where $\ve{R}_{i,t}=(R_{i,t},\ldots,R_{i,t-p+1})^\prime$, $\ve{\iota}_{p}=(1,\ldots,1)^\prime$, $\ve{R}_{m,t}=(R_{m,t},\ldots,R_{m,t-p+1})^\prime$, $\widehat{\ve{\varepsilon}}_{i,t}=(\widehat{\varepsilon}_{i,t},\ldots,\widehat{\varepsilon}_{i,t-p+1})^\prime$, and $\ve{e}_1=(1,0,\ldots,0)^\prime$. All of these are $p\times 1$ vectors; $\left(\widehat{\ve{Q}}^{QD_1}_{i,t}\right)^{1/2}$ and $\widehat{\ve{\Gamma}}_{i,t}^{-1/2}$ are $p\times p$ matrices. For $t=1$, set:
    \[
      \widetilde{R}_{i,1}^{QD_1}=R_{i,1},\quad \widetilde{C}_{i,1}^{QD_1}=1,\quad \widetilde{R}_{m,1}^{QD_1}=R_{M,1},\quad \widetilde{\eta}_{i,1}^{QD_1}=\widehat\varepsilon_{i,1}.
    \]
    Missing lagged values are filled with zeros. The completed vectors for $t=1,\ldots,p$ are then stacked, producing
    \[
      \widetilde{\ve{R}}_{i}^{QD_1}=
      \begin{bmatrix}
	\widetilde{R}_{i,1}^{QD_1}\\
	\widetilde{R}_{i,2}^{QD_1}\\
	\vdots\\
	\widetilde{R}_{i,p}^{QD_1}\\
      \end{bmatrix},\quad
      \widetilde{\ve{C}}_{i}^{QD_1}=
      \begin{bmatrix}
	\widetilde{C}_{i,1}^{QD_1}\\
	\widetilde{C}_{i,2}^{QD_1}\\
	\vdots\\
	\widetilde{C}_{i,p}^{QD_1}\\
      \end{bmatrix},\quad
      \widetilde{\ve{R}}_{m}^{QD_1}=
      \begin{bmatrix}
	\widetilde{R}_{m,1}^{QD_1}\\
	\widetilde{R}_{m,2}^{QD_1}\\
	\vdots\\
	\widetilde{R}_{m,p}^{QD_1}\\
      \end{bmatrix},\quad
      \widetilde{\ve{\eta}}_{i}^{QD_1}=
      \begin{bmatrix}
	\widetilde{\eta}_{i,1}^{QD_1}\\
	\widetilde{\eta}_{i,2}^{QD_1}\\
	\vdots\\
	\widetilde{\eta}_{i,p}^{QD_1}\\
      \end{bmatrix}.
    \]
    We then obtain the variance of $\widehat{\ve{\eta}}_{i}^{QD_1}$ as follows:
    \[
      \var(\widetilde{\ve{\eta}}_{i}^{QD_1})=\widetilde{\ve{\Sigma}}_{i}^{QD_1}=
      \begin{bmatrix}
	\tilde\sigma_{i,1}^2 & 0 & \cdots & 0 \\
	0 & \tilde\sigma_{i,2}^2 & \cdots & 0 \\
	\vdots & \vdots & \ddots & \vdots \\
	0 & 0 & \cdots & \tilde\sigma_{i,p}^2 \\
      \end{bmatrix}.
    \]
   \item \textbf{Estimate model parameters via GLS.}~~Now, stack the $QD_1$ and $QD_2$ variables:
  \[
    \widetilde{\ve{R}}_{i}^{QD}=
    \begin{bmatrix}
      \widetilde{\ve{R}}_{i}^{QD_1}\\
      \widetilde{\ve{R}}_{i}^{QD_2}\\
    \end{bmatrix},\quad
    \widetilde{\ve{C}}_{i}^{QD}=
    \begin{bmatrix}
      \widetilde{\ve{C}}_{i}^{QD_1}\\
      \widetilde{\ve{C}}_{i}^{QD_2}\\
    \end{bmatrix},\quad
    \widetilde{\ve{R}}_{m}^{QD}=
    \begin{bmatrix}
      \widetilde{\ve{R}}_{m}^{QD_1}\\
      \widetilde{\ve{R}}_{m}^{QD_2}\\
    \end{bmatrix},\quad
    \widetilde{\ve{\eta}}_{i}^{QD}=
    \begin{bmatrix}
      \widetilde{\ve{\eta}}_{i}^{QD_1}\\
      \widetilde{\ve{\eta}}_{i}^{QD_2}\\
    \end{bmatrix}.
  \]
  Equation (\ref{eq1}) can then be expressed as:
  \begin{equation}
    \widetilde{\ve{R}}_{i}^{QD}=\widetilde{\ve{C}}_{i}^{QD}\alpha_i+\widetilde{\ve{R}}_{m}^{QD}\beta_i+\widetilde{\ve{\eta}}_i^{QD}=\underbrace{\begin{bmatrix}
	\widetilde{\ve{C}}_{i}^{QD} & \widetilde{\ve{R}}_{m}^{QD}\\
    \end{bmatrix}}_{=\widetilde{\ve{Z}}_i}
    \underbrace{\begin{bmatrix}
	\alpha_i\\
	\beta_i\\
    \end{bmatrix}}_{=\ve{\kappa}_i}
    + \widetilde{\ve{\eta}}_i^{QD}\label{eq6}
  \end{equation}
  Thus, the GLS estimator of $\ve{\kappa}_i$ is:
  \[
    \widehat{\ve{\kappa}}_i^{GLS}=
    \begin{bmatrix}
      \widehat{\alpha}_i^{GLS}\\
      \widehat{\beta}_i^{GLS}\\
    \end{bmatrix}
    =
    \left(\widetilde{\ve{Z}}_i^\prime \left(\widetilde{\ve{\Sigma}}_{i}^{QD}\right)^{-1}\widetilde{\ve{Z}}_i\right)^{-1}\widetilde{\ve{Z}}_i^\prime \widetilde{\ve{\Sigma}}_{i}^{QD-1}\widetilde{\ve{R}}_{i}^{QD}
  \]
  where
  \[
    \widetilde{\ve{\Sigma}}_{i}^{QD}=
    \begin{bmatrix}
      \widetilde{\ve{\Sigma}}_{i}^{QD_1} & \ve{0}\\
      \ve{0} & \widetilde{\ve{\Sigma}}_{i}^{QD_2}\\
    \end{bmatrix}
  \]
  This multi-step procedure ensures consistent and efficient estimation in the presence of both serial correlation and time-varying volatility. The resulting parameter estimates form the basis for the subsequent analysis of abnormal returns and cumulative abnormal returns (CAR) in the event study framework.
\end{enumerate}

\subsection{Event Study Design}

To assess the impact of wartime economic controls on stock market efficiency, we implement an event study based on the abnormal returns estimated from the CAPM-AR($p$)-SV model. Our event study is structured according to standard practice (see \citet{mackinlay1997ese}), partitioning the time series into an estimation window, an event window, and a post-event window. The event date is denoted by $t=0$. The estimation window ($t=\tau_0$ to $t=\tau_1$) is set to 120 trading days prior to the event, which provides sufficient data for model parameter estimation. The event window ($t=\tau_1+1$ to $t=\tau_2$) covers the period from 10 days before to 10 days after the event ($[-10,10]$), allowing us to detect both immediate and delayed market responses. The post-event window ($t=\tau_2+1$ to $t=\tau_3$) is used to assess any persistent effects.

For each portfolio $i$ and day $t$, the abnormal return (AR) is defined as:
\[
 AR_{i,t} = R_{i,t} - \ex[R_{i,t} | \ve{\Phi}_t],
\]
where $R_{i,t}$ is the observed return, and $\ex[R_{i,t} | \ve{\Phi}_t]$ is the expected normal return conditional on the information set $\ve{\Phi}_t$ computed as:
\[
 \ex[R_{i,t} | \ve{\Phi}_t]=\hat\alpha_i^{GLS}+R_{m,t}\hat\beta_i^{GLS}=\ve{z}^\prime_t\hat{\ve{\kappa}}_i^{GLS}
\]
or equivalently,
\begin{align}
 AR_{i,t}&=R_{i,t}-(\hat{\alpha}_i^{GLS}+R_{m,t}\hat{\beta}_i^{GLS})=R_{i,t}-\ve{z}_t^\prime\hat{\ve{\kappa}}_i^{GLS}\nonumber\\
	 &=\ve{z}_t^\prime\ve{\kappa}_i+\varepsilon_{i,t}-\ve{z}_t^\prime\hat{\ve{\kappa}}_i^{GLS}=\varepsilon_{i,t}-\ve{z}_t^\prime(\hat{\ve{\kappa}}_i^{GLS}-\ve{\kappa}_i),\nonumber
\end{align}
where $\ve{z}_t=(1,R_{m,t})^\prime$.

The variance of the abnormal return accounts for both the time-varying volatility and parameter uncertainty: 
\begin{align}
 \var(\widehat{AR}_{i,t}|\ve{\Phi}_t)&=\var(\varepsilon_{i,t}|\ve{\Phi}_t)+\var(\ve{z}_t^\prime\hat{\ve{\kappa}}_i^{GLS}|\ve{\Phi}_t)\nonumber\\
 &=\tilde{\sigma}_{i,t}^2+\ve{z}_t^\prime\left(\widetilde{\ve{Z}}^\prime \left(\widetilde{\ve{\Sigma}}^{QD}\right)^{-1} \widetilde{\ve{Z}}\right)^{-1}\ve{z}_t,\nonumber
\end{align}
where $\tilde\sigma_{i,t}^2$ is the estimated conditional variance, and $\tilde{\ve{Z}}$ is the transformed design matrix. The ($1-\alpha$)\% confidence interval for $\widehat{AR}_{i,t}$ is then:
\[
 \widehat{AR}_{i,t} \pm z_{1-\frac{\alpha}{2}}\sqrt{\widetilde{\sigma}_{i,t}^2 + \ve{z}_t^\prime\left(\widetilde{\ve{Z}}^\prime \left(\widetilde{\ve{\Sigma}}^{QD}\right)^{-1} \widetilde{\ve{Z}}\right)^{-1}\ve{z}_t},
\]
where where $z_{1-\frac{\alpha}{2}}$ is the corresponding quantile of the standard normal distribution.

The cumulative abnormal return (CAR) for portfolio $i$ from $\tau_1$ to $\tau_2$ is defined as:
\[
 \widehat{CAR}_{i}[\tau_1:\tau_2]=\sum_{u=\tau_1}^{\tau_2} \widehat{AR}_{i,u}.
\]
The variance of the CAR is 
{\footnotesize{\begin{align}
 \var(\widehat{CAR}_{i}[\tau_1:\tau_2]| \ve{\Phi}_t)&=\sum_{u=\tau_1}^{\tau_2}\sum_{v=\tau_1}^{\tau_2}\cov(\widehat{AR}_{i,u}, \widehat{AR}_{i,v}| \ve{\Phi}_t)\nonumber\\
 &=\sum_{t=\tau_1}^{\tau_2}\left[\tilde\sigma_{i,t}^2 + \ve{z}_t^\prime\left(\widetilde{\ve{Z}}^\prime \left(\widetilde{\ve{\Sigma}}^{QD}\right)^{-1} \widetilde{\ve{Z}}\right)^{-1}\ve{z}_t\right] + 2\sum_{\tau_1\leq u < v \leq \tau_2}\ve{z}_u^\prime \left(\widetilde{\ve{Z}}^\prime \left(\widetilde{\ve{\Sigma}}^{QD}\right)^{-1} \widetilde{\ve{Z}}\right)^{-1}\ve{z}_v,\nonumber
\end{align}}}
where
\[
 \cov(\widehat{AR}_{i,u}, \widehat{AR}_{i,v}| \ve{\Phi}_t)=
 \begin{cases}
  \tilde\sigma_{i,u}^2 + \ve{z}_u^\prime \left(\widetilde{\ve{Z}}^\prime \left(\widetilde{\ve{\Sigma}}^{QD}\right)^{-1} \widetilde{\ve{Z}}\right)^{-1}\ve{z}_u & u=v,\\
  \ve{z}_u^\prime \left(\widetilde{\ve{Z}}^\prime \left(\widetilde{\ve{\Sigma}}^{QD}\right)^{-1} \widetilde{\ve{Z}}\right)^{-1}\ve{z}_v & u\neq v.\\
 \end{cases}
\]
The $(1-\alpha)$\% confidence interval for the CAR is:
\[
 \widehat{CAR}_{i}[\tau_1:\tau_2] \pm z_{1-\frac{\alpha}{2}}\sqrt{\var(\widehat{CAR}_{i}[\tau_1:\tau_2]| \ve{\Phi}_t)},
\]
where $z_{1-\alpha/2}$ is the corresponding quantile from the standard normal distribution.

For hypothesis testing, we use the standardized cumulative abnormal return (SCAR) statistic:
\[
 \widehat{SCAR}_{i}[\tau_1:\tau_2]=\frac{\widehat{CAR}_{i}[\tau_1:\tau_2]}{\sqrt{\var(\widehat{CAR}_{i}[\tau_1:\tau_2]| \ve{\Phi}_t)}},
\]
which is approximately standard normally distributed under the null hypothesis of no abnormal performance.

%% file: zaibatsu_data.tex
\section{Data}\label{zaibatsu_sec5}

This study constructs three capitalization-weighted market indices from daily closing prices of short-term clearing futures transactions on the Tokyo Stock Exchange between January 4, 1930 and August 30, 1943. The first is the Price Index (PI), based on observed stock prices. The second is the Adjusted Price Index (API), which adjusts the PI for ex-rights effects associated with capital increases and related corporate actions.\footnote{Such related corporate actions include additional payments and the allocation of shares in other companies associated with capital increases (for details, see \citet{saito2016pps}).} The third is the Total Return Index (TRI), which further incorporates dividends and thus measures total returns rather than price changes alone. All daily price data, specifically morning-session closing prices, are taken from the {\it Chugai Shogyo Shimpo}, the most reliable primary source on Japanese financial markets for this period.\footnote{The construction of the API and TRI requires information on the number of issued shares, additional payments, and dividends, which we obtain from the Japan Stock Yearbook ({\it Kabushiki Kaisha Nenkan}) and the Semi-Annual Sales Report ({\it Koka-jo}).} By contrast, \citet{bassino2015iet} use daily data from the {\it Toyo Keizai Shimpo} to examine weak-form market efficiency in prewar Japanese stock and bond markets. Their dataset, however, has two serious limitations.

First, their daily indices are simple arithmetic averages rather than capitalization-weighted aggregates. As a result, they cannot track changes in aggregate market value accurately, since firms with very different market capitalizations receive equal weight. This departs from standard index-construction practice and is likely to introduce bias, especially when large firms dominate trading and capitalization. Second, their series lack continuity. The industrial and utilities indices combine two distinct series based on inconsistent benchmarks: a simple arithmetic average up to December 1932 and a rebased index, with December 10, 1931 set to 100, from January 1933 onward. This discontinuity creates a structural break and makes the series unsuitable for rigorous time-series analysis or for tracing changes in market efficiency over time.

To address these shortcomings, we reconstruct three consistent capitalization-weighted index series from daily transaction-level data. The sample consists of the listed issues reported in Table \ref{age_war_tab2}, counting old and new shares separately when they were traded as distinct securities. Firms are classified along two dimensions: Zaibatsu affiliation and military classification under the Temporary Funds Adjustment Law. Specifically, firms in Category A industries, which received the highest priority in wartime capital allocation, are classified as military, and the remainder as non-military. On this basis, we construct the market portfolio and four capitalization-weighted sub-portfolios: (1) Zaibatsu (military), (2) Zaibatsu (non-military), (3) non-Zaibatsu (military), and (4) non-Zaibatsu (non-military).

\begin{center}
 (Table \ref{age_war_tab2} around here)
\end{center}

Table \ref{age_war_tab2} reports descriptive statistics for the listed issues in our sample. The variable ``listed issue'' identifies each security and distinguishes, where relevant, between old ({\it Kyu-kabu}) and new ({\it Shin-kabu}) shares listed on the Tokyo Stock Exchange. ``Zaibatsu'' indicates affiliation with one of the major conglomerates, while ``Military'' indicates whether the firm belonged to a Category A industry under the Temporary Funds Adjustment Law. These classifications allow us to capture both institutional heterogeneity, namely Zaibatsu versus non-Zaibatsu, and policy-relevant industrial heterogeneity, namely military versus non-military, within a unified empirical framework.

These distinctions matter for two reasons. First, Japan's wartime economy was institutionally segmented. Zaibatsu groups maintained internal capital markets through affiliated banks and trading companies, and therefore enjoyed privileged access to finance, procurement, and information. Such advantages likely affected both expected cash flows and risk, making their price dynamics not directly comparable to those of independent firms. Second, the military classification is not based on a broad notion of war-related activity, but on the state's formal prioritization of industries under the Temporary Law for Funds Adjustment. This legal and administrative classification captures the uneven allocation of funds, materials, contracts, and other policy support across firms and industries. It thus provides a sharper basis for examining how market efficiency and abnormal returns varied across firms differentially exposed to wartime industrial policy.

\begin{center}
 (Figure \ref{age_war_fig1} around here)
\end{center}

Figure~\ref{age_war_fig1} presents two complementary views of Japan's equity market during 1930--1943. The top panel plots the market-wide PI, API, and TRI, each normalized to 100 on January 4, 1930. The bottom panel shows each portfolio's share of total market capitalization. The top panel indicates that the three market indices move closely together, but the gaps between them also show that ex-rights adjustments and dividend payments matter cumulatively for measuring wartime stock market performance. The bottom panel highlights a marked shift in market composition. After 1937, the zaibatsu military share of total market capitalization rises sharply and exceeds 60\% by 1943, whereas the non-zaibatsu military share falls to a negligible level.

This shift points to an important structural transformation in the wartime equity market. Market capitalization did not simply reflect per-unit stock-price performance. It also reflected firms' ability to expand scale under wartime controls. As state-led capital mobilization, procurement priorities, and preferential access to finance and materials increasingly favored zaibatsu military firms, these firms were able to translate valuations into economic scale more effectively than independent military firms.

This is precisely the capitalization mechanism formalized in Proposition~\ref{prop:cap_concentration}: per-unit valuations and market-capitalization shares can diverge when zaibatsu military firms face lower financing wedges and can convert market valuations into installed capital more effectively.

\begin{center}
 (Figure \ref{age_war_fig2} around here)
\end{center}

Figure \ref{age_war_fig2} decomposes the market indices by portfolio and plots the PI, API, and TRI for each of the four firm groups. Relative performance differed markedly across portfolios over time. In particular, non-Zaibatsu military firms experienced a pronounced relative rise in the late 1930s, whereas Zaibatsu non-military firms tended to lag behind. At the same time, the comparison of PI, API, and TRI within each panel shows that corporate actions and dividend payments are quantitatively important for evaluating wartime stock performance and should therefore not be ignored in the empirical analysis.

Our empirical analysis also requires a daily risk-free rate in order to compute risk premiums for the market portfolio and the four sub-portfolios, since all expected returns in the CAPM-AR($p$)-SV framework are defined relative to this benchmark. We use the Tokyo overnight call rate, obtained from the {\it Chugai Shogyo Shimpo}, as the daily risk-free rate.

\begin{center}
 (Table \ref{age_war_tab3} around here)
\end{center}

Table \ref{age_war_tab3} reports descriptive statistics and unit root test results for the daily excess return series computed from the PI, API, and TRI for the capitalization-weighted market portfolio and the four sub-portfolios. Mean daily excess returns are negative across all portfolios and all three return definitions, reflecting the adverse market environment and heightened uncertainty of the wartime period.

Across the three index definitions, the market portfolio has the lowest standard deviation and is therefore the most stable. Among the four sub-portfolios, the Zaibatsu (military) portfolio consistently has the highest standard deviation, followed by the Zaibatsu (non-military) and non-Zaibatsu (non-military) portfolios, while the non-Zaibatsu (military) portfolio has the lowest volatility among the sub-portfolios. The data therefore do not support the view that independent military firms were the most volatile. Instead, the greatest return variability is observed among Zaibatsu-affiliated military firms.

In the event-study framework, all variables entering the moment conditions must be stationary for inference to be valid. To verify this requirement, we apply the augmented Dickey--Fuller generalized least squares (ADF--GLS) test of \citet{elliott1996eta}, which is relatively robust to size distortions (see \citet{ng2001lls}). The ADF--GLS results reject the null of a unit root at the 5\% level for all variables, confirming that the series used in the event study are stationary. This supports the reliability of the abnormal-return estimates and the subsequent statistical inference.

%% file: zaibatsu_empirical.tex
\section{Empirical Results}\label{zaibatsu_sec6}

\subsection{Preliminary Results}
\begin{center}
 (Tables \ref{age_war_tab4} to \ref{age_war_tab6} around here)
\end{center}

Tables \ref{age_war_tab4} to \ref{age_war_tab6} report GLS estimates of the intercepts and market betas from the CAPM-AR($p$)-SV model for four capitalization-weighted portfolios sorted by zaibatsu affiliation and industrial orientation, using the PI, API, and TRI, respectively.\footnote{As shown in the \href{https://at-noda.com/appendix/zaibatsu_appendix.pdf}{Online Appendix} A.1, OLS estimates with HAC standard errors can yield materially different inference. This finding is consistent with the concern emphasized in the Introduction: wartime daily returns often violate homoskedasticity and serial independence, so inference may be distorted when these features are addressed only imperfectly ex post.} The model relaxes two assumptions that are especially fragile in historical wartime daily data, namely homoskedastic innovations and serially independent residuals. Specifically, stochastic volatility (SV) captures time-varying uncertainty, while the AR($p$) component accommodates residual persistence arising from market frictions and non-synchronous trading. To address possible parameter instability, we estimate the model for the full sample (Jan 4, 1930--Aug 31, 1943) and for two subsamples split at the Marco Polo Bridge Incident (Jul 7, 1937).

A first result concerns residual dependence. Across all three index definitions, the Bayesian Information Criterion uniformly selects $p=1$ for all portfolios in the full sample and in both subsamples. Once time-varying volatility is taken into account, a first-order residual process therefore appears sufficient to capture residual dependence in these data.

Turning to the alpha estimates, the full-sample results reveal a clear but selective pattern. At the 1\% level, the zaibatsu military portfolio exhibits a positive alpha under the PI, API, and TRI, whereas the non-zaibatsu military portfolio exhibits a negative alpha under all three index definitions. By contrast, the alpha for the zaibatsu non-military portfolio is not statistically distinguishable from zero in any specification. The alpha for the non-zaibatsu non-military portfolio is insignificant under the PI, but significantly positive under the API and TRI. Taken together, these results suggest that the most robust full-sample contrast is between the relatively strong performance of the zaibatsu military portfolio and the persistent underperformance of the non-zaibatsu military portfolio.

In the prewar subsample, the evidence for abnormal performance is much weaker. At the 1\% level, most alphas are not statistically different from zero across the three index definitions. The principal exception is the non-zaibatsu military portfolio, whose alpha is significantly negative under the API and TRI, though not under the PI. This pattern suggests that before Jul 7, 1937, cross-sectional differences in abnormal performance were limited, and that only the independent military segment displays a persistently weak pattern once capital adjustments and dividends are taken into account.

The wartime subsample yields a much sharper contrast. Under the PI, all four portfolios have alphas that are statistically different from zero at the 1\% level: both zaibatsu portfolios and the non-zaibatsu non-military portfolio are positive, whereas the non-zaibatsu military portfolio is strongly negative. Under the API and TRI, however, statistical significance is concentrated in the non-zaibatsu portfolios. In both cases, the non-zaibatsu military portfolio remains significantly negative and the non-zaibatsu non-military portfolio significantly positive, while the alphas of both zaibatsu portfolios are not statistically different from zero. The most robust wartime result across the adjusted-price and total-return measures is therefore a pronounced polarization within the non-zaibatsu segment, with independent military firms underperforming sharply and independent non-military firms outperforming strongly.

The beta estimates reinforce this interpretation. In the full sample, the beta of the zaibatsu military portfolio exceeds unity under all three index definitions, whereas the beta of the non-zaibatsu military portfolio is consistently below unity. In the wartime subsample, this contrast becomes even sharper: the beta of the non-zaibatsu military portfolio falls to values ranging from approximately 0.50 to 0.57, whereas the beta of the non-zaibatsu non-military portfolio rises to values between 1.17 and 1.36. The zaibatsu portfolios lie between these two extremes. The beta of the zaibatsu non-military portfolio exceeds unity during the wartime period under all three index definitions, whereas the beta of the zaibatsu military portfolio remains above unity under the PI but falls slightly below unity under the API and TRI. These estimates suggest that wartime cross-sectional differences reflected not only abnormal performance but also marked differences in systematic market exposure.

The adjusted $R^2$ values further clarify how market comovement changed across regimes. In the prewar period, the market factor explains a smaller share of return variation for the zaibatsu portfolios than for the non-zaibatsu portfolios across all three index definitions. During wartime, however, the fit rises sharply for the zaibatsu military portfolio, reaching 0.89 under the PI and approximately 0.72 under both the API and TRI, while it falls markedly for the non-zaibatsu military portfolio to values between 0.43 and 0.51. This contrast suggests that wartime conditions strengthened the linkage between aggregate market movements and returns on zaibatsu military firms, while independent military firms became relatively less synchronized with the market and more exposed to portfolio-specific shocks.

\subsection{Event Study Results}\label{age_war_sec5-2}
\begin{center}
(Tables \ref{age_war_tab7} to \ref{age_war_tab9} about here)
\end{center}

Tables \ref{age_war_tab7} to \ref{age_war_tab9} report day-0 abnormal returns ($AR$ at $t=0$) and cumulative measures ($CAR$ and standardized $CAR$, $SCAR$) for four capitalization-weighted portfolios, namely zaibatsu military, zaibatsu non-military, non-zaibatsu military, and non-zaibatsu non-military, under three alternative index definitions, PI, API, and TRI. The tables report only the sign and statistical significance of each measure, where ``{\color{red}{$\boldsymbol{+}$}}'', ``{\color{blue}{$\boldsymbol{-}$}}'', and ``0'' denote significantly positive, significantly negative, and insignificant reactions, respectively.\footnote{All supplementary figures and robustness checks are presented in the \href{https://at-noda.com/appendix/zaibatsu_appendix.pdf}{Online Appendix} A.3 to A.11.} A significant day-0 $AR$ is consistent with rapid incorporation of public news. By contrast, significant $CAR$ or $SCAR$ at $t=0$ indicates that abnormal returns had accumulated by the event date, but does not by itself distinguish post-announcement drift from pre-event information leakage or other shocks before the event date.

Taken together, the three tables do not support a uniform view of wartime Japan as either semi-strong efficient or semi-strong inefficient. Rather, they suggest a segmented market in which the price impact of public news varied with institutional position. Many major events, including several wartime regulations after 1937, generate no significant reaction in any portfolio under any index definition. Yet when significant responses do appear, they often differ across zaibatsu and non-zaibatsu firms and across military and non-military firms. Cumulative abnormal returns are present in selected episodes, but they are not pervasive.

Some common patterns emerge in the early 1930s. Following the assassination attempt on Prime Minister Hamaguchi on November 14, 1930, the zaibatsu military portfolio shows no significant reaction, whereas the zaibatsu non-military portfolio shows a significantly positive $AR$ only, the non-zaibatsu military portfolio shows significantly negative $AR$, $CAR$, and $SCAR$, and the non-zaibatsu non-military portfolio shows significantly positive reactions on all three measures. The March Incident on March 20, 1931 produces no day-0 effect, but the zaibatsu military portfolio records significantly positive $CAR$ and $SCAR$. The Mukden Incident on September 18, 1931 yields a significantly negative $AR$ only for the zaibatsu non-military portfolio. Even before full wartime mobilization, political and external shocks were thus associated with different responses across firm groups.

The reimposition of the gold-export ban on December 13, 1931 provides one of the clearest common episodes across the three tables. The zaibatsu military portfolio shows significantly positive $AR$, $CAR$, and $SCAR$ under all three index definitions, while the zaibatsu non-military portfolio also shows positive cumulative reactions in all three tables, although its $AR$ is insignificant under TRI. By contrast, the non-zaibatsu military portfolio shows a significantly positive $AR$ without cumulative significance, and the non-zaibatsu non-military portfolio shows a significantly negative $AR$ without cumulative significance. Positive cumulative abnormal returns by the event date were therefore concentrated in the zaibatsu portfolios, especially the military one.

Other early-1930s events display a related form of cross-sectional segmentation. The First Shanghai Incident on January 28, 1932 is associated with significantly positive $AR$, $CAR$, and $SCAR$ for the non-zaibatsu military portfolio and significantly negative reactions on all three measures for the non-zaibatsu non-military portfolio, while the zaibatsu non-military portfolio shows only a negative $AR$. The Bank of Japan's underwriting of government deficit bonds on November 25, 1932 produces significantly positive $AR$, $CAR$, and $SCAR$ for the non-zaibatsu military portfolio and significantly negative reactions on all three measures for the non-zaibatsu non-military portfolio.

These episodes are naturally interpreted through the regime-risk channel in Proposition~\ref{prop:regime_insulation}. Broad political, diplomatic, or monetary shocks need not affect all portfolios symmetrically. Zaibatsu portfolios could be partially insulated by internal capital markets, diversified group claims, and preferential access to administrative allocation, whereas non-zaibatsu portfolios remained more exposed to external finance, policy uncertainty, and reallocation risk. This interpretation explains why several early-1930s regime shocks leave the zaibatsu portfolios close to zero while generating stronger, and sometimes opposite-signed, responses among the non-zaibatsu portfolios.

By contrast, the League of Nations' condemnation on February 24, 1933 and the Tanggu Truce on May 31, 1933 mainly affect the zaibatsu military portfolio, with the former yielding a negative $AR$ only and the latter yielding significantly positive $AR$, $CAR$, and $SCAR$. These events are better interpreted through the embedded-rent channel in Proposition~\ref{prop:embedded_rent}. Rather than operating mainly as broad shocks to external dependence, they appear to have revised the value of being embedded in the state-led strategic order, though the direction and persistence of the response differed across events. Such revisions can generate cumulative responses concentrated in the zaibatsu military portfolio, especially when public information is incorporated gradually.

Domestic political conflict in the mid-1930s again produces heterogeneous reactions, although the common findings are narrower than a superficial reading might suggest. During the Military Academy Incident on November 20, 1934, the zaibatsu military portfolio shows significantly positive $CAR$ and $SCAR$, but not a significant $AR$. The February 26 Incident on February 26, 1936 yields significantly positive $AR$, $CAR$, and $SCAR$ for the zaibatsu military portfolio and significantly negative $AR$, $CAR$, and $SCAR$ for the non-zaibatsu non-military portfolio in all three tables. However, some additional reactions are index-dependent: the zaibatsu non-military portfolio shows a positive $AR$ only under API, and the non-zaibatsu military portfolio shows a negative $AR$ only under API. Likewise, the effects of Japan's withdrawal from the Washington Naval Treaty on December 29, 1934 differ across PI, API, and TRI.

The Military Academy Incident provides another example of the embedded-rent mechanism in Proposition~\ref{prop:embedded_rent}. Although the day-0 $AR$ is insignificant, the significant cumulative response of the zaibatsu military portfolio is consistent with a revision in the continuation value of firms closely tied to the state-led military order. The February 26 Incident points in the same direction, but with a sharper cross-sectional contrast: the zaibatsu military portfolio was positively revalued, while the non-zaibatsu non-military portfolio was negatively revalued. These patterns suggest that domestic political shocks did not simply raise or lower market valuations uniformly; they changed the relative value of institutional proximity to the military and the state.

The transition from mobilization to total war does not produce a uniform strengthening of announcement effects. The Marco Polo Bridge Incident on July 7, 1937, the Temporary Funds Adjustment Law on September 8, 1937, and the National Mobilization Law on March 24, 1938 all show no significant reaction in any portfolio under any index definition. The Company Profit Dividends and Fund Accommodation Ordinance on April 1, 1939 also does not appear to have primarily benefited the zaibatsu portfolios. Instead, it produces a significantly negative $AR$ for the non-zaibatsu military portfolio and a significantly positive $AR$ for the non-zaibatsu non-military portfolio, with no cumulative significance.

The outbreak of World War I\hspace{-1.2pt}I on September 1, 1939 produces one of the sharpest response patterns in the tables. The zaibatsu military portfolio shows a significantly positive $AR$ under all three index definitions. The zaibatsu non-military portfolio also shows a significantly positive $AR$ in all three tables, although positive $CAR$ and $SCAR$ for that portfolio appear only under TRI. By contrast, the non-zaibatsu military portfolio shows significantly negative $AR$, $CAR$, and $SCAR$, whereas the non-zaibatsu non-military portfolio shows significantly positive reactions on all three measures. This episode is consistent with a marked divergence not only between military and non-military firms, but also between affiliated and independent firms.

In the 1940s, some of the clearest cumulative patterns again appear in a limited set of events. The Second Armistice at Compi\`egne on June 22, 1940 produces significantly negative $AR$, $CAR$, and $SCAR$ for the zaibatsu military portfolio and significantly positive $CAR$ and $SCAR$ for the non-zaibatsu non-military portfolio, with no significant $AR$ for the latter. The U.S. freezing of Japanese assets and the subsequent oil embargo on July 26, 1941 yields a significantly positive $AR$ for the non-zaibatsu military portfolio and significantly negative $AR$, $CAR$, and $SCAR$ for the non-zaibatsu non-military portfolio. This is the late-sample counterpart of the regime-risk mechanism in Proposition~\ref{prop:regime_insulation}: sanctions shocks altered exposure to external resource constraints and reallocation risk, generating stronger responses among non-zaibatsu portfolios than among zaibatsu portfolios.

The outbreak of the Pacific War on December 8, 1941 generates significantly positive $AR$, $CAR$, and $SCAR$ only for the non-zaibatsu non-military portfolio. The Ministry of Finance's pivot to stricter capital controls on November 15, 1942 produces a significantly negative $AR$ for the zaibatsu military portfolio and significantly positive $AR$ for the zaibatsu non-military and non-zaibatsu non-military portfolios. Finally, the Fall of Guadalcanal on February 9, 1943 yields significantly positive $CAR$ and $SCAR$ for the zaibatsu non-military portfolio in all three tables. This episode is better interpreted through the optional group-continuation state in the \href{https://at-noda.com/appendix/zaibatsu_appendix.pdf}{Online Appendix} than through the embedded-rent channel for zaibatsu military firms. The response is concentrated in the zaibatsu non-military portfolio, suggesting a revision in the value of group-level financial claims rather than a direct military-demand effect.

Some late-war responses, however, are sensitive to the choice of index definition. The Tripartite Pact on September 27, 1940 produces negative cumulative reactions for the non-zaibatsu military portfolio under all three definitions, but the corresponding positive response appears for the non-zaibatsu non-military portfolio only under PI and for the zaibatsu military portfolio only under TRI. Similarly, the Battle of Midway on June 5, 1942 produces negative cumulative reactions for the zaibatsu military portfolio under PI, no significant effect under API, and negative cumulative reactions for the non-zaibatsu military portfolio under TRI. Some inferences are therefore common across price definitions, whereas others depend materially on whether prices are measured by PI, API, or TRI.

Overall, Tables \ref{age_war_tab7} to \ref{age_war_tab9} do not support a blanket characterization of wartime Japan as semi-strong inefficient. Most events still produce no detectable cumulative abnormal returns by $t=0$, and several important wartime regulations are associated with no significant reaction at all. However, when abnormal returns do emerge, both their sign and their cumulative profile by the event date vary systematically with business-group affiliation and military orientation. The evidence therefore suggests that public news was not incorporated uniformly across firms, but instead was filtered through institutional position.

\subsection{Historical Interpretations}

This section interprets the empirical results in light of Japanese economic and business history. The model in Proposition~\ref{prop:cap_concentration} clarifies why the divergence between price performance and market-capitalization dominance is not contradictory. Price indices measure the value of one unit of installed capital, whereas market capitalization also depends on the scale of capital that firms can finance and install. Once zaibatsu military firms faced lower financing wedges and greater issuance capacity, they could come to dominate market capitalization even when some independent military firms exhibited favorable per-unit price performance. This distinction is crucial for interpreting the post-1937 shift in market structure.

Table \ref{age_war_tab3} provides the basic tone. Mean excess returns are negative across all portfolios under PI, API, and TRI, volatility differs substantially across portfolios, and the return series are stationary. Historically, these baseline results matter because they indicate that any wartime ``advantage'' identified below was a relative advantage within a generally adverse and unstable environment rather than evidence of uniformly favorable conditions for military firms as such.

Once Table \ref{age_war_tab2} is read together with Tables \ref{age_war_tab4} to \ref{age_war_tab6}, the four portfolios appear not as abstract cells in a 2 $\times$ 2 classification but as historically distinct firm groupings exposed to wartime controls from different organizational positions. The zaibatsu military portfolio consisted largely of firms in mining, shipping, heavy machinery, and metals; the non-zaibatsu military portfolio included strategically important firms in oil, electric power, and railway/colonial transport; and the two non-military portfolios combined different mixtures of chemicals, paper, food, fisheries, textiles, sugar, and beer. These differences matter because the wartime economy did not allocate finance, materials, and procurement uniformly across all firms that were formally ``military'' or ``non-military.'' The later return patterns therefore need to be interpreted as differences across historically distinct groups of firms rather than as simple contrasts between two homogeneous sectors.

The most robust contrast in Tables \ref{age_war_tab4} to \ref{age_war_tab6} is between the relative strength of the zaibatsu military portfolio and the persistent weakness of the non-zaibatsu military portfolio. In the full sample, the zaibatsu military portfolio has a positive alpha under PI, API, and TRI, whereas the non-zaibatsu military portfolio has a negative alpha under all three measures. In the prewar subsample, cross-sectional differences are still limited, but the non-zaibatsu military portfolio is already weak under API and TRI. In the wartime subsample, the divergence becomes much sharper. Under PI, both zaibatsu portfolios and the non-zaibatsu non-military portfolio have positive alphas, whereas the non-zaibatsu military portfolio is strongly negative. Under API and TRI, significance shifts away from the zaibatsu portfolios and toward a pronounced polarization within the non-zaibatsu segment: non-zaibatsu military remains significantly negative, while non-zaibatsu non-military becomes significantly positive. Historically, this suggests that formal inclusion in priority sectors did not by itself secure superior returns. What mattered was whether wartime priority could be converted into dependable access to credit, materials, and procurement (see \href{https://at-noda.com/appendix/zaibatsu_appendix.pdf}{Online Appendix} A.2). This interpretation is also consistent with \citetapos{fukumoto2026cms} finding that sanction-exposed sectors moved toward authoritarian alignment under material duress, whereas procurement-linked sectors did not display an equally straightforward pro-regime shift, suggesting that wartime advantage depended less on formal strategic designation alone than on the organizational capacity to convert state policy into durable economic security.

The beta estimates and adjusted $R^2$ values reinforce this interpretation. In the full sample, the beta of the zaibatsu military portfolio exceeds unity under all three index definitions, whereas the beta of the non-zaibatsu military portfolio remains below unity. In wartime, this divergence becomes more pronounced: the non-zaibatsu military beta falls to roughly 0.50--0.57, while the non-zaibatsu non-military beta rises well above unity. At the same time, adjusted $R^2$ rises sharply for the zaibatsu military portfolio during wartime but falls markedly for the non-zaibatsu military portfolio. A cautious historical interpretation is that zaibatsu military firms became more fully incorporated into the aggregate dynamics of wartime mobilization, whereas independent military firms became more peripheral to them. This reading is consistent with the industrial composition of the portfolios. Firms in mining, shipping, machinery, and metals were positioned nearer the center of wartime industrial coordination than firms whose strategic importance was real but whose organizational integration was weaker. The fact that BIC uniformly selects AR(1) across all samples also strengthens the point that these contrasts are not artifacts of neglected residual dependence.

The distinction among PI, API, and TRI is central to the historical interpretation of these results. Under wartime Japan, these are not simply alternative technical measures of the same process. PI captures observed market prices. API adjusts for rights issues and related capital actions. TRI further incorporates dividends and therefore comes closest to realized shareholder return. Under the Temporary Funds Adjustment regime and the 1939 Company Profit Dividends and Fund Accommodation Ordinance, capital increases, rights allocations, retained earnings, and payout restrictions were themselves part of the institutional structure through which wartime gains or losses reached shareholders. This helps explain why the relative strength of zaibatsu military firms is clearest in PI, whereas API and TRI bring out a stronger polarization within the non-zaibatsu segment. The former pattern is consistent with the capitalization of expected priority access, procurement, and future expansion. The latter suggests that price appreciation alone did not exhaust wartime valuation and that some non-military firms could preserve shareholder value through payout and capital-adjustment channels even when they were not dominant in price performance alone.

This point is especially important for understanding the non-zaibatsu non-military portfolio. Its positive wartime alpha under API and TRI, but not under PI, is not easy to reconcile with a simple narrative in which wartime regulation produced only military winners and civilian losers. A more cautious interpretation is that some firms outside the formally prioritized military sectors retained the capacity to preserve shareholder value through the regulated channels of dividends, retained earnings, and capital adjustment. Given the industrial composition of this portfolio---textiles, paper, food, sugar, beer, fisheries, and chemicals---that possibility is historically plausible. These were not industries outside the wartime economy; some remained tied to basic consumption, intermediate inputs, or quasi-essential demand. The result therefore suggests that wartime valuation was shaped not only by military priority but also by the corporate and regulatory channels through which returns were distributed.

Tables \ref{age_war_tab7} to \ref{age_war_tab9} reinforce these conclusions when the event-study evidence is read by pattern rather than chronology. Several events were broadly favorable to the zaibatsu military portfolio, especially in cumulative terms. The reimposition of the gold export ban in December 1931, a decisive break with gold-standard orthodoxy and a major turning point in Japan's reflationary policy regime, generated positive reactions for zaibatsu military across AR, CAR, and SCAR. The Tanggu Truce of 1933, the agreement that ended large-scale fighting after the Manchurian crisis and consolidated Japan's strategic position in North China, and the Military Academy Incident of 1934, an internal army purge that revealed the intensifying political influence of the military, were also associated with positive cumulative responses. Likewise, the February 26 Incident of 1936, an attempted coup by young army officers in Tokyo that exposed both regime instability and the growing weight of the military in politics, produced a broadly positive response for zaibatsu military while non-zaibatsu portfolios reacted negatively or weakly. These were not simply ``war events'' in a narrow sense. Rather, they were episodes in which investors could plausibly revise upward the value of firms already embedded in the core circuits of state-led industrial mobilization.

By contrast, many major wartime regulations after 1937 generated little immediate response across portfolios under any index definition. This includes the Temporary Funds Adjustment Law, which placed major corporate fundraising under administrative approval; the National Mobilization Law, the framework statute that greatly expanded state authority over labor, prices, and production; the Company Accounting Control Ordinance, which strengthened wartime control over accounting practices and internal funds; the Stock Price Control Ordinance, which sought to restrain equity-price movements under wartime conditions; and the Wartime Finance Bank Law, which established a policy-based financial institution to channel funds toward priority sectors. This does not diminish the historical importance of these measures. It suggests instead that the significance of wartime regulation lay not necessarily in the announcement of formal rules at a single date, but in the longer-run restructuring of firms' access to finance, materials, and procurement. The market, in other words, did not mechanically reprice firms every time a major control measure was enacted; it responded more selectively when events altered expectations about which firms would actually benefit from the controlled system (see \href{https://at-noda.com/appendix/zaibatsu_appendix.pdf}{Online Appendix} A.1).

The event-study evidence is equally important for understanding the persistent weakness of the non-zaibatsu military portfolio. Its reactions are not uniformly negative, but they are unstable and often fail to cumulate. It reacts negatively across all three measures to the Hamaguchi assassination attempt of November 1930, the attack on a prime minister closely associated with austerity and the return to gold, and negatively to the 1939 dividend-and-fund-accommodation ordinance, the wartime measure that tightened restrictions on corporate finance and payout. It also reacts negatively on a cumulative basis to the outbreak of World War I\hspace{-1.2pt}I in Europe in September 1939 and to the Tripartite Pact of September 1940, the formal alliance linking Japan with Germany and Italy. Even where it reacts positively---for example, around the reimposition of the gold export ban or the U.S. asset freeze and oil embargo of July 1941, the sanctions shock that sharply tightened Japan's external resource constraint---the response is often confined to AR and does not extend to CAR or SCAR. This pattern suggests that independent military firms were not priced as secure long-run beneficiaries of wartime mobilization. Strategic designation did not eliminate uncertainty over financing, materials, execution, or the terms on which wartime demand would be converted into earnings.

The broader point, then, is more specific than a general claim about wartime ``distortion.'' Wartime Japan did not generate a single premium on military production. It created a hierarchically structured field in which some firms---especially large zaibatsu-affiliated firms in mining, shipping, machinery, metals, and related strategic industries---could convert policy priority into credible claims on future earnings, while others, including many independent military firms, could not. The market was therefore pricing not only military demand, but also organizational position within a controlled economy.

\subsection{Wartime Reallocation, Group Finance, and Institutional Investors}

This section considers the concrete channels through which wartime controls translated into differential shareholder outcomes. The zaibatsu were not merely large firms, but diversified business groups that combined holding companies, affiliated banks, trading companies, and industrial subsidiaries. In such a structure, wartime gains did not accrue only at the factory level. They were transmitted through internal capital markets, intercorporate shareholding, dividend flows, and capital gains within the group. This organizational structure is essential for understanding why the wartime stock market could favor some firms more consistently than others.

The first channel was the reallocation of investment toward sectors most directly tied to military mobilization. As \citet{takeda2020ze} shows, the major zaibatsu had already built substantial positions in heavy and chemical industry before full wartime mobilization, but this orientation intensified sharply between 1937 and 1941. Mitsui's investment share in heavy and chemical industries rose from 22.1\% to 39.9\%, Mitsubishi's from 27.1\% to 36.5\%, and Sumitomo's from 35.2\% to 65.5\%. The Sino-Japanese War, which began in 1937, marked the transition to a far more intensive regime of military procurement and industrial coordination. Under these conditions, sectors such as mining, metals, machinery, shipbuilding, and aircraft production became central because they occupied the core of the wartime production system.

Yet the historical importance of these figures lies less in the fact of industrial expansion alone than in who could undertake it under wartime controls. The government relied heavily on large zaibatsu-affiliated firms to expand military production capacity, and these groups were better able to do so because they combined preexisting productive capacity with access to internal finance and administrative coordination. Mitsubishi Heavy Industries in aircraft and shipbuilding, Sumitomo Metal Industries in strategic materials, and other large group-affiliated firms were therefore not merely beneficiaries of military demand in the abstract. They were organizationally positioned to convert state priority into plant expansion, procurement contracts, and relatively stable production increases. In this respect, the historical evidence in this section complements the findings in Section~\ref{zaibatsu_sec6}: formal designation as a military-related firm did not by itself produce superior market performance; what mattered was whether that designation could be translated into effective access to capital, materials, and coordinated investment.

A second channel operated through the profit structure of the holding companies themselves. Wartime gains were not confined to operating profits at the subsidiary level. They also appeared as financial returns captured by the apex of the group. \citet{hamabuchi1981rmz} shows that, in Mitsubishi's holding company accounts for 1936--1941, dividends accounted for approximately 51\% of total revenue and capital gains from stock sales for roughly 44\%. Securities, primarily shares in affiliated firms, comprised more than 80\% of total assets. Zaibatsu profitability under wartime conditions was therefore not simply a story of factories producing more weapons. It was equally a story of group-level financial claims on subsidiary earnings and asset appreciation. Wartime expansion in heavy industry was monetized not only through industrial output but also through the ownership structure of the groups themselves.

This financial structure also helps explain why wartime controls did not necessarily reduce the importance of equity ownership. On the contrary, under dividend restrictions and capital controls, retained earnings and intercorporate shareholding became even more important mechanisms through which value was preserved and reallocated. \citet{takeda2020ze} notes that zaibatsu firms used retained earnings, strengthened by wartime payout restrictions, both to expand facilities during the war and to reserve resources for the postwar period. From the standpoint of market valuation, this means that wartime gains could reach shareholders through several channels at once: current dividends where permitted, rising expected future earnings, capital gains within group portfolios, and the accumulation of retained profits that strengthened future competitive position. This helps explain why the wartime pricing of zaibatsu firms cannot be reduced to a simple reaction to current military orders alone.

The contrast with non-zaibatsu and regional firms is equally revealing. Many such firms faced limited direct government support, shortages of raw materials, and tighter financing constraints. Under wartime controls, these disadvantages mattered not only because they restricted current output, but because they weakened firms' ability to convert military demand into durable earnings. This historical asymmetry provides one plausible background for the persistent weakness of the non-zaibatsu military portfolio documented above. The broader implication is that wartime Japan did not reward all firms inside military-priority sectors equally. It favored firms that combined formal priority with organizational capacity, financial flexibility, and established administrative standing.

A third channel operated through the changing role of institutional investors, especially life insurance companies. Life insurers were among the few institutions capable of supplying long-term funds on a large scale in a market that had previously been more heavily dominated by individual investors. From the late 1920s onward, life insurance companies expanded their shareholdings and became increasingly important as stable shareholders and absorbers of newly issued securities (\citet{asajima1991hhj,takeda2009ibl}). Firms such as Nippon Life and Dai-ichi Life became major shareholders in prominent corporations and contributed to market depth and stability. Their rise therefore represented a structural shift in the composition of demand in Japanese capital markets.

At the same time, life insurers were not homogeneous. Zaibatsu-affiliated insurers such as Mitsui Life and Sumitomo Life tended to pursue relatively conservative long-term strategies, whereas some smaller insurers adopted more aggressive investment practices in pursuit of higher yields (\citet{takeda2009ibl}). \citet{asajima1991hhj} notes that from 1933 onward, life insurers expanded their holdings of non-zaibatsu firms, suggesting that they partially supplemented a capital allocation system otherwise tilted toward large business groups. During wartime, however, insurers were also drawn more directly into the state's financing apparatus. They purchased government bonds, underwrote corporate bonds, and extended loans to military-related enterprises (\citet{shibata2011mcw}). Life insurers were thus not merely passive portfolio investors, but institutional conduits through which wartime finance was stabilized and directed.

The composition of insurer portfolios changed accordingly. Wartime portfolios became more heavily oriented toward government bonds, corporate bonds, and securities related to military industries. Despite low-interest policies and declining government bond yields, major insurers such as Nippon Life, Teikoku Life, Dai-ichi Life, and Meiji Life were able to maintain relatively stable nominal returns by combining bond income, loan interest, and stock dividends. By contrast, zaibatsu-affiliated insurers, including Mitsui, Sumitomo, and Asahi, tended to hold larger proportions of equities and therefore had greater exposure to unrealized gains. Some smaller insurers also achieved high returns through more aggressive securities investment in the late 1930s (\citet{takeda2009ibl}). At the same time, these nominally stable returns should not be overstated: \citet{hirayama2019pwj} points out that real returns deteriorated in the later wartime years because of inflation. Even so, life insurers played a crucial role as institutional investors and providers of long-term finance in a controlled wartime economy.

Taken together, these observations clarify why wartime shareholder outcomes cannot be understood solely through the binary contrast between military and civilian firms. The historical evidence points instead to a more differentiated structure. Zaibatsu groups benefited because they combined industrial presence in priority sectors with internal financial channels that allowed profits to be transmitted, retained, and redeployed within the group. Life insurers mattered because they provided stable demand for securities and long-term finance, while also becoming increasingly integrated into the machinery of wartime allocation. This perspective helps reconcile the results in the preceding subsections. The relative strength of zaibatsu military firms reflected not only direct exposure to military production, but also the institutional capacity to convert wartime controls into finance, investment, and shareholder value. Conversely, the weakness of many non-zaibatsu military firms suggests that formal strategic importance without comparable organizational backing was often insufficient. More broadly, wartime market valuation was shaped not only by the direction of military demand, but by the institutional channels through which gains could be secured, distributed, and sustained over time.

This historical evidence provides the institutional interpretation of the group-continuation state introduced in the \href{https://at-noda.com/appendix/zaibatsu_appendix.pdf}{Online Appendix}. The state should not be read as current military operating profit. Rather, it captures revisions in the value of group-level financial claims transmitted through holding-company dividends, intercorporate shareholding, retained earnings, and internal claims on subsidiary profits. This channel is especially useful for interpreting events such as the Fall of Guadalcanal, where the cumulative response is concentrated in the zaibatsu non-military portfolio rather than in the zaibatsu military portfolio.

%% file: zaibatsu_conclusion.tex
\section{Concluding Remarks}\label{zaibatsu_sec7}

This paper has examined how the Japanese stock market priced wartime controls in a hybrid regime in which administrative allocation expanded without fully displacing market trading. The main conclusion is not that the market ceased to function, nor that wartime Japan can be characterized in blanket terms as semi-strong inefficient. Rather, price formation remained operative, but the information incorporated into valuations was filtered through firms' institutional position, organizational structure, and access to state-controlled resources.

The portfolio results indicate a differentiated rather than uniform wartime reordering. In the full sample, the clearest contrast is between the positive alpha of the zaibatsu military portfolio and the negative alpha of the non-zaibatsu military portfolio. Cross-sectional differences are more limited before the war, but they widen during the wartime years. Under the PI, both zaibatsu portfolios and the non-zaibatsu non-military portfolio are positive, whereas the non-zaibatsu military portfolio is strongly negative. Under the API and TRI, the sharpest contrast appears within the non-zaibatsu segment, where independent military firms remain negative while independent non-military firms become positive. These patterns are difficult to reconcile with any simple claim that wartime Japan generated a uniform premium for either zaibatsu firms or military firms as such. Relative advantage instead depended on the interaction of sectoral position with financial and organizational capacity.

The model clarifies why this distinction matters. The rise of the zaibatsu military portfolio in market capitalization should not be interpreted solely as a price-index phenomenon. Market capitalization reflects both per-unit valuation and the scale of installed capital. If zaibatsu military firms faced lower financing wedges and could translate valuations into expansion more effectively, they could come to dominate market capitalization even when some independent military firms displayed favorable price performance. In this sense, wartime controls affected not only expected payoffs but also the conversion of valuations into economic scale.

The event-study results point in the same general direction, but they also require caution. Tables \ref{age_war_tab7} to \ref{age_war_tab9} do not support a blanket characterization of wartime Japan as either semi-strong efficient or semi-strong inefficient. Many major events, including several important wartime regulations, generate no significant reaction, and even where significant responses appear, their sign and cumulative profile vary across the four portfolios. Public news, in other words, was not incorporated uniformly across firms. At the same time, significant CARs and SCARs should not be read mechanically as evidence of delayed adjustment, since the data do not permit a clean separation among post-announcement drift, pre-event leakage, and other contemporaneous shocks.

The extended model helps organize this heterogeneity. Some political, diplomatic, monetary, and sanctions-related events are best interpreted as regime-risk shocks. These shocks moved non-zaibatsu firms more strongly because such firms were more exposed to external finance, policy uncertainty, sanctions, and administrative reallocation, whereas zaibatsu groups were partially insulated by internal capital markets, diversified group claims, and privileged access. Other events are better interpreted as embedded-rent shocks: they revised the value of being inserted into the state-led strategic order and generated cumulative responses concentrated in the zaibatsu military portfolio. Still other episodes, such as the Fall of Guadalcanal, are better interpreted through group-continuation values, because the cumulative response is concentrated in the zaibatsu non-military portfolio rather than in firms directly tied to military demand.

Historically, the relative strength of the zaibatsu military portfolio was plausibly associated with more than formal placement in a priority sector. Internal capital markets, group governance, centralized monitoring, and preferential access to credit, materials, and procurement made it easier for major zaibatsu firms to convert wartime priority into more dependable shareholder outcomes. Conversely, the persistent weakness of the non-zaibatsu military portfolio suggests that formal strategic importance without comparable financial and administrative support was often insufficient. The distinction among the PI, API, and TRI reinforces this point. In wartime Japan, these measures were not fully interchangeable: the PI appears to capture more directly the capitalization of expected priority access in market prices, whereas the API and TRI more clearly reflect the channels through which value reached shareholders once capital adjustments, rights issues, retained earnings, and payout restrictions are taken into account.

These findings suggest that the central issue in wartime market efficiency is not whether intervention simply destroyed the market, but how price formation operated when the mapping from fundamentals to returns was mediated by administrative controls and political hierarchy. Stock prices continued to respond to news, but the news being priced included institutional rents, financing advantages, regime-risk exposure, and group-level continuation values. This is the sense in which wartime Japanese stock prices are best understood as exhibiting institutionally contingent efficiency.

More broadly, the paper highlights the value of analyzing hybrid or controlled market regimes not through a simple market-versus-state dichotomy, but as settings in which price formation remains active under institutionally uneven conditions. At the same time, the paper's identification strategy remains constrained by coarse portfolio classifications, survivorship among listed firms, and uncertainty about when information was actually acquired in a censored wartime environment. More modestly, but more securely, this paper shows that state intervention in wartime Japan reshaped rather than eliminated the informational role of prices, and that the effects of this reshaping were distributed unevenly across firms according to their institutional location.

%% file: zaibatsu_ack.tex
\section*{Acknowledgments}
The authors thank Shinichi Hirota, Tatsuki Inoue, Hideaki Miyajima, Kaiji Motegi, Tetsuji Okazaki, Yusuke Osaki, Masato Shizume, Yusuke Tsujimoto, Tomoaki Yamada, Junichi Yamanoi, Tatsuma Wada, as well as seminar participants at Meiji University and Waseda University and conference participants at the Japanese Economic Association 2026 Spring Meeting and the 100th Annual Conference of the Western Economic Association International, for their helpful comments and suggestions. The authors also acknowledge financial support from the Japan Society for the Promotion of Science (JSPS) KAKENHI (Grant Numbers 23H00838 and 23K25535) and the Japan Science and Technology Agency (JST), Moonshot Research and Development Program (Grant Number JPMJMS2215). All data and code used in this study are available from the authors upon request.

%% file: zaibatsu_table.tex
\bigskip

\bigskip

\bigskip

\setcounter{table}{0}
\renewcommand{\thetable}{\arabic{table}}

\clearpage

\begin{table}[p]
\caption{Industry Classification under the Temporary Funds Adjustment Law}
\label{age_war_tab1}
\begin{center}{\scriptsize
\begin{tabular}{p{15cm}}\hline\hline
Category A ({\it{Ko}}-shu)\\\hline
\vspace{-5mm}
\begin{enumerate}
 \setlength{\itemsep}{0pt}       
 \setlength{\parskip}{0pt}       
 \item Major part of metal mining
 \item Coal mining (excluding lignite)
 \item Petroleum mining
 \item Major part of iron manufacturing
 \item Machine tool manufacturing (excluding sawmilling and woodworking machines)
 \item Vehicle manufacturing---locomotives, freight cars and automobiles
 \item Shipbuilding---steel vessels
 \item Aircraft manufacturing industry
 \item Manufacturing of weapons and weapon components
 \item Glass and glass product manufacturing---optical glass, tempered glass, etc.
 \item Transportation---railways and tramways necessary for military and industrial purposes; ocean and coastal shipping; and aviation
\end{enumerate}
\\[-2mm]
\\[0.5mm] 
\hline
Category B ({\it{Otsu}}-shu)\\\hline
\vspace{-5mm}
\begin{enumerate}
 \setlength{\itemsep}{0pt}       
 \setlength{\parskip}{0pt}       
 \item Coal mining---lignite
 \item Raw silk manufacturing
 \item Synthetic fiber manufacturing
 \item Textile industry---artificial fiber textiles and linen textiles
 \item Most of non-ferrous metal materials manufacturing
 \item Crane manufacturing
 \item Sewing machine manufacturing
 \item Vehicle manufacturing---passenger cars, electric cars, and motorcycles
 \item Shipbuilding---wooden ships
 \item Synthetic rubber manufacturing
 \item Pulp manufacturing
 \item Industrial salt manufacturing
 \item Sawmilling
 \item Sugar manufacturing
 \item Agriculture and forestry
 \item Fisheries
 \item Warehousing
 \item Businesses related to education, culture, sports, social services, and medical care
\end{enumerate}
\\[-2mm]
\\[0.5mm] 
\hline
Category C ({\it{Hei}}-shu)\\\hline
\vspace{-5mm}
\begin{enumerate}
 \setlength{\itemsep}{0pt}       
 \setlength{\parskip}{0pt}       
 \item Silk floss and cotton manufacturing
 \item Spinning industry (excluding flax yarn)
 \item Knitting and braiding manufacturing
 \item Pig iron smelting industry
 \item Manufacturing of metal tableware
 \item Manufacturing of household electrical appliances
 \item Elevator manufacturing
 \item Clock manufacturing
 \item Safe manufacturing
 \item Musical instrument manufacturing
 \item Bicycle manufacturing
 \item Waterworks equipment manufacturing
 \item Cement manufacturing
 \item Lime manufacturing
 \item Superphosphate of lime manufacturing
 \item Seasoning manufacturing
 \item Match manufacturing
 \item Department store industry
 \item Financial services
 \item Newspaper publishing, book and magazine publishing
 \item Businesses related to entertainment and amusement
 \setlength{\itemsep}{3pt}       
 \setlength{\parskip}{5pt} 
\end{enumerate}
\\[-2mm]
\\[0.5mm]\hline\hline
\end{tabular}}
{\resizebox{15cm}{!}{\begin{minipage}{500pt}\scriptsize
\vspace*{5pt}
{\underline{Note:}} The electric power industry falls under Category A when it supplies power necessary for Category A businesses; otherwise it falls under Category B. Gas is classified as Category B (see \citet[pp.75--76]{mof1957hfm}).\\
{\underline{Source:}} Compiled by the authors based on Tables 3 and 4 in \citet[p.294]{boj1982hyh}
\end{minipage}}}%
\end{center}
\end{table}

\clearpage

\begin{landscape}
\begin{table}[p]
\caption{Descriptive Statistics for the Nominal Stock Prices in the Tokyo Stock Exchange}
\label{age_war_tab2}
\begin{center}\resizebox{22cm}{!}{
\begin{tabular}{lp{18cm}ccccccccccccc}\hline\hline
\underline{Listed issue} &  & Zaibatsu & Military &  & Mean & SD & Min & Max & & Skewness & Kurtosis & & $\mathcal{N}$ &  \\\cline{3-4}\cline{6-9}\cline{11-12}\cline{14-14}
\quad Tokyo stock exchange, new share &  &  &  &  & $1386.15$  & $271.98$  & $803.00$  & $2279.00$ &  & $0.22$  & $2.89$ &  & $3911$ &  \\
\quad Osaka stock exchange, new share &  &  &  &  & $813.65$  & $204.00$  & $388.00$  & $1376.00$ &  & $0.58$  & $3.04$ &  & $2910$ &  \\
\quad Kanebo (Kanegafuchi Boseki), old share &  &  &  &  & $1953.32$  & $444.70$  & $1225.00$  & $3271.00$ &  & $0.26$  & $2.08$ &  & $3625$ &  \\
\quad Kanebo (Kanegafuchi Boseki), new share\#1 &  &  &  &  & $1283.24$  & $578.86$  & $500.00$  & $3170.00$ &  & $1.21$  & $3.99$ &  & $2341$ &  \\
\quad Kanebo (Kanegafuchi Boseki), new share\#2 &  &  &  &  & $897.69$  & $140.82$  & $668.00$  & $1271.00$ &  & $0.56$  & $2.25$ &  & $1384$ &  \\
\quad Dalian stock and commodity exchange, new share &  &  &  &  & $82.35$  & $13.19$  & $55.00$  & $112.00$ &  & $0.52$  & $2.72$ &  & $75$ &  \\
\quad Teikoku jinzo kenshi (Teijin), new share &  &  &  &  & $920.64$  & $166.72$  & $632.00$  & $1364.00$ &  & $0.42$  & $2.29$ &  & $2713$ &  \\
\quad Toyo rayon (Toray Industries), new share &  & $\bigcirc$ &  &  & $787.80$  & $135.96$  & $513.00$  & $1176.00$ &  & $0.37$  & $2.50$ &  & $2613$ &  \\
\quad Asano sement, new share &  & $\bigcirc$ &  &  & $235.48$  & $90.48$  & $35.00$  & $417.00$ &  & $0.02$  & $2.15$ &  & $1769$ &  \\
\quad Fuji seishi (Fuji paper company) &  &  &  &  & $217.28$  & $80.57$  & $92.00$  & $479.00$ &  & $0.94$  & $3.46$ &  & $605$ &  \\
\quad Oji seishi (Oji paper company) &  & $\bigcirc$ &  &  & $944.27$  & $172.46$  & $597.00$  & $1190.00$ &  & $-0.34$  & $1.60$ &  & $910$ &  \\
\quad Nissin boseki (nissin spinning company) &  &  &  &  & $1082.34$  & $104.21$  & $849.00$  & $1234.00$ &  & $-0.66$  & $2.11$ &  & $716$ &  \\
\quad Nichiro gyogyo &  &  &  &  & $655.10$  & $84.34$  & $386.00$  & $869.00$ &  & $-0.13$  & $3.29$ &  & $2806$ &  \\
\quad Nippon kogyo (Nippon mining corporation), old share &  & $\bigcirc$ & $\bigcirc$ &  & $801.21$  & $235.73$  & $480.00$  & $1332.00$ &  & $0.68$  & $2.12$ &  & $1892$ &  \\
\quad Nippon kogyo (Nippon mining corporation), new share &  & $\bigcirc$ & $\bigcirc$ &  & $728.52$  & $280.38$  & $335.00$  & $1155.00$ &  & $-0.12$  & $1.26$ &  & $993$ &  \\
\quad Nippon yusen, old share &  & $\bigcirc$ & $\bigcirc$ &  & $852.66$  & $152.00$  & $582.00$  & $1180.00$ &  & $0.19$  & $1.91$ &  & $2286$ &  \\
\quad Nippon yusen, new share &  & $\bigcirc$ & $\bigcirc$ &  & $436.61$  & $295.55$  & $50.00$  & $1000.00$ &  & $0.36$  & $1.65$  & & $3660$ &  \\
\quad Mitsubishi jukogyo (Mitubishi heavy industries) &  & $\bigcirc$ & $\bigcirc$ &  & $855.82$  & $95.93$  & $688.00$  & $1184.00$ &  & $0.78$  & $3.05$ &  & $2025$ &  \\
\quad Nippon oil &  &  & $\bigcirc$ &  & $616.49$  & $179.15$  & $204.00$  & $1063.00$ &  & $-0.02$  & $2.05$ &  & $3920$ &  \\
\quad Hitachi seisakusho (Hitachi) &  & $\bigcirc$ & $\bigcirc$ &  & $870.47$  & $155.28$  & $599.00$  & $1473.00$ &  & $0.74$  & $4.03$ &  & $1976$ &  \\
\quad Nihon Sangyo (Nissan corporation) &  & $\bigcirc$ &  &  & $688.99$  & $362.21$  & $124.00$  & $1444.00$ &  & $0.07$  & $2.00$ &  & $2332$ &  \\
\quad Manshu jukogyo kaihatsu (Manchukuo heavy industries development corporation) &  & $\bigcirc$ & $\bigcirc$ &  & $678.98$  & $96.81$  & $500.00$  & $891.00$ &  & $0.33$  & $1.89$ &  & $1649$ &  \\
\quad Mitsubishi kogyo (Mitubishi mining corporation) &  & $\bigcirc$ & $\bigcirc$ &  & $1189.03$  & $86.80$  & $988.00$  & $1429.00$ &  & $0.75$  & $3.18$ &  & $827$ &  \\
\quad Nippon denki kogyo (Resonac holdings corporation) &  & $\bigcirc$ & $\bigcirc$ &  & $754.29$  & $109.09$  & $565.00$  & $1037.00$ &  & $0.63$  & $2.48$ &  & $725$ &  \\
\quad Nippon kokan (JFE engineering) &  & $\bigcirc$ & $\bigcirc$ &  & $929.70$  & $234.85$  & $535.00$  & $1609.00$ &  & $0.70$  & $3.06$ &  & $2673$ &  \\
\quad Kokura seiko (Kokura iron \& steel corporation) &  & $\bigcirc$ & $\bigcirc$ &  & $618.02$  & $127.20$  & $430.00$  & $980.00$ &  & $1.04$  & $2.77$ &  & $1159$ &  \\
\quad Tokyo dento (Tokyo electric light) &  &  & $\bigcirc$ &  & $428.83$  & $169.34$  & $127.00$  & $690.00$ &  & $-0.45$  & $1.67$ &  & $3432$ &  \\
\quad Nippon denryoku (Nippon electric power company) &  &  &  &  & $521.87$  & $43.68$  & $414.00$  & $625.00$ &  & $-0.28$  & $2.74$ &  & $931$ &  \\
\quad Minami manshu tetsudo (South manchuria railway), old share &  &  & $\bigcirc$ &  & $593.03$  & $150.25$  & $230.00$  & $823.00$ &  & $-1.10$  & $3.17$ &  & $3157$ &  \\
\quad Minami manshu tetsudo (South manchuria railway), new share &  &  & $\bigcirc$ &  & $389.55$  & $122.05$  & $164.00$  & $601.00$ &  & $-0.12$  & $1.85$ &  & $777$ &  \\
\quad Hokkaido tanko kisen (Hokkaido colliery \& steamship) &  & $\bigcirc$ & $\bigcirc$ &  & $775.00$  & $93.17$  & $571.00$  & $1050.00$ &  & $0.79$  & $3.26$ &  & $2192$ &  \\
\quad Dai-nippon jinzo hiryo (Nissan chemical corporation) &  &  &  &  & $525.05$  & $57.07$  & $394.00$  & $637.00$ &  & $-0.22$  & $2.24$ &  & $813$ &  \\
\quad Showa hiryo (Resonac holdings corporation) &  & $\bigcirc$ &  &  & $768.23$  & $90.30$  & $576.00$  & $972.00$ &  & $-0.30$  & $2.40$ &  & $630$ &  \\
\quad Dai-nippon seito (Dai-nippon sugar corporation), new shares\#2 &  & $\bigcirc$ &  &  & $154.44$  & $66.29$  & $40.00$  & $236.00$ &  & $-0.39$  & $1.71$ &  & $141$ &  \\
\quad Dai-nippon seito (Dai-nippon sugar corporation), new shares integrated &  & $\bigcirc$ &  &  & $639.76$  & $219.13$  & $120.00$  & $1046.00$ &  & $-0.89$  & $2.91$ &  & $3200$ &  \\
\quad Meiji seito (Meiji sugar corporation) &  &  &  &  & $920.63$  & $214.31$  & $528.00$  & $1212.00$ &  & $-0.44$  & $1.43$ &  & $1342$ &  \\
\quad Ensuiko sugar refining (Taiwan) &  &  &  &  & $671.77$  & $91.46$  & $489.00$  & $873.00$ &  & $0.71$  & $2.41$ &  & $1151$ &  \\
\quad Nippon soda &  & $\bigcirc$ &  &  & $664.71$  & $216.68$  & $212.00$  & $1266.00$ &  & $-0.01$  & $3.31$ &  & $1436$ &  \\
\quad Dai-nippon bear (Dai-nippon brewery company), new share &  &  &  &  & $386.11$  & $83.80$  & $178.00$  & $629.00$ &  & $-0.11$  & $2.52$ &  & $3425$ &  \\
\quad Nippon suisan (Nissui corporation), new share &  & $\bigcirc$ &  &  & $481.61$  & $78.33$  & $356.00$  & $686.00$ &  & $0.71$  & $2.47$ &  & $1286$ &  \\
\quad Rasa kogyo (Rasa industries) &  &  &  &  & $638.21$  & $163.32$  & $270.00$  & $862.00$ &  & $-0.51$  & $1.81$ &  & $664$ &  \\\hline\hline
\end{tabular}}
{\resizebox{21.5cm}{!}{\begin{minipage}{800pt}\scriptsize
\vspace*{5pt}
{\underline{Sources:}} All data are obtained from the {\it Chugai Shogyo Shimpo}.
\end{minipage}}}%
\end{center}
\end{table}
\end{landscape}

\clearpage

\begin{figure}[tbp]
 \caption{Time-Series Plots of Market Indices and Market-Capitalization Shares}
 \label{age_war_fig1}
 \begin{center}
  \includegraphics[scale=0.6]{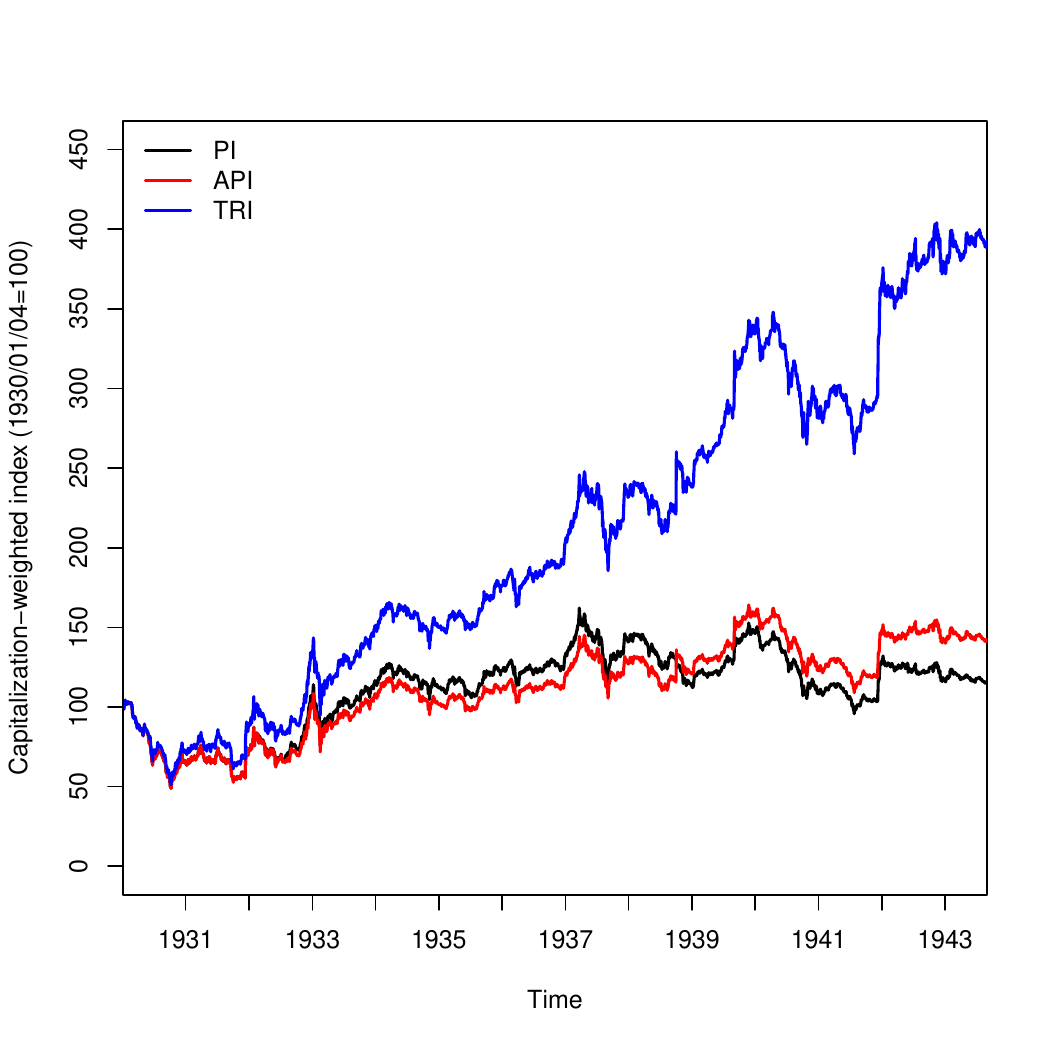}
  \includegraphics[scale=0.6]{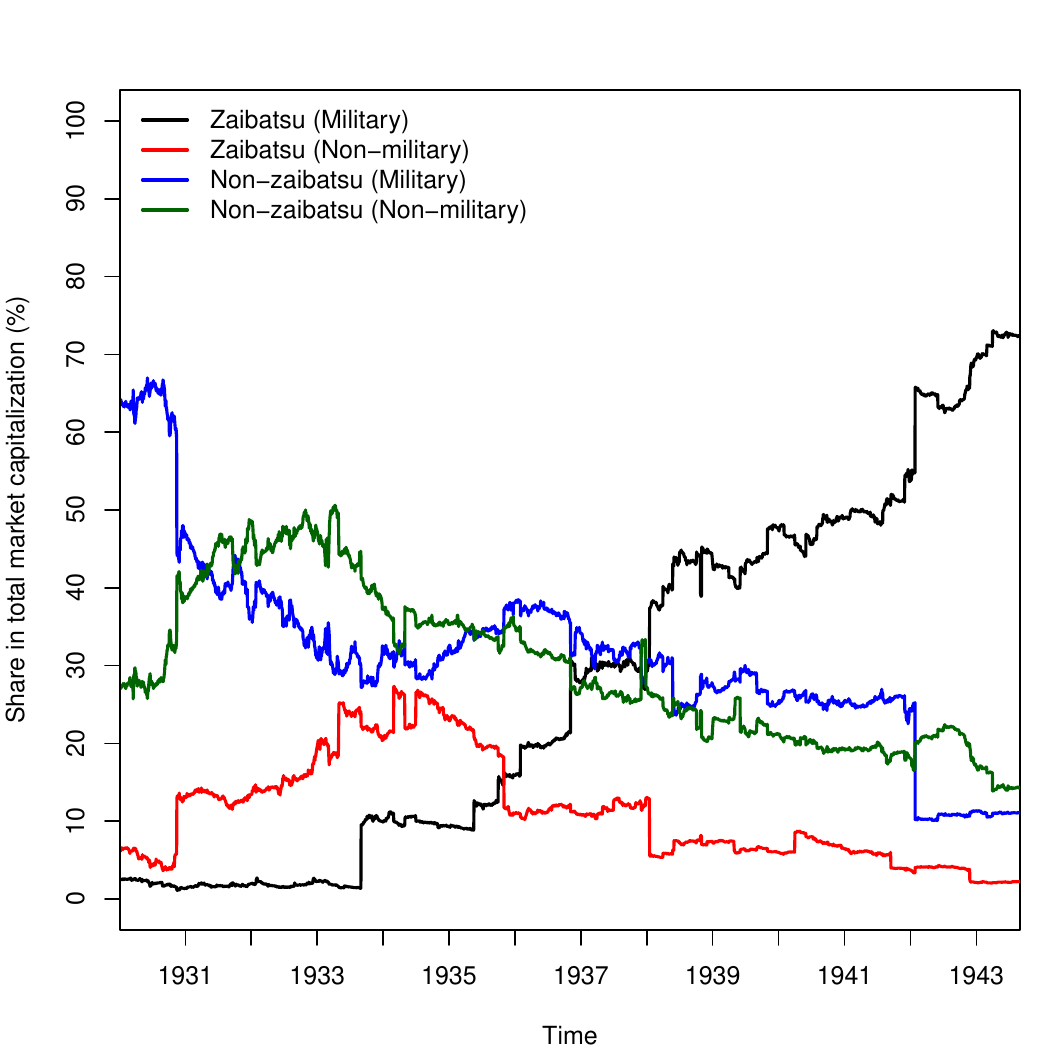}
 \end{center}
\end{figure}

\clearpage

\begin{figure}[tbp]
 \caption{Time-Series Plots of the Market Indices for Each Portfolio}
 \label{age_war_fig2}
 \begin{center}
  \includegraphics[scale=0.7]{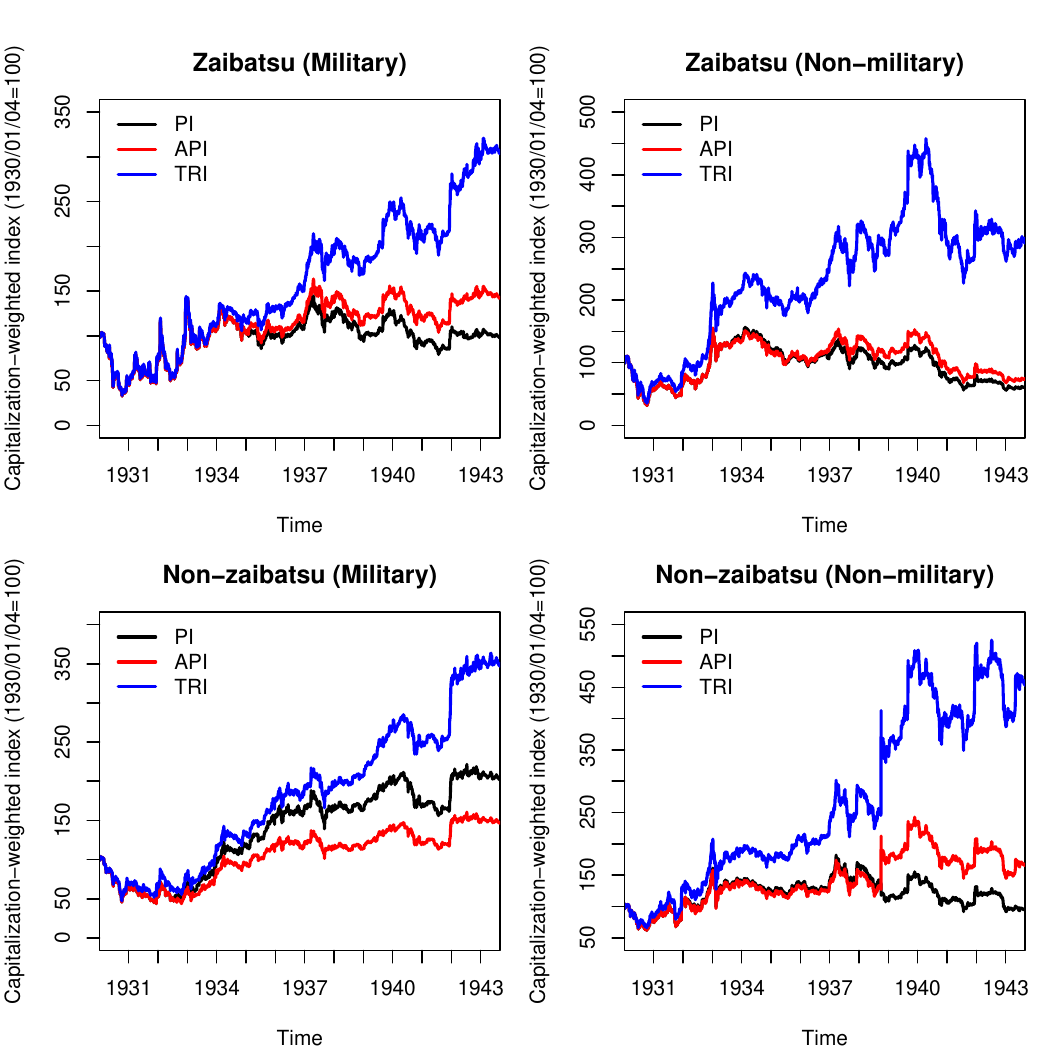}
 \end{center}
\end{figure}

\clearpage

\begin{table}[tbp]
\caption{Descriptive Statistics and Unit Root Tests for the Return Series}
\label{age_war_tab3}
\begin{center}\resizebox{15cm}{!}{
\begin{tabular}{ccccccccccc}\hline\hline
\multicolumn{3}{l}{(1) Price Index (PI)} &  &  & \multicolumn{2}{c}{$Zaibatsu$} &  & \multicolumn{2}{c}{$Non\text{--}zaibatsu$} &  \\\cline{6-7}\cline{9-10}
 &  &  & $R_{mkt}$ &  & $R_{zm}$ & $R_{zn}$ &  & $R_{nm}$ & $R_{nn}$ &  \\\hline
 & Mean &  & $-0.0078$  &  & $-0.0078$  & $-0.0080$  &  & $-0.0077$  & $-0.0078$  &  \\
 & SD &  & 0.0111  &  & 0.0195  & 0.0147  &  & 0.0123  & 0.0129  &  \\
 & Min &  & $-0.1225$  &  & $-0.1608$  & $-0.1342$  &  & $-0.1278$  & $-0.1650$  &  \\
 & Max &  & 0.1181  &  & 0.2289  & 0.1397  &  & 0.1155  & 0.1212  &  \\\hdashline
 & ADF-GLS &  & $-8.7292$  &  & $-39.8607$  & $-11.6358$  &  & $-61.8306$  & $-5.1845$  &  \\
 & $\hat\phi$ &  & 0.1424  &  & 0.0379  & 0.1103  &  & 0.0247  & 0.3673  &  \\
 & Lag &  & 16  &  & 1  & 11  &  & 0  & 16  &  \\\hline
\multicolumn{3}{l}{(2) Adjusted Price Index (API)} &  &  & \multicolumn{2}{c}{$Zaibatsu$} &  & \multicolumn{2}{c}{$Non\text{--}zaibatsu$} &  \\\cline{6-7}\cline{9-10}
 &  &  & $R_{mkt}$ &  & $R_{zm}$ & $R_{zn}$ &  & $R_{nm}$ & $R_{nn}$ &  \\\hline
 & Mean &  & $-0.0077$  &  & $-0.0077$  & $-0.0079$  &  & $-0.0077$  & $-0.0077$  &  \\
 & SD &  & 0.0113  &  & 0.0194  & 0.0147  &  & 0.0120  & 0.0152  &  \\
 & Min &  & $-0.1225$  &  & $-0.1608$  & $-0.1342$  &  & $-0.1278$  & $-0.1650$  &  \\
 & Max &  & 0.1435  &  & 0.2289  & 0.1397  &  & 0.1155  & 0.5053  &  \\\hdashline
 & ADF-GLS &  & $-8.9022$  &  & $-39.9692$  & $-11.5630$  &  & $-61.6870$  & $-5.6848$  &  \\
 & $\hat\phi$ &  & 0.1291  &  & 0.0360  & 0.1087  &  & 0.0271  & 0.2259  &  \\
 & Lag &  & 16  &  & 1  & 11  &  & 0  & 16  &  \\\hline
\multicolumn{3}{l}{(3) Total Return Index (TRI)} &  &  & \multicolumn{2}{c}{$Zaibatsu$} &  & \multicolumn{2}{c}{$Non\text{--}zaibatsu$} &  \\\cline{6-7}\cline{9-10}
 &  &  & $R_{mkt}$ &  & $R_{zm}$ & $R_{zn}$ &  & $R_{nm}$ & $R_{nn}$ &  \\\hline
 & Mean &  & $-0.0075$  &  & $-0.0076$  & $-0.0076$  &  & $-0.0075$  & $-0.0075$  &  \\
 & SD &  & 0.0112  &  & 0.0194  & 0.0148  &  & 0.0119  & 0.0152  &  \\
 & Min &  & $-0.1225$  &  & $-0.1608$  & $-0.1342$  &  & $-0.1278$  & $-0.1650$  &  \\
 & Max &  & 0.1435  &  & 0.2289  & 0.1397  &  & 0.1155  & 0.5053  &  \\\hdashline
 & ADF-GLS &  & $-8.6291$  &  & $-39.9690$  & $-11.2776$  &  & $-61.6271$  & $-5.5787$  &  \\
 & $\hat\phi$ &  & 0.1357  &  & 0.0367  & 0.1117  &  & 0.0280  & 0.2352  &  \\
 & Lag &  & 16  &  & 1  & 11  &  & 0  & 16  &  \\\hline
 & $\mathcal{N}$ &  & 4019  &  & 4019  & 4019  &  & 4019  & 4019  &  \\\hline\hline
\end{tabular}}
{\begin{minipage}{420pt}\footnotesize
\vspace*{5pt}
{\underline{Notes:}}
\begin{itemize}
\item[(1)] ``$R_{mkt}$,'' ``$R_{zm}$,'' ``$R_{zn}$,'' ``$R_{nm}$,'' and ``$R_{nn}$'' denote the returns on the market, Zaibatsu military industry, Zaibatsu non-military industry, Non-zaibatsu military industry, Non-zaibatsu non-military industry portfolios, respectively.
\item[(2)] ``ADF-GLS'' denotes the ADF-GLS test statistics, ``Lags'' denotes the lag order selected by the MBIC, and ``$\hat\phi$'' denotes the coefficients vector in the GLS detrended series (see Equation (6) in \citet{ng2001lls}).
\item[(3)] In computing the ADF-GLS test, a model with a time trend and a constant is assumed. The critical value at the 1\% significance level for the ADF-GLS test is ``$-3.42$''.
\item[(4)] ``$\mathcal{N}$'' denotes the number of observations.
\item[(5)] R version 4.5.3 was used to compute the statistics.
\end{itemize}
\end{minipage}}%
\end{center}
\end{table}

\clearpage

\begin{table}[tbp]
\caption{Preliminary Estimation Results for the CAPM-AR($p$)-SV Model (PI)}
\label{age_war_tab4}
\begin{center}\resizebox{15cm}{!}{
\begin{tabular}{ccccccccc}\hline\hline
 \multicolumn{2}{c}{(1) Jan 4, 1930--Aug 31, 1943} &  & \multicolumn{2}{c}{$Zaibatsu$} &  & \multicolumn{2}{c}{$Non\text{--}zaibatsu$} &  \\\cline{4-5}\cline{7-8}
 &  &  & $Military$ & $Non\text{--}military$ &  & $Military$ & $Non\text{--}military$ &  \\\hline
 & \multirow{2}*{$\hat\alpha_i^{GLS}$} &  & $0.0005$ & $0.0001$ &  & $-0.0010$ & $0.0001$ &  \\
 &  &  & $[0.0001]$ & $[0.0001]$ &  & $[0.0001]$ & $[0.0001]$ &  \\
 & \multirow{2}*{$\hat\beta_i^{GLS}$} &  & $1.0731$ & $1.0311$ &  & $0.8566$ & $1.0186$ &  \\
 &  &  & $[0.0065]$ & $[0.0098]$ &  & $[0.0084]$ & $[0.0076]$ &  \\\hline
 & AR($p$) &  & $1$ & $1$ &  & $1$ & $1$ &  \\
 & $\bar{R}^2$ &  & $0.3890$ & $0.6172$ &  & $0.6937$ & $0.7731$ &  \\\hline\hline
 \multicolumn{2}{c}{(2) Jan 4, 1930--Jul 6, 1937} &  & \multicolumn{2}{c}{$Zaibatsu$} &  & \multicolumn{2}{c}{$Non\text{--}zaibatsu$} &  \\\cline{4-5}\cline{7-8}
 &  &  & $Military$ & $Non\text{--}military$ &  & $Military$ & $Non\text{--}military$ &  \\\hline
 & \multirow{2}*{$\hat\alpha_i^{GLS}$} &  & $0.0001$ & $-0.0002$ &  & $-0.0002$ & $-0.0002$ &  \\
 &  &  & $[0.0002]$ & $[0.0001]$ &  & $[0.0001]$ & $[0.0001]$ &  \\
 & \multirow{2}*{$\hat\beta_i^{GLS}$} &  & $1.0178$ & $0.9820$ &  & $0.9733$ & $0.9694$ &  \\
 &  &  & $[0.0183]$ & $[0.0122]$ &  & $[0.0105]$ & $[0.0089]$ &  \\\hline
 & AR($p$) &  & $1$ & $1$ &  & $1$ & $1$ &  \\
 & $\bar{R}^2$ &  & $0.3392$ & $0.6089$ &  & $0.7434$ & $0.7994$ &  \\\hline\hline
 \multicolumn{2}{c}{(3) Jul 7, 1937--Aug 31, 1943} &  & \multicolumn{2}{c}{$Zaibatsu$} &  & \multicolumn{2}{c}{$Non\text{--}zaibatsu$} &  \\\cline{4-5}\cline{7-8}
 &  &  & $Military$ & $Non\text{--}military$ &  & $Military$ & $Non\text{--}military$ &  \\\hline
 & \multirow{2}*{$\hat\alpha_i^{GLS}$} &  & $0.0005$ & $0.0007$ &  & $-0.0028$ & $0.0011$ &  \\
 &  &  & $[0.0001]$ & $[0.0001]$ &  & $[0.0001]$ & $[0.0001]$ &  \\
 & \multirow{2}*{$\hat\beta_i^{GLS}$} &  & $1.0787$ & $1.1386$ &  & $0.5671$ & $1.1736$ &  \\
 &  &  & $[0.0067]$ & $[0.0158]$ &  & $[0.0107]$ & $[0.0138]$ &  \\\hline
 & AR($p$) &  & $1$ & $1$ &  & $1$ & $1$ &  \\
 & $\bar{R}^2$ &  & $0.8910$ & $0.6464$ &  & $0.5101$ & $0.7232$ &  \\\hline\hline
\end{tabular}}
{\begin{minipage}{420pt}\footnotesize
{\underline{Notes:}}
\vspace*{5pt}
 \begin{itemize}
  \item[(1)] ``($\hat\alpha_i^{GLS},\hat\beta_i^{GLS})$,'' ``AR$(p)$,'' and ``$\bar{R}^2$ denote the GLS estimates, the optimal lag order for the CAPM-AR($p$)-SV model,  and the adjusted $R^2$, respectively.
  \item[(2)] Standard errors for the GLS estimates are between brackets.
  \item[(3)] R version 4.5.3 was used to compute the estimates.
 \end{itemize}
\end{minipage}}%
\end{center}
\end{table}

\clearpage

\begin{table}[tbp]
\caption{Preliminary Estimation Results for the CAPM-AR($p$)-SV Model (API)}
\label{age_war_tab5}
\begin{center}\resizebox{15cm}{!}{
\begin{tabular}{ccccccccc}\hline\hline
 \multicolumn{2}{c}{(1) Jan 4, 1930--Aug 31, 1943} &  & \multicolumn{2}{c}{$Zaibatsu$} &  & \multicolumn{2}{c}{$Non\text{--}zaibatsu$} &  \\\cline{4-5}\cline{7-8}
 &  &  & $Military$ & $Non\text{--}military$ &  & $Military$ & $Non\text{--}military$ &  \\\hline
 & \multirow{2}*{$\hat\alpha_i^{GLS}$} &  & $0.0004$ & $-0.0001$ &  & $-0.0013$ & $0.0006$ &  \\
 &  &  & $[0.0001]$ & $[0.0001]$ &  & $[0.0001]$ & $[0.0001]$ &  \\
 & \multirow{2}*{$\hat\beta_i^{GLS}$} &  & $1.0540$ & $1.0173$ &  & $0.8241$ & $1.0798$ &  \\
 &  &  & $[0.0064]$ & $[0.0096]$ &  & $[0.0084]$ & $[0.0078]$ &  \\\hline
 & AR($p$) &  & $1$ & $1$ &  & $1$ & $1$ &  \\
 & $\bar{R}^2$ &  & $0.3699$ & $0.5964$ &  & $0.6721$ & $0.7093$ &  \\\hline\hline
 \multicolumn{2}{c}{(2) Jan 4, 1930--Jul 6, 1937} &  & \multicolumn{2}{c}{$Zaibatsu$} &  & \multicolumn{2}{c}{$Non\text{--}zaibatsu$} &  \\\cline{4-5}\cline{7-8}
 &  &  & $Military$ & $Non\text{--}military$ &  & $Military$ & $Non\text{--}military$ &  \\\hline
 & \multirow{2}*{$\hat\alpha_i^{GLS}$} &  & $0.0002$ & $-0.0001$ &  & $-0.0004$ & $-0.0002$ &  \\
 &  &  & $[0.0002]$ & $[0.0001]$ &  & $[0.0001]$ & $[0.0001]$ &  \\
 & \multirow{2}*{$\hat\beta_i^{GLS}$} &  & $1.0266$ & $0.9887$ &  & $0.9629$ & $0.9758$ &  \\
 &  &  & $[0.0183]$ & $[0.0121]$ &  & $[0.0106]$ & $[0.0090]$ &  \\\hline
 & AR($p$) &  & $1$ & $1$ &  & $1$ & $1$ &  \\
 & $\bar{R}^2$ &  & $0.3389$ & $0.6115$ &  & $0.7491$ & $0.8063$ &  \\\hline\hline
 \multicolumn{2}{c}{(3) Jul 7, 1937--Aug 31, 1943} &  & \multicolumn{2}{c}{$Zaibatsu$} &  & \multicolumn{2}{c}{$Non\text{--}zaibatsu$} &  \\\cline{4-5}\cline{7-8}
 &  &  & $Military$ & $Non\text{--}military$ &  & $Military$ & $Non\text{--}military$ &  \\\hline
 & \multirow{2}*{$\hat\alpha_i^{GLS}$} &  & $-0.0002$ & $0.0000$ &  & $-0.0033$ & $0.0024$ &  \\
 &  &  & $[0.0001]$ & $[0.0001]$ &  & $[0.0001]$ & $[0.0001]$ &  \\
 & \multirow{2}*{$\hat\beta_i^{GLS}$} &  & $0.9794$ & $1.0504$ &  & $0.5025$ & $1.3451$ &  \\
 &  &  & $[0.0079]$ & $[0.0160]$ &  & $[0.0099]$ & $[0.0170]$ &  \\\hline
 & AR($p$) &  & $1$ & $1$ &  & $1$ & $1$ &  \\
 & $\bar{R}^2$ &  & $0.7177$ & $0.5362$ &  & $0.4284$ & $0.6745$ &  \\\hline\hline
\end{tabular}}
{\begin{minipage}{420pt}\footnotesize
{\underline{Notes:}} As for Table \ref{age_war_tab4}.
\end{minipage}}%
\end{center}
\end{table}

\clearpage

\begin{table}[tbp]
\caption{Preliminary Estimation Results for the CAPM-AR($p$)-SV Model (TRI)}
\label{age_war_tab6}
\begin{center}\resizebox{15cm}{!}{
\begin{tabular}{ccccccccc}\hline\hline
 \multicolumn{2}{c}{(1) Jan 4, 1930--Aug 31, 1943} &  & \multicolumn{2}{c}{$Zaibatsu$} &  & \multicolumn{2}{c}{$Non\text{--}zaibatsu$} &  \\\cline{4-5}\cline{7-8}
 &  &  & $Military$ & $Non\text{--}military$ &  & $Military$ & $Non\text{--}military$ &  \\\hline
 & \multirow{2}*{$\hat\alpha_i^{GLS}$} &  & $0.0003$ & $0.0000$ &  & $-0.0014$ & $0.0006$ &  \\
 &  &  & $[0.0001]$ & $[0.0001]$ &  & $[0.0001]$ & $[0.0001]$ &  \\
 & \multirow{2}*{$\hat\beta_i^{GLS}$} &  & $1.0532$ & $1.0204$ &  & $0.8151$ & $1.0831$ &  \\
 &  &  & $[0.0065]$ & $[0.0097]$ &  & $[0.0085]$ & $[0.0079]$ &  \\\hline
 & AR($p$) &  & $1$ & $1$ &  & $1$ & $1$ &  \\
 & $\bar{R}^2$ &  & $0.3694$ & $0.5837$ &  & $0.6672$ & $0.7075$ &  \\\hline\hline
 \multicolumn{2}{c}{(2) Jan 4, 1930--Jul 6, 1937} &  & \multicolumn{2}{c}{$Zaibatsu$} &  & \multicolumn{2}{c}{$Non\text{--}zaibatsu$} &  \\\cline{4-5}\cline{7-8}
 &  &  & $Military$ & $Non\text{--}military$ &  & $Military$ & $Non\text{--}military$ &  \\\hline
 & \multirow{2}*{$\hat\alpha_i^{GLS}$} &  & $0.0002$ & $0.0000$ &  & $-0.0004$ & $-0.0001$ &  \\
 &  &  & $[0.0002]$ & $[0.0001]$ &  & $[0.0001]$ & $[0.0001]$ &  \\
 & \multirow{2}*{$\hat\beta_i^{GLS}$} &  & $1.0240$ & $0.9922$ &  & $0.9611$ & $0.9770$ &  \\
 &  &  & $[0.0185]$ & $[0.0123]$ &  & $[0.0108]$ & $[0.0091]$ &  \\\hline
 & AR($p$) &  & $1$ & $1$ &  & $1$ & $1$ &  \\
 & $\bar{R}^2$ &  & $0.3382$ & $0.6018$ &  & $0.7437$ & $0.8023$ &  \\\hline\hline
 \multicolumn{2}{c}{(3) Jul 7, 1937--Aug 31, 1943} &  & \multicolumn{2}{c}{$Zaibatsu$} &  & \multicolumn{2}{c}{$Non\text{--}zaibatsu$} &  \\\cline{4-5}\cline{7-8}
 &  &  & $Military$ & $Non\text{--}military$ &  & $Military$ & $Non\text{--}military$ &  \\\hline
 & \multirow{2}*{$\hat\alpha_i^{GLS}$} &  & $-0.0002$ & $0.0000$ &  & $-0.0034$ & $0.0024$ &  \\
 &  &  & $[0.0001]$ & $[0.0001]$ &  & $[0.0001]$ & $[0.0001]$ &  \\
 & \multirow{2}*{$\hat\beta_i^{GLS}$} &  & $0.9780$ & $1.0495$ &  & $0.4972$ & $1.3604$ &  \\
 &  &  & $[0.0080]$ & $[0.0159]$ &  & $[0.0101]$ & $[0.0173]$ &  \\\hline
 & AR($p$) &  & $1$ & $1$ &  & $1$ & $1$ &  \\
 & $\bar{R}^2$ &  & $0.7179$ & $0.5171$ &  & $0.4289$ & $0.6783$ &  \\\hline\hline
\end{tabular}}
{\begin{minipage}{420pt}\footnotesize
{\underline{Notes:}} As for Table \ref{age_war_tab4}.
\end{minipage}}%
\end{center}
\end{table}

\clearpage

\begin{landscape}
\begin{table}[p]
\caption{Event Study Results  at $t=0$ (PI)}
\label{age_war_tab7}
\begin{center}\resizebox{23cm}{!}{
\begin{tabular}{lp{18cm}cccccccccccccccccc}\hline\hline
 &  &  & \multicolumn{7}{c}{$Zaibatsu$} &  & \multicolumn{7}{c}{$Non\text{--}zaibatsu$} \\\cline{4-19}
 &  &  & \multicolumn{3}{c}{$Military$} &  & \multicolumn{3}{c}{$Non\text{--}military$} &  & \multicolumn{3}{c}{$Military$} &  & \multicolumn{3}{c}{$Non\text{--}military$} \\\cline{4-6}\cline{8-10}\cline{12-14}\cline{16-18}
Event Day & Event Name & Event Type & $AR$ & $CAR$ & $SCAR$ &  & $AR$ & $CAR$ & $SCAR$ &  & $AR$ & $CAR$ & $SCAR$ &  & $AR$ & $CAR$ & $SCAR$\\\hline
1930-10-07 & Life Insurance Securities Corporation was established. & Market & 0 & 0 & 0 &  & 0 & 0 & 0 &  & 0 & {\color{blue}{\large $\boldsymbol{-}$}} & {\color{blue}{\large $\boldsymbol{-}$}} &  & 0 & {\color{red}{\large $\boldsymbol{+}$}} & {\color{red}{\large $\boldsymbol{+}$}} \\
1930-11-14 & An assassination attempt was made on Prime Minister Hamaguchi. & Politics & 0 & 0 & 0 &  & {\color{red}{\large $\boldsymbol{+}$}} & 0 & 0 &  & {\color{blue}{\large $\boldsymbol{-}$}} & {\color{blue}{\large $\boldsymbol{-}$}} & {\color{blue}{\large $\boldsymbol{-}$}} &  & {\color{red}{\large $\boldsymbol{+}$}} & {\color{red}{\large $\boldsymbol{+}$}} & {\color{red}{\large $\boldsymbol{+}$}} \\
1931-03-20 & The March Incident (a coup d'\'etat attempt in March) took place. & Politics & 0 & {\color{red}{\large $\boldsymbol{+}$}} & {\color{red}{\large $\boldsymbol{+}$}} &  & 0 & 0 & 0 &  & 0 & 0 & 0 &  & 0 & 0 & 0 \\
1931-09-18 & Mukden Incident occurred. & War & 0 & 0 & 0 &  & {\color{blue}{\large $\boldsymbol{-}$}} & 0 & 0 &  & 0 & 0 & 0 &  & 0 & 0 & 0 \\
1931-10-21 & The October Incident (a coup d'\'etat attempt in October) took place. & Politics & 0 & 0 & 0 &  & 0 & 0 & 0 &  & 0 & 0 & 0 &  & 0 & 0 & 0 \\
1931-12-13 & The Japanese government reimposed the ban on gold exports. & Market & {\color{red}{\large $\boldsymbol{+}$}} & {\color{red}{\large $\boldsymbol{+}$}} & {\color{red}{\large $\boldsymbol{+}$}} &  & {\color{red}{\large $\boldsymbol{+}$}} & {\color{red}{\large $\boldsymbol{+}$}} & {\color{red}{\large $\boldsymbol{+}$}} &  & {\color{red}{\large $\boldsymbol{+}$}} & 0 & 0 &  & {\color{blue}{\large $\boldsymbol{-}$}} & 0 & 0 \\
1932-01-28 & First Shanghai Incident broke out. & War & 0 & 0 & 0 &  & {\color{blue}{\large $\boldsymbol{-}$}} & 0 & 0 &  & {\color{red}{\large $\boldsymbol{+}$}} & {\color{red}{\large $\boldsymbol{+}$}} & {\color{red}{\large $\boldsymbol{+}$}} &  & {\color{blue}{\large $\boldsymbol{-}$}} & {\color{blue}{\large $\boldsymbol{-}$}} & {\color{blue}{\large $\boldsymbol{-}$}} \\
1932-02-09 & The League of Blood Incident occurred. & Politics & 0 & 0 & 0 &  & 0 & 0 & 0 &  & 0 & {\color{red}{\large $\boldsymbol{+}$}} & {\color{red}{\large $\boldsymbol{+}$}} &  & 0 & 0 & 0 \\
1932-03-01 & Japan proclaimed the establishment of Manchukuo. & Politics & 0 & 0 & 0 &  & 0 & 0 & 0 &  & 0 & 0 & 0 &  & 0 & 0 & 0 \\
1932-05-15 & The May 15 Incident occurred. & Politics & 0 & 0 & 0 &  & 0 & 0 & 0 &  & 0 & 0 & 0 &  & 0 & 0 & 0 \\
1932-11-25 & The Bank of Japan began underwriting government deficit bonds. & Market & 0 & 0 & 0 &  & 0 & 0 & 0 &  & {\color{red}{\large $\boldsymbol{+}$}} & {\color{red}{\large $\boldsymbol{+}$}} & {\color{red}{\large $\boldsymbol{+}$}} &  & {\color{blue}{\large $\boldsymbol{-}$}} & {\color{blue}{\large $\boldsymbol{-}$}} & {\color{blue}{\large $\boldsymbol{-}$}} \\
1933-02-24 & The League of Nations passed a resolution condemning Japan's actions in Manchuria. & War & {\color{blue}{\large $\boldsymbol{-}$}} & 0 & 0 &  & 0 & 0 & 0 &  & 0 & 0 & 0 &  & 0 & 0 & 0 \\
1933-03-29 & The Foreign Exchange Control Law was enacted. & Regulations & 0 & 0 & 0 &  & 0 & 0 & 0 &  & 0 & 0 & 0 &  & 0 & 0 & 0 \\
1933-05-31 & Japan and China signed the Tanggu Truce Agreement. & War & {\color{red}{\large $\boldsymbol{+}$}} & {\color{red}{\large $\boldsymbol{+}$}} & {\color{red}{\large $\boldsymbol{+}$}} &  & 0 & 0 & 0 &  & 0 & 0 & 0 &  & 0 & 0 & 0 \\
1934-11-20 & The Military Academy Incident occurred. & Politics & 0 & {\color{red}{\large $\boldsymbol{+}$}} & {\color{red}{\large $\boldsymbol{+}$}} &  & 0 & 0 & 0 &  & 0 & 0 & 0 &  & 0 & 0 & 0 \\
1934-12-29 & Japan gave notice of its withdrawal from the Washington Naval Treaty. & War & 0 & 0 & 0 &  & 0 & 0 & 0 &  & 0 & {\color{red}{\large $\boldsymbol{+}$}} & {\color{red}{\large $\boldsymbol{+}$}} &  & 0 & {\color{blue}{\large $\boldsymbol{-}$}} & {\color{blue}{\large $\boldsymbol{-}$}} \\
1935-08-12 & The Aizawa Incident occurred. & Politics & 0 & 0 & 0 &  & 0 & 0 & 0 &  & 0 & 0 & 0 &  & 0 & 0 & 0 \\
1936-01-15 & Japan withdrew from the Second London Naval Disarmament Conference. & War & 0 & 0 & 0 &  & 0 & 0 & 0 &  & 0 & 0 & 0 &  & 0 & 0 & 0 \\
1936-02-26 & The February 26 Incident occurred. & Politics & {\color{red}{\large $\boldsymbol{+}$}} & {\color{red}{\large $\boldsymbol{+}$}} & {\color{red}{\large $\boldsymbol{+}$}} &  & 0 & 0 & 0 &  & 0 & 0 & 0 &  & {\color{blue}{\large $\boldsymbol{-}$}} & {\color{blue}{\large $\boldsymbol{-}$}} & {\color{blue}{\large $\boldsymbol{-}$}} \\
1936-05-18 & The Active-Duty Military Ministers System was restored. & Politics & 0 & 0 & 0 &  & 0 & 0 & 0 &  & 0 & 0 & 0 &  & 0 & 0 & 0 \\
1936-08-07 & The government adopted the Fundamentals of National Policy guidelines. & Regulations & 0 & 0 & 0 &  & 0 & 0 & 0 &  & 0 & 0 & 0 &  & 0 & 0 & 0 \\
1936-11-25 & Japan concluded the Anti-Comintern Pact with Germany. & War & 0 & {\color{blue}{\large $\boldsymbol{-}$}} & {\color{blue}{\large $\boldsymbol{-}$}} &  & 0 & 0 & 0 &  & 0 & 0 & 0 &  & 0 & 0 & 0 \\\hdashline
1937-07-07 & The Marco Polo Bridge Incident occurred. & War & 0 & 0 & 0 &  & 0 & 0 & 0 &  & 0 & 0 & 0 &  & 0 & 0 & 0 \\
1937-09-08 & The Temporary Funds Adjustment Law was approved. & Regulations & 0 & 0 & 0 &  & 0 & 0 & 0 &  & 0 & 0 & 0 &  & 0 & 0 & 0 \\
1938-03-24 & National Mobilization Law approved. & Regulations & 0 & 0 & 0 &  & 0 & 0 & 0 &  & 0 & 0 & 0 &  & 0 & 0 & 0 \\
1939-04-01 & Company Profit Dividends and Fund Accommodation Ordinance was promulgated and enforced. & Regulations & 0 & 0 & 0 &  & 0 & 0 & 0 &  & {\color{blue}{\large $\boldsymbol{-}$}} & 0 & 0 &  & {\color{red}{\large $\boldsymbol{+}$}} & 0 & 0 \\
1939-09-01 & World War I\hspace{-1.2pt}I broke out. & War & {\color{red}{\large $\boldsymbol{+}$}} & 0 & 0 &  & {\color{red}{\large $\boldsymbol{+}$}} & 0 & 0 &  & {\color{blue}{\large $\boldsymbol{-}$}} & {\color{blue}{\large $\boldsymbol{-}$}} & {\color{blue}{\large $\boldsymbol{-}$}} &  & {\color{red}{\large $\boldsymbol{+}$}} & {\color{red}{\large $\boldsymbol{+}$}} & {\color{red}{\large $\boldsymbol{+}$}} \\
1940-06-22 & The Second Armistice at Compi\`egne was signed. & War & {\color{blue}{\large $\boldsymbol{-}$}} & {\color{blue}{\large $\boldsymbol{-}$}} & {\color{blue}{\large $\boldsymbol{-}$}} &  & 0 & 0 & 0 &  & 0 & 0 & 0 &  & 0 & {\color{red}{\large $\boldsymbol{+}$}} & {\color{red}{\large $\boldsymbol{+}$}} \\
1940-09-27 & The Tripartite Pact was signed. & War & 0 & 0 & 0 &  & 0 & 0 & 0 &  & 0 & {\color{blue}{\large $\boldsymbol{-}$}} & {\color{blue}{\large $\boldsymbol{-}$}} &  & 0 & {\color{red}{\large $\boldsymbol{+}$}} & {\color{red}{\large $\boldsymbol{+}$}} \\
1940-10-19 & Company Accounting Control Ordinance was promulgated and enforced. & Regulations & 0 & 0 & 0 &  & 0 & 0 & 0 &  & 0 & 0 & 0 &  & 0 & 0 & 0 \\
1941-04-13 & Soviet--Japanese Neutrality Pact was signed. & War & 0 & 0 & 0 &  & 0 & 0 & 0 &  & 0 & 0 & 0 &  & 0 & 0 & 0 \\
1941-07-26 & The U.S. freezing of Japanese assets and subsequent oil embargo. & War & 0 & 0 & 0 &  & 0 & 0 & 0 &  & {\color{red}{\large $\boldsymbol{+}$}} & 0 & 0 &  & {\color{blue}{\large $\boldsymbol{-}$}} & {\color{blue}{\large $\boldsymbol{-}$}} & {\color{blue}{\large $\boldsymbol{-}$}} \\
1941-08-30 & Stock Price Control Ordinance was promulgated and enforced & Regulations & 0 & 0 & 0 &  & 0 & 0 & 0 &  & 0 & 0 & 0 &  & 0 & 0 & 0 \\
1941-12-08 & The Pacific War broke out. & War & 0 & 0 & 0 &  & 0 & 0 & 0 &  & 0 & 0 & 0 &  & {\color{red}{\large $\boldsymbol{+}$}} & {\color{red}{\large $\boldsymbol{+}$}} & {\color{red}{\large $\boldsymbol{+}$}} \\
1942-02-12 & The Wartime Finance Bank Law was approved. & Regulations & 0 & 0 & 0 &  & 0 & 0 & 0 &  & 0 & 0 & 0 &  & 0 & 0 & 0 \\
1942-04-18 & The first air raid on the Japanese mainland. & War & 0 & 0 & 0 &  & 0 & 0 & 0 &  & 0 & 0 & 0 &  & 0 & 0 & 0 \\
1942-06-05 & Battle of Midway broke out. & War & 0 & {\color{blue}{\large $\boldsymbol{-}$}} & {\color{blue}{\large $\boldsymbol{-}$}} &  & 0 & 0 & 0 &  & 0 & 0 & 0 &  & 0 & {\color{red}{\large $\boldsymbol{+}$}} & {\color{red}{\large $\boldsymbol{+}$}} \\
1942-11-15 & The Ministry of Finance pivoted to stricter capital controls. & Market & {\color{blue}{\large $\boldsymbol{-}$}} & 0 & 0 &  & {\color{red}{\large $\boldsymbol{+}$}} & 0 & 0 &  & 0 & 0 & 0 &  & {\color{red}{\large $\boldsymbol{+}$}} & 0 & 0 \\
1943-02-09 & The Fall of Guadalcanal. & War & 0 & 0 & 0 &  & 0 & {\color{red}{\large $\boldsymbol{+}$}} & {\color{red}{\large $\boldsymbol{+}$}} &  & 0 & 0 & 0 &  & 0 & 0 & 0 \\
1943-03-31 & The Tokyo Stock Exchange delisted short-term clearing futures transactions. & Market & 0 & 0 & 0 &  & 0 & 0 & 0 &  & 0 & 0 & 0 &  & 0 & 0 & 0 \\
1943-06-30 & The Japanese Stock Exchange was founded. & Market & 0 & 0 & 0 &  & 0 & 0 & 0 &  & 0 & 0 & 0 &  & 0 & 0 & 0 \\\hline\hline
\end{tabular}}
{\resizebox{23cm}{!}{\begin{minipage}{800pt}\scriptsize
\vspace*{5pt}
{\underline{Notes:}}
 \begin{itemize}
  \item[(1)] The symbols ``{\color{red}{$\boldsymbol{+}$}}'', ``{\color{blue}{$\boldsymbol{-}$}}'', and ``$0$'' denote statistically significant positive, significant negative, and insignificant AR, CAR, or SCAR at time $t=0$, respectively.
  \item[(2)] R version 4.5.3 was used to compute the estimates.
 \end{itemize}
\end{minipage}}}%
\end{center}
\end{table}
\end{landscape}

\clearpage

\begin{landscape}
\begin{table}[p]
\caption{Event Study Results at $t=0$ (API)}
\label{age_war_tab8}
\begin{center}\resizebox{23cm}{!}{
\begin{tabular}{lp{18cm}cccccccccccccccccc}\hline\hline
 &  &  & \multicolumn{7}{c}{$Zaibatsu$} &  & \multicolumn{7}{c}{$Non\text{--}zaibatsu$} \\\cline{4-19}
 &  &  & \multicolumn{3}{c}{$Military$} &  & \multicolumn{3}{c}{$Non\text{--}military$} &  & \multicolumn{3}{c}{$Military$} &  & \multicolumn{3}{c}{$Non\text{--}military$} \\\cline{4-6}\cline{8-10}\cline{12-14}\cline{16-18}
Event Day & Event Name & Event Type & $AR$ & $CAR$ & $SCAR$ &  & $AR$ & $CAR$ & $SCAR$ &  & $AR$ & $CAR$ & $SCAR$ &  & $AR$ & $CAR$ & $SCAR$\\\hline
1930-10-07 & Life Insurance Securities Corporation was established. & Market & 0 & 0 & 0 &  & 0 & 0 & 0 &  & 0 & {\color{blue}{\large $\boldsymbol{-}$}} & {\color{blue}{\large $\boldsymbol{-}$}} &  & 0 & {\color{red}{\large $\boldsymbol{+}$}} & {\color{red}{\large $\boldsymbol{+}$}} \\
1930-11-14 & An assassination attempt was made on Prime Minister Hamaguchi. & Politics & 0 & 0 & 0 &  & {\color{red}{\large $\boldsymbol{+}$}} & 0 & 0 &  & {\color{blue}{\large $\boldsymbol{-}$}} & {\color{blue}{\large $\boldsymbol{-}$}} & {\color{blue}{\large $\boldsymbol{-}$}} &  & {\color{red}{\large $\boldsymbol{+}$}} & {\color{red}{\large $\boldsymbol{+}$}} & {\color{red}{\large $\boldsymbol{+}$}} \\
1931-03-20 & The March Incident (a coup d'\'etat attempt in March) took place. & Politics & 0 & {\color{red}{\large $\boldsymbol{+}$}} & {\color{red}{\large $\boldsymbol{+}$}} &  & 0 & 0 & 0 &  & 0 & 0 & 0 &  & 0 & 0 & 0 \\
1931-09-18 & Mukden Incident occurred. & War & 0 & 0 & 0 &  & {\color{blue}{\large $\boldsymbol{-}$}} & 0 & 0 &  & 0 & 0 & 0 &  & 0 & 0 & 0 \\
1931-10-21 & The October Incident (a coup d'\'etat attempt in October) took place. & Politics & 0 & 0 & 0 &  & 0 & 0 & 0 &  & 0 & 0 & 0 &  & 0 & 0 & 0 \\
1931-12-13 & The Japanese government reimposed the ban on gold exports. & Market & {\color{red}{\large $\boldsymbol{+}$}} & {\color{red}{\large $\boldsymbol{+}$}} & {\color{red}{\large $\boldsymbol{+}$}} &  & {\color{red}{\large $\boldsymbol{+}$}} & {\color{red}{\large $\boldsymbol{+}$}} & {\color{red}{\large $\boldsymbol{+}$}} &  & {\color{red}{\large $\boldsymbol{+}$}} & 0 & 0 &  & {\color{blue}{\large $\boldsymbol{-}$}} & 0 & 0 \\
1932-01-28 & First Shanghai Incident broke out. & War & 0 & 0 & 0 &  & {\color{blue}{\large $\boldsymbol{-}$}} & 0 & 0 &  & {\color{red}{\large $\boldsymbol{+}$}} & {\color{red}{\large $\boldsymbol{+}$}} & {\color{red}{\large $\boldsymbol{+}$}} &  & {\color{blue}{\large $\boldsymbol{-}$}} & {\color{blue}{\large $\boldsymbol{-}$}} & {\color{blue}{\large $\boldsymbol{-}$}} \\
1932-02-09 & The League of Blood Incident occurred. & Politics & 0 & 0 & 0 &  & 0 & 0 & 0 &  & 0 & {\color{red}{\large $\boldsymbol{+}$}} & {\color{red}{\large $\boldsymbol{+}$}} &  & 0 & 0 & 0 \\
1932-03-01 & Japan proclaimed the establishment of Manchukuo. & Politics & 0 & 0 & 0 &  & 0 & 0 & 0 &  & 0 & 0 & 0 &  & 0 & 0 & 0 \\
1932-05-15 & The May 15 Incident occurred. & Politics & 0 & 0 & 0 &  & 0 & 0 & 0 &  & 0 & 0 & 0 &  & 0 & 0 & 0 \\
1932-11-25 & The Bank of Japan began underwriting government deficit bonds. & Market & 0 & 0 & 0 &  & 0 & 0 & 0 &  & {\color{red}{\large $\boldsymbol{+}$}} & {\color{red}{\large $\boldsymbol{+}$}} & {\color{red}{\large $\boldsymbol{+}$}} &  & {\color{blue}{\large $\boldsymbol{-}$}} & {\color{blue}{\large $\boldsymbol{-}$}} & {\color{blue}{\large $\boldsymbol{-}$}} \\
1933-02-24 & The League of Nations passed a resolution condemning Japan's actions in Manchuria. & War & {\color{blue}{\large $\boldsymbol{-}$}} & 0 & 0 &  & 0 & 0 & 0 &  & 0 & 0 & 0 &  & 0 & 0 & 0 \\
1933-03-29 & The Foreign Exchange Control Law was enacted. & Regulations & 0 & 0 & 0 &  & 0 & 0 & 0 &  & 0 & 0 & 0 &  & 0 & 0 & 0 \\
1933-05-31 & Japan and China signed the Tanggu Truce Agreement. & War & {\color{red}{\large $\boldsymbol{+}$}} & {\color{red}{\large $\boldsymbol{+}$}} & {\color{red}{\large $\boldsymbol{+}$}} &  & 0 & 0 & 0 &  & 0 & 0 & 0 &  & 0 & 0 & 0 \\
1934-11-20 & The Military Academy Incident occurred. & Politics & 0 & {\color{red}{\large $\boldsymbol{+}$}} & {\color{red}{\large $\boldsymbol{+}$}} &  & 0 & 0 & 0 &  & 0 & 0 & 0 &  & 0 & 0 & 0 \\
1934-12-29 & Japan gave notice of its withdrawal from the Washington Naval Treaty. & War & 0 & 0 & 0 &  & {\color{blue}{\large $\boldsymbol{-}$}} & 0 & 0 &  & {\color{red}{\large $\boldsymbol{+}$}} & {\color{red}{\large $\boldsymbol{+}$}} & {\color{red}{\large $\boldsymbol{+}$}} &  & 0 & {\color{blue}{\large $\boldsymbol{-}$}} & {\color{blue}{\large $\boldsymbol{-}$}} \\
1935-08-12 & The Aizawa Incident occurred. & Politics & 0 & 0 & 0 &  & 0 & 0 & 0 &  & 0 & 0 & 0 &  & 0 & 0 & 0 \\
1936-01-15 & Japan withdrew from the Second London Naval Disarmament Conference. & War & 0 & 0 & 0 &  & 0 & 0 & 0 &  & 0 & 0 & 0 &  & 0 & 0 & 0 \\
1936-02-26 & The February 26 Incident occurred. & Politics & {\color{red}{\large $\boldsymbol{+}$}} & {\color{red}{\large $\boldsymbol{+}$}} & {\color{red}{\large $\boldsymbol{+}$}} &  & {\color{red}{\large $\boldsymbol{+}$}} & 0 & 0 &  & {\color{blue}{\large $\boldsymbol{-}$}} & 0 & 0 &  & {\color{blue}{\large $\boldsymbol{-}$}} & {\color{blue}{\large $\boldsymbol{-}$}} & {\color{blue}{\large $\boldsymbol{-}$}} \\
1936-05-18 & The Active-Duty Military Ministers System was restored. & Politics & 0 & 0 & 0 &  & 0 & 0 & 0 &  & 0 & 0 & 0 &  & 0 & 0 & 0 \\
1936-08-07 & The government adopted the Fundamentals of National Policy guidelines. & Regulations & 0 & 0 & 0 &  & 0 & 0 & 0 &  & 0 & 0 & 0 &  & 0 & 0 & 0 \\
1936-11-25 & Japan concluded the Anti-Comintern Pact with Germany. & War & 0 & {\color{blue}{\large $\boldsymbol{-}$}} & {\color{blue}{\large $\boldsymbol{-}$}} &  & 0 & 0 & 0 &  & 0 & 0 & 0 &  & 0 & 0 & 0 \\\hdashline
1937-07-07 & Marco Polo Bridge Incident occured. & War & 0 & 0 & 0 &  & 0 & 0 & 0 &  & 0 & 0 & 0 &  & 0 & 0 & 0 \\
1937-09-08 & The Temporary Fund Adjustment Law approved. & Regulations & 0 & 0 & 0 &  & 0 & 0 & 0 &  & 0 & 0 & 0 &  & 0 & 0 & 0 \\
1938-03-24 & National Mobilization Law approved. & Regulations & 0 & 0 & 0 &  & 0 & 0 & 0 &  & 0 & 0 & 0 &  & 0 & 0 & 0 \\
1939-04-01 & Company Profit Dividends and Fund Accommodation Ordinance was promulgated and enforced. & Regulations & 0 & 0 & 0 &  & 0 & 0 & 0 &  & {\color{blue}{\large $\boldsymbol{-}$}} & 0 & 0 &  & {\color{red}{\large $\boldsymbol{+}$}} & 0 & 0 \\
1939-09-01 & World War I\hspace{-1.2pt}I broke out. & War & {\color{red}{\large $\boldsymbol{+}$}} & 0 & 0 &  & {\color{red}{\large $\boldsymbol{+}$}} & 0 & 0 &  & {\color{blue}{\large $\boldsymbol{-}$}} & {\color{blue}{\large $\boldsymbol{-}$}} & {\color{blue}{\large $\boldsymbol{-}$}} &  & {\color{red}{\large $\boldsymbol{+}$}} & {\color{red}{\large $\boldsymbol{+}$}} & {\color{red}{\large $\boldsymbol{+}$}} \\
1940-06-22 & The Second Armistice at Compiegne was signed. & War & {\color{blue}{\large $\boldsymbol{-}$}} & {\color{blue}{\large $\boldsymbol{-}$}} & {\color{blue}{\large $\boldsymbol{-}$}} &  & 0 & 0 & 0 &  & 0 & 0 & 0 &  & 0 & {\color{red}{\large $\boldsymbol{+}$}} & {\color{red}{\large $\boldsymbol{+}$}} \\
1940-09-27 & Tripartite Pact was signed & War & 0 & 0 & 0 &  & 0 & 0 & 0 &  & 0 & {\color{blue}{\large $\boldsymbol{-}$}} & {\color{blue}{\large $\boldsymbol{-}$}} &  & 0 & 0 & 0 \\
1940-10-19 & Company Accounting Control Ordinance was promulgated and enforced. & Regulations & 0 & 0 & 0 &  & 0 & 0 & 0 &  & 0 & 0 & 0 &  & 0 & 0 & 0 \\
1941-04-13 & Soviet--Japanese Neutrality Pact was signed. & War & 0 & 0 & 0 &  & 0 & 0 & 0 &  & 0 & 0 & 0 &  & 0 & 0 & 0 \\
1941-07-26 & The U.S. freezing of Japanese assets and subsequent oil embargo. & War & 0 & 0 & 0 &  & 0 & 0 & 0 &  & {\color{red}{\large $\boldsymbol{+}$}} & 0 & 0 &  & {\color{blue}{\large $\boldsymbol{-}$}} & {\color{blue}{\large $\boldsymbol{-}$}} & {\color{blue}{\large $\boldsymbol{-}$}} \\
1941-08-30 & Stock Price Control Ordinance was promulgated and enforced & Regulations & 0 & 0 & 0 &  & 0 & 0 & 0 &  & 0 & 0 & 0 &  & 0 & 0 & 0 \\
1941-12-08 & The Pacific War broke out. & War & 0 & 0 & 0 &  & 0 & 0 & 0 &  & 0 & 0 & 0 &  & {\color{red}{\large $\boldsymbol{+}$}} & {\color{red}{\large $\boldsymbol{+}$}} & {\color{red}{\large $\boldsymbol{+}$}} \\
1942-02-12 & Wartime Finance Bank Law aprroved. & Regulations & 0 & 0 & 0 &  & 0 & 0 & 0 &  & 0 & 0 & 0 &  & 0 & 0 & 0 \\
1942-04-18 & The first air raid on the Japanese mainland. & War & 0 & 0 & 0 &  & 0 & 0 & 0 &  & 0 & 0 & 0 &  & 0 & 0 & 0 \\
1942-06-05 & Battle of Midway broke out. & War & 0 & 0 & 0 &  & 0 & 0 & 0 &  & 0 & 0 & 0 &  & 0 & 0 & 0 \\
1942-11-15 & The Ministry of Finance pivoted to stricter capital controls. & Market & {\color{blue}{\large $\boldsymbol{-}$}} & 0 & 0 &  & {\color{red}{\large $\boldsymbol{+}$}} & 0 & 0 &  & 0 & 0 & 0 &  & {\color{red}{\large $\boldsymbol{+}$}} & 0 & 0 \\
1943-02-09 & The Fall of Guadalcanal. & War & 0 & 0 & 0 &  & 0 & {\color{red}{\large $\boldsymbol{+}$}} & {\color{red}{\large $\boldsymbol{+}$}} &  & 0 & 0 & 0 &  & 0 & 0 & 0 \\
1943-03-31 & The Tokyo Stock Exchange delisted short-term clearing futures transactions. & Market & 0 & 0 & 0 &  & 0 & 0 & 0 &  & 0 & 0 & 0 &  & 0 & 0 & 0 \\
1943-06-30 & The Japanese Stock Exchange was founded. & Market & 0 & 0 & 0 &  & 0 & 0 & 0 &  & 0 & 0 & 0 &  & 0 & 0 & 0 \\\hline\hline
\end{tabular}}
{\resizebox{23cm}{!}{\begin{minipage}{800pt}\scriptsize
\vspace*{5pt}
{\underline{Notes:}} As for Table \ref{age_war_tab7}.
\end{minipage}}}%
\end{center}
\end{table}
\end{landscape}

\clearpage

\begin{landscape}
\begin{table}[p]
\caption{Event Study Results at $t=0$ (TRI)}
\label{age_war_tab9}
\begin{center}\resizebox{23cm}{!}{
\begin{tabular}{lp{18cm}cccccccccccccccccc}\hline\hline
 &  &  & \multicolumn{7}{c}{$Zaibatsu$} &  & \multicolumn{7}{c}{$Non\text{--}zaibatsu$} \\\cline{4-19}
 &  &  & \multicolumn{3}{c}{$Military$} &  & \multicolumn{3}{c}{$Non\text{--}military$} &  & \multicolumn{3}{c}{$Military$} &  & \multicolumn{3}{c}{$Non\text{--}military$} \\\cline{4-6}\cline{8-10}\cline{12-14}\cline{16-18}
Event Day & Event Name & Event Type & $AR$ & $CAR$ & $SCAR$ &  & $AR$ & $CAR$ & $SCAR$ &  & $AR$ & $CAR$ & $SCAR$ &  & $AR$ & $CAR$ & $SCAR$\\\hline
1930-10-07 & Life Insurance Securities Corporation was established. & Market & 0 & 0 & 0 &  & 0 & 0 & 0 &  & 0 & {\color{blue}{\large $\boldsymbol{-}$}} & {\color{blue}{\large $\boldsymbol{-}$}} &  & 0 & 0 & 0 \\
1930-11-14 & An assassination attempt was made on Prime Minister Hamaguchi. & Politics & 0 & 0 & 0 &  & {\color{red}{\large $\boldsymbol{+}$}} & 0 & 0 &  & {\color{blue}{\large $\boldsymbol{-}$}} & {\color{blue}{\large $\boldsymbol{-}$}} & {\color{blue}{\large $\boldsymbol{-}$}} &  & {\color{red}{\large $\boldsymbol{+}$}} & {\color{red}{\large $\boldsymbol{+}$}} & {\color{red}{\large $\boldsymbol{+}$}} \\
1931-03-20 & The March Incident (a coup d'\'etat attempt in March) took place. & Politics & 0 & {\color{red}{\large $\boldsymbol{+}$}} & {\color{red}{\large $\boldsymbol{+}$}} &  & 0 & 0 & 0 &  & 0 & 0 & 0 &  & 0 & 0 & 0 \\
1931-09-18 & Mukden Incident occurred. & War & 0 & 0 & 0 &  & {\color{blue}{\large $\boldsymbol{-}$}} & 0 & 0 &  & 0 & 0 & 0 &  & 0 & 0 & 0 \\
1931-10-21 & The October Incident (a coup d'\'etat attempt in October) took place. & Politics & 0 & 0 & 0 &  & 0 & 0 & 0 &  & 0 & 0 & 0 &  & 0 & 0 & 0 \\
1931-12-13 & The Japanese government reimposed the ban on gold exports. & Market & {\color{red}{\large $\boldsymbol{+}$}} & {\color{red}{\large $\boldsymbol{+}$}} & {\color{red}{\large $\boldsymbol{+}$}} &  & 0 & {\color{red}{\large $\boldsymbol{+}$}} & {\color{red}{\large $\boldsymbol{+}$}} &  & {\color{red}{\large $\boldsymbol{+}$}} & 0 & 0 &  & {\color{blue}{\large $\boldsymbol{-}$}} & 0 & 0 \\
1932-01-28 & First Shanghai Incident broke out. & War & 0 & 0 & 0 &  & {\color{blue}{\large $\boldsymbol{-}$}} & 0 & 0 &  & {\color{red}{\large $\boldsymbol{+}$}} & {\color{red}{\large $\boldsymbol{+}$}} & {\color{red}{\large $\boldsymbol{+}$}} &  & {\color{blue}{\large $\boldsymbol{-}$}} & {\color{blue}{\large $\boldsymbol{-}$}} & {\color{blue}{\large $\boldsymbol{-}$}} \\
1932-02-09 & The League of Blood Incident occurred. & Politics & 0 & 0 & 0 &  & {\color{blue}{\large $\boldsymbol{-}$}} & 0 & 0 &  & 0 & 0 & 0 &  & 0 & 0 & 0 \\
1932-03-01 & Japan proclaimed the establishment of Manchukuo. & Politics & 0 & 0 & 0 &  & {\color{blue}{\large $\boldsymbol{-}$}} & 0 & 0 &  & 0 & 0 & 0 &  & 0 & 0 & 0 \\
1932-05-15 & The May 15 Incident occurred. & Politics & 0 & 0 & 0 &  & {\color{blue}{\large $\boldsymbol{-}$}} & 0 & 0 &  & 0 & 0 & 0 &  & 0 & 0 & 0 \\
1932-11-25 & The Bank of Japan began underwriting government deficit bonds. & Market & 0 & 0 & 0 &  & 0 & 0 & 0 &  & {\color{red}{\large $\boldsymbol{+}$}} & {\color{red}{\large $\boldsymbol{+}$}} & {\color{red}{\large $\boldsymbol{+}$}} &  & {\color{blue}{\large $\boldsymbol{-}$}} & {\color{blue}{\large $\boldsymbol{-}$}} & {\color{blue}{\large $\boldsymbol{-}$}} \\
1933-02-24 & The League of Nations passed a resolution condemning Japan's actions in Manchuria. & War & {\color{blue}{\large $\boldsymbol{-}$}} & 0 & 0 &  & 0 & 0 & 0 &  & 0 & 0 & 0 &  & 0 & 0 & 0 \\
1933-03-29 & The Foreign Exchange Control Law was enacted. & Regulations & 0 & 0 & 0 &  & {\color{blue}{\large $\boldsymbol{-}$}} & 0 & 0 &  & 0 & 0 & 0 &  & 0 & 0 & 0 \\
1933-05-31 & Japan and China signed the Tanggu Truce Agreement. & War & {\color{red}{\large $\boldsymbol{+}$}} & {\color{red}{\large $\boldsymbol{+}$}} & {\color{red}{\large $\boldsymbol{+}$}} &  & 0 & 0 & 0 &  & 0 & 0 & 0 &  & 0 & 0 & 0 \\
1934-11-20 & The Military Academy Incident occurred. & Politics & 0 & {\color{red}{\large $\boldsymbol{+}$}} & {\color{red}{\large $\boldsymbol{+}$}} &  & 0 & 0 & 0 &  & 0 & 0 & 0 &  & 0 & 0 & 0 \\
1934-12-29 & Japan gave notice of its withdrawal from the Washington Naval Treaty. & War & 0 & 0 & 0 &  & {\color{blue}{\large $\boldsymbol{-}$}} & 0 & 0 &  & 0 & 0 & 0 &  & 0 & 0 & 0 \\
1935-08-12 & The Aizawa Incident occurred. & Politics & 0 & 0 & 0 &  & {\color{blue}{\large $\boldsymbol{-}$}} & 0 & 0 &  & 0 & 0 & 0 &  & 0 & 0 & 0 \\
1936-01-15 & Japan withdrew from the Second London Naval Disarmament Conference. & War & 0 & 0 & 0 &  & {\color{blue}{\large $\boldsymbol{-}$}} & 0 & 0 &  & 0 & 0 & 0 &  & 0 & 0 & 0 \\
1936-02-26 & The February 26 Incident occurred. & Politics & {\color{red}{\large $\boldsymbol{+}$}} & {\color{red}{\large $\boldsymbol{+}$}} & {\color{red}{\large $\boldsymbol{+}$}} &  & 0 & 0 & 0 &  & 0 & 0 & 0 &  & {\color{blue}{\large $\boldsymbol{-}$}} & {\color{blue}{\large $\boldsymbol{-}$}} & {\color{blue}{\large $\boldsymbol{-}$}} \\
1936-05-18 & The Active-Duty Military Ministers System was restored. & Politics & 0 & 0 & 0 &  & {\color{blue}{\large $\boldsymbol{-}$}} & 0 & 0 &  & 0 & 0 & 0 &  & 0 & 0 & 0 \\
1936-08-07 & The government adopted the Fundamentals of National Policy guidelines. & Regulations & 0 & 0 & 0 &  & {\color{blue}{\large $\boldsymbol{-}$}} & 0 & 0 &  & 0 & 0 & 0 &  & 0 & 0 & 0 \\
1936-11-25 & Japan concluded the Anti-Comintern Pact with Germany. & War & 0 & {\color{blue}{\large $\boldsymbol{-}$}} & {\color{blue}{\large $\boldsymbol{-}$}} &  & 0 & 0 & 0 &  & 0 & 0 & 0 &  & 0 & 0 & 0 \\\hdashline
1937-07-07 & Marco Polo Bridge Incident occured. & War & 0 & 0 & 0 &  & 0 & 0 & 0 &  & 0 & 0 & 0 &  & 0 & 0 & 0 \\
1937-09-08 & The Temporary Fund Adjustment Law approved. & Regulations & 0 & 0 & 0 &  & 0 & 0 & 0 &  & 0 & 0 & 0 &  & 0 & 0 & 0 \\
1938-03-24 & National Mobilization Law approved. & Regulations & 0 & 0 & 0 &  & 0 & 0 & 0 &  & 0 & 0 & 0 &  & 0 & 0 & 0 \\
1939-04-01 & Company Profit Dividends and Fund Accommodation Ordinance was promulgated and enforced. & Regulations & 0 & 0 & 0 &  & 0 & 0 & 0 &  & {\color{blue}{\large $\boldsymbol{-}$}} & 0 & 0 &  & {\color{red}{\large $\boldsymbol{+}$}} & 0 & 0 \\
1939-09-01 & World War I\hspace{-1.2pt}I broke out. & War & {\color{red}{\large $\boldsymbol{+}$}} & 0 & 0 &  & {\color{red}{\large $\boldsymbol{+}$}} & {\color{red}{\large $\boldsymbol{+}$}} & {\color{red}{\large $\boldsymbol{+}$}} &  & {\color{blue}{\large $\boldsymbol{-}$}} & {\color{blue}{\large $\boldsymbol{-}$}} & {\color{blue}{\large $\boldsymbol{-}$}} &  & {\color{red}{\large $\boldsymbol{+}$}} & {\color{red}{\large $\boldsymbol{+}$}} & {\color{red}{\large $\boldsymbol{+}$}} \\
1940-06-22 & The Second Armistice at Compiegne was signed. & War & {\color{blue}{\large $\boldsymbol{-}$}} & {\color{blue}{\large $\boldsymbol{-}$}} & {\color{blue}{\large $\boldsymbol{-}$}} &  & 0 & 0 & 0 &  & 0 & 0 & 0 &  & 0 & {\color{red}{\large $\boldsymbol{+}$}} & {\color{red}{\large $\boldsymbol{+}$}} \\
1940-09-27 & Tripartite Pact was signed & War & 0 & {\color{red}{\large $\boldsymbol{+}$}} & {\color{red}{\large $\boldsymbol{+}$}} &  & 0 & 0 & 0 &  & 0 & {\color{blue}{\large $\boldsymbol{-}$}} & {\color{blue}{\large $\boldsymbol{-}$}} &  & 0 & 0 & 0 \\
1940-10-19 & Company Accounting Control Ordinance was promulgated and enforced. & Regulations & 0 & 0 & 0 &  & 0 & 0 & 0 &  & 0 & 0 & 0 &  & 0 & 0 & 0 \\
1941-04-13 & Soviet--Japanese Neutrality Pact was signed. & War & 0 & 0 & 0 &  & 0 & 0 & 0 &  & 0 & 0 & 0 &  & 0 & 0 & 0 \\
1941-07-26 & The U.S. freezing of Japanese assets and subsequent oil embargo. & War & 0 & 0 & 0 &  & 0 & 0 & 0 &  & {\color{red}{\large $\boldsymbol{+}$}} & 0 & 0 &  & {\color{blue}{\large $\boldsymbol{-}$}} & {\color{blue}{\large $\boldsymbol{-}$}} & {\color{blue}{\large $\boldsymbol{-}$}} \\
1941-08-30 & Stock Price Control Ordinance was promulgated and enforced & Regulations & 0 & 0 & 0 &  & 0 & 0 & 0 &  & 0 & 0 & 0 &  & 0 & 0 & 0 \\
1941-12-08 & The Pacific War broke out. & War & 0 & 0 & 0 &  & 0 & 0 & 0 &  & 0 & 0 & 0 &  & {\color{red}{\large $\boldsymbol{+}$}} & {\color{red}{\large $\boldsymbol{+}$}} & {\color{red}{\large $\boldsymbol{+}$}} \\
1942-02-12 & Wartime Finance Bank Law aprroved. & Regulations & 0 & 0 & 0 &  & 0 & 0 & 0 &  & 0 & 0 & 0 &  & 0 & 0 & 0 \\
1942-04-18 & The first air raid on the Japanese mainland. & War & 0 & 0 & 0 &  & 0 & 0 & 0 &  & 0 & 0 & 0 &  & 0 & 0 & 0 \\
1942-06-05 & Battle of Midway broke out. & War & 0 & 0 & 0 &  & 0 & 0 & 0 &  & 0 & {\color{blue}{\large $\boldsymbol{-}$}} & {\color{blue}{\large $\boldsymbol{-}$}} &  & 0 & 0 & 0 \\
1942-11-15 & The Ministry of Finance pivoted to stricter capital controls. & Market & {\color{blue}{\large $\boldsymbol{-}$}} & 0 & 0 &  & {\color{red}{\large $\boldsymbol{+}$}} & 0 & 0 &  & 0 & 0 & 0 &  & {\color{red}{\large $\boldsymbol{+}$}} & 0 & 0 \\
1943-02-09 & The Fall of Guadalcanal. & War & 0 & 0 & 0 &  & 0 & {\color{red}{\large $\boldsymbol{+}$}} & {\color{red}{\large $\boldsymbol{+}$}} &  & 0 & 0 & 0 &  & 0 & 0 & 0 \\
1943-03-31 & The Tokyo Stock Exchange delisted short-term clearing futures transactions. & Market & 0 & 0 & 0 &  & 0 & 0 & 0 &  & 0 & 0 & 0 &  & 0 & 0 & 0 \\
1943-06-30 & The Japanese Stock Exchange was founded. & Market & 0 & 0 & 0 &  & 0 & 0 & 0 &  & 0 & 0 & 0 &  & 0 & 0 & 0 \\\hline\hline
\end{tabular}}
{\resizebox{23cm}{!}{\begin{minipage}{800pt}\scriptsize
\vspace*{5pt}
{\underline{Notes:}} As for Table \ref{age_war_tab7}.
\end{minipage}}}%
\end{center}
\end{table}
\end{landscape}